\documentclass[]{article}
\usepackage [english] {babel}
\usepackage{amsmath,amsfonts,amssymb,amsthm,epsfig,epstopdf,titling,array}
\usepackage [utf8]{inputenc}
\usepackage[T1]{fontenc}
\usepackage{csquotes}
\usepackage{textcomp}
\usepackage{geometry}
\usepackage{bbm}
\usepackage{bm}
\usepackage[margin=10pt, font=scriptsize, labelfont=bf]{caption}
\usepackage{subcaption}
\usepackage{braket}
\usepackage{environ}
\usepackage{verbatim}
\usepackage{dsfont}
\usepackage{tikz}
\usetikzlibrary{shapes}
\usetikzlibrary{positioning}
\usepackage[colorinlistoftodos]{todonotes}
\usepackage{newunicodechar}
\usepackage{authblk}
\usepackage[colorlinks]{hyperref}
\usepackage{caption}
\usepackage{subcaption}
\usepackage{graphicx}
\usepackage{wrapfig}
\usetikzlibrary{plotmarks}
\usetikzlibrary{trees}
\usetikzlibrary{calc}
\usetikzlibrary{decorations.pathmorphing, decorations.pathreplacing}
\usetikzlibrary{decorations.markings}
\usetikzlibrary{backgrounds,automata}
\usepackage[style=alphabetic,url=false,doi=false,backend=biber,maxbibnames=20,sortcites=true,sorting=ynt]{biblatex}
\AtEveryBibitem{\clearfield{issue} \clearfield{number} \clearfield{issn}}
\renewbibmacro{in:}{%
  \ifentrytype{article}{}{\printtext{\bibstring{in}\intitlepunct}}}
\DeclareFieldFormat
  [article,inbook,incollection,inproceedings,patent,thesis,unpublished]
  {titlecase:title}{\MakeSentenceCase*{#1}}
\usepackage[colorlinks]{hyperref}
\usepackage[capitalize]{cleveref}

\usetikzlibrary{trees}
\usetikzlibrary{calc}
\usetikzlibrary{decorations.pathmorphing}
\usetikzlibrary{decorations.markings}
\usetikzlibrary{shapes}
\usetikzlibrary{fit}

\makeatletter
\newsavebox{\measure@tikzpicture}
\NewEnviron{scaletikz}[1]{%
  \def\tikz@width{#1}%
  \begin{lrbox}{\measure@tikzpicture}%
  \BODY
  \end{lrbox}%
  \pgfmathparse{#1/\wd\measure@tikzpicture}%
  \BODY
}
\makeatother

\theoremstyle{plain}

\pgfdeclareplotmark{triangle*}
{%
  \pgfpathmoveto{\pgfqpoint{0pt}{\pgfplotmarksize}}%
   \pgfpathlineto{\pgfqpoint{.85\pgfplotmarksize}{-.85pt}}%
  \pgfpathlineto{\pgfqpoint{-.85\pgfplotmarksize}{-.85pt}}%
  \pgfpathclose%
  \pgfusepathqfillstroke%
}

\newtheorem{theorem}{Theorem}
\numberwithin{theorem}{section}

\newtheorem{proposition}[theorem]{Proposition}

\newtheorem{remark}[theorem]{Remark}

\numberwithin{equation}{section}
\numberwithin{equation}{subsection}
\renewcommand*{\theequation}{%
  \ifnum\value{subsection}=0 %
    \thesection
  \else
    \thesubsection
  \fi
  .\arabic{equation}%
}

\theoremstyle{definition}

\definecolor{darkGreen}{rgb}{0,0.5,0}
\hypersetup{
    colorlinks,
    citecolor=darkGreen,
    filecolor=red,
    linkcolor=blue}

\DeclareMathOperator{\dist}{dist}

\newcommand{\s}{\sigma}

\newcommand{\ul}[1]{{\ensuremath{\underline{#1}}}}

\newcommand{\B}{\ensuremath{\textup{B}} }
\newcommand{\C}{\ensuremath{\textup{C}} }
\newcommand{\D}{\ensuremath{\textup{D}} }
\newcommand{\E}{\ensuremath{\textup{E}} }
\newcommand{\M}{\ensuremath{\textup{M}} }
\renewcommand{\P}{\ensuremath{\textup{P}} }
\newcommand{\R}{\ensuremath{\textup{R}} }

\newcommand{\qf}{\ensuremath{\textup{qf}}}
\newcommand{\Bs}{\ensuremath{\textup{spin}}}
\newcommand{\spin}{\Bs}

\newcommand{\f}{\ensuremath{\textup{f}}}
\newcommand{\p}{\ensuremath{\textup{p}}}

\newcommand{\cB}{\ensuremath{\mathcal{B}}}

\newcommand{\cD}{\ensuremath{\mathcal{D}}}

\newcommand{\cL}{\ensuremath{\mathcal{L}}}

\newcommand{\cP}{\ensuremath{\mathcal{P}}}

\newcommand{\cR}{\ensuremath{\mathcal{R}}}
\newcommand{\cS}{\ensuremath{\mathcal{S}}}
\newcommand{\cT}{\ensuremath{\mathcal{T}}}
\newcommand{\cV}{\ensuremath{\mathcal{V}}}

\newcommand{\cW}{\ensuremath{\mathcal{W}}}
\newcommand{\cM}{\ensuremath{\mathcal{M}}}

\newcommand{\tcR}{\ensuremath{\widetilde \cR}}

\newcommand{\fG}{\ensuremath{\mathfrak{G}}}

\newcommand{\fg}{\ensuremath{\mathfrak{g}}}

\newcommand{\bs}[1]{{\ensuremath{\boldsymbol{#1}}}}

\newcommand{\piecewise}[1]{\left\{ \begin{array}{*2{>{\displaystyle}l}}  #1 \end{array} \right. }

\DeclareMathOperator{\Pf}{Pf}

\newcommand{\media}[1]{ { \left\langle #1 \right\rangle}}

\tikzset{vertex/.style={circle,fill=black,inner sep=2pt},
ctVertex/.style={diamond,fill=black,inner sep=2pt},
bigvertex/.style={circle,fill=black,inner sep=3pt},
E/.append style={fill=white,draw},
G/.append style={fill=gray,draw},
probeEP/.style={circle,fill=black,draw,inner sep=2pt,
  prefix after command= {\pgfextra{\tikzset{every pin/.style = {pin edge={decorate,decoration={snake,amplitude=2pt,segment length =4pt}}}}}}
},
SpinEP/.style={rectangle,fill=gray,draw,inner sep=3pt,
  prefix after command= {\pgfextra{\tikzset{every pin/.style = {pin edge={decorate,decoration={snake,amplitude=2pt,segment length =4pt}}}}}}
},
bareProbeEP/.style={rectangle,fill=black,draw,inner sep=3pt,
  prefix after command= {\pgfextra{\tikzset{every pin/.style = {pin edge={decorate,decoration={snake,amplitude=2pt,segment length =4pt}}}}}}
},
FSSpinEP/.style={regular polygon, regular polygon sides=3,fill=white,draw,inner sep=1.5pt,
  prefix after command= {\pgfextra{\tikzset{every pin/.style = {pin edge={decorate,decoration={snake,amplitude=2pt,segment length =4pt}}}}}}
},
TSpinEP/.style={regular polygon, regular polygon sides=3,fill=gray,draw,inner sep=1.8pt,
  prefix after command= {\pgfextra{\tikzset{every pin/.style = {pin edge={decorate,decoration={snake,amplitude=2pt,segment length =4pt}}}}}}
},
baseline=(current  bounding  box.center),doubled/.style={double distance= 1pt,line width=1.5pt}
}
\pgfdeclarelayer{background}
\pgfsetlayers{background,main}

\newcommand{\tikzvertex}[1]{\tikz[baseline=default]{ \node [#1] {};}}

\def\L{\Lambda}
\def\bH{\mathbb{H}}
\def\l{\lambda}

\def\tJJ{{\tilde{\bs{J}}}}

\def\s{\sigma}

\def\0{{\bf{0}}}

\def\CC{\mathcal{C}}

\begin{document}
\title{The scaling limit of boundary spin correlations in non-integrable Ising models}
\author[1,a]{Giulia Cava}
\author[1]{Alessandro Giuliani}
\author[2,*]{Rafael Leon Greenblatt}
\affil[1]{\small{Universit\`a degli Studi Roma Tre, Dipartimento di Matematica e Fisica, L.go S. L. Murialdo 1, 00146 Roma, Italy}}
\affil[2]{\small{Universit\`a degli Studi di Roma ``Tor Vergata'', Dipartimento di Matematica, Via della Ricerca Scientifica 1, 00133 Roma, Italy}}
\affil[a]{Current affiliation: Sapienza Università di Roma, Dipartimento di Matematica, Piazzale Aldo Moro 5, 00185, Roma, Italy}
\affil[*]{{Corresponding author.  Email address: \texttt{\href{mailto:greenblatt@mat.uniroma2.it}{greenblatt@mat.uniroma2.it}}}}
\date{\today}

\maketitle

\begin{abstract} 
We consider a class of non-integrable $2D$ Ising models obtained by perturbing the nearest-neighbor model via a weak, finite range potential which preserves translation and spin-flip symmetry,
and we study its critical theory in the half-plane. 
We prove that the leading order long-distance behavior of the correlation functions for spins on the boundary is the same as for the nearest-neighbor model, up to an analytic multiplicative renormalization constant.  In particular, the scaling limit is the Pfaffian of an explicit matrix.  
The proof is based on an exact representation of the generating function of correlations in terms 
of a Grassmann integral and on a multiscale analysis thereof, generalizing previous results to include observables located on the boundary.
\end{abstract}

\tableofcontents

\section{Introduction}\label{sec:intro}

The Ising model with pair interactions on a planar graph is ``free fermions'' in the sense that many correlation functions can be expressed in terms of Pfaffian formulas\footnote{We recall that the Pfaffian of a $m\times m$ antisymmetric matrix $M$ is defined as $$\Pf M = \frac1{2^{m/2}(m/2)!} \sum_\pi (-1)^\pi M_{\pi(1),\pi(2)}\ldots M_{\pi (m-1),\pi(m)}$$ where the sum is over permutations $\pi$ of $(1,\ldots, m)$, with $(-1)^\pi$ denoting the signature. One of the properties of the Pfaffian is that $(\Pf M)^2 = \det M$.} similar to the Wick rule for Fermionic quantum field theories.  This is most directly true for the so-called Kadanoff-Ceva Fermions or parafermionic observables \cite{KC,CCK,Izyurov:Critical_Ising_interfaces_in_multiply_connected_domains}, which are not naturally defined as local functions of the spin configuration, but the correlation functions of some local observables also have such a form.  In particular the correlation functions of an arbitrary number of `energy observables' (i.e., products of pairs of neighbouring spins), and of 
an even number of spins located on the boundary of a domain with open boundary conditions also have such a Pfaffian form \cite{MW, GBK78}.
Alternatively these fermionic fields are related to the distribution of interfaces or curves which are one of the main ingredients in proofs of convergence of the interface process to a conformally invariant distribution known as Schramm-Loewner evolution (SLE) \cite{CDHKS:Convergence_of_Ising_interfaces_to_SLE}.  In the context of the Ising model the formulation associated with interfaces between different spin values in Dobrushin boundary conditions is closely related to Kadanoff-Ceva fermions.  With an appropriate normalization, boundary values of these observables are given by correlation functions of disorder observables, which are equal by duality to spin correlation functions in open boundary conditions (see \cite{CCK} for a detailed discussion of all of these quantities and the relationships between them).

These properties of the planar Ising model are closely related to integrability or exact solvability, and hold for all planar lattices \cite{CCK} but are generically lost if the Hamiltonian is modified by adding non-planar couplings or interaction terms involving four or more spins.
On the other hand, the planar Ising model is expected to be just a representative of a universality class; loosely speaking, other ferromagnetic 
local models with the same symmetries should have the same behaviour at long distances close to the critical temperature (that is, close to \emph{their} critical temperature, which may be different for different models in the same universality class).
In previous works \cite{GGM,AGG_AHP,AGG_CMP}, using constructive, Fermionic, Renormalization Group techniques, we substantiated this prediction for the multipoint bulk 
energy correlations of non-planar Ising models, possibly with even many-spin interactions, by fully computing their scaling limit.
Concerning boundary spin correlations in the half-plane, 
Aizenman et.\ al.\ \cite{ADTW} proved that their leading long-range behavior can be expressed, as expected, in terms of Pfaffian formulas, even though their techniques, based on the 
random current representation and a generalization of Russo-Seymour-Welsh theory \cite{Russo,SW78}, did not allow them to compute the scaling limit, or to prove that the critical exponent of boundary spin correlations is the same as the one of the integrable Ising model. This is what universality indicates: the long-distance behavior of the critical correlation functions should be the same for different models, up to rescaling and change of the critical temperature.
In the planar case, this is known to be the case for a wide variety of graphs so long as the couplings are chosen in a precise way to satisfy a local criticality condition \cite{CS,Chelkak.embeddings}.
On the other hand, many other assignments of coupling constants (such as quenched disorder) qualitatively change the critical behavior, see for example \cite{DD,GG.Balanced,MW.random1,MW.random2,CGG.Continuum}, although some (in particular quasi-periodic disorder) do not \cite{GM23}.

In this work we study boundary spin correlations for a wide variety of short-range Hamiltonians (including weak non-planar and multispin terms) in the square lattice in a half-plane, and show that at their critical temperature they asymptotically have the same form as in the uniform planar case, and in particular the leading terms for the many-spin correlations can be expressed as a Pfaffian in terms of the two-point function, which (asymptotically and up to a rescaling) is the same as that of the integrable, critical, nearest neighbour Ising model.

For the planar Ising model, spin correlations on the boundary are related by duality to the partition function of the model with Dobrushin boundary conditions.  Although the relationship is not so direct in the non-integrable case, the result is still a Grassmann integral of the same form.  Furthermore, it is possible to extend this to a quantity defined away from the boundary and which can be used to control the interface distribution in the same way as the Kadanoff-Ceva fermion of the planar model \cite{GP.prog}.  The present work is a first step in controlling this quantity as well.

We proceed by using the representation of the Ising model on the cylinder in terms of Grassmann variables to apply Fermionic constructive renormalization group techniques.
The Grassmann representation for the planar Ising model is based on a formula for the partition function in terms of a formal integral (known as a Grassmann or Berezin integral) on polynomials in anticommuting variables which can be associated with a contour representation of the Ising model (here we use a form associated with the high-temperature contours) or with the Kadanoff-Ceva fermions (see \cite{CCK} and the references found there).

The renormalization group techniques we use were developed for similar formulations of quantum field theories (reviewed for example in \cite{Rivasseau.book,Ma}) and also applied to other statistical mechanical models such as eight-vertex models \cite{BFM02}, quantum spin chains \cite{BM10,BM11}, and 
interacting dimer models \cite{GMT17a,GMT17,GMT15,GMT19}. In the context of perturbed 2D Ising models,
this approach was introduced to study the effect of a particular four-spin term on the singularity of the specific heat by Pinson and Spencer \cite{PS,S00}; later, it was 
generalized to coupled pairs of Ising layers (Ashkin-Teller model) \cite{GM05,M04}, and to the study of multipoint energy correlations in single-layer, non-integrable, Ising models \cite{GGM}.
Until recently, most constructive renormalization group treatments were formulated using the Fourier transform in ways that relied strongly on translation invariance and conserved momentum, making it impossible to consider boundary effects (although several works involved other departures from translation invariance \cite{Mas99,Mas17.localization,GM23}), until \cite{AGG_AHP,AGG_CMP} which implemented the technique on a cylinder, studying energy correlations at a distance from the boundary comparable to their mutual distances.  Among other things, the present work extends this technique by employing it to study observables localized on the boundary.

\paragraph{The models under consideration.} 
For positive integers $L$ and $M$, we let $G_\L$ be the discrete cylinder with side lengths $L$ and $M$ in the horizontal and vertical directions, respectively, with periodic boundary conditions in the horizontal direction and open boundary conditions in the vertical direction. To be precise, we consider the weighted graph $G_\L$ with vertex set $\L=\left( \mathbb Z/L\mathbb Z \right) \times \left(\mathbb Z\cap [1, M]\right)$ (with the set $\mathbb Z/L\mathbb Z$ identified with the discrete interval $\mathbb Z\cap (-L/2,L/2]$ with periodic boundary conditions), 
and edge set $\mathfrak{B}_\L$ consisting of all pairs of the form $\{z, z+\hat{e}_i\}$ for $z \in \L$, $i \in \{1,2\}$ and $\hat{e}_1,\hat{e}_2$ the unit vectors in the two coordinate directions, including the pairs of the form $\{(\lceil\frac{L-1}2\rceil,(z)_2), (\lceil-\frac{L-1}2\rceil,(z)_2) \}$  with $((z)_2 \in \{1,\dots,M\}$\footnote{We denote the components of $z\in \L$ by $(z)_1, (z)_2$, rather than by $z_1,z_2$, 
in order to avoid confusion with the first two elements of an $n$-tuple $\bs z\in\L^n$, for which we will use the notation $\bs{z}=(z_1,\ldots, z_n)$.}; each edge $\{z, z+\hat{e}_i\}$ is associated with the weight $J_i>0$. The model we consider is defined by the Hamiltonian
\begin{equation}
\label{eq:HM}
H_{\L}(\sigma)=-\sum_{i=1}^2J_i\sum_{\{z, z+\hat{e}_i\} \in \mathfrak{B}_\L}\sigma_{z}\sigma_{z+\hat{e}_i}  -\lambda \sum_{X\subset \L} V(X) \prod_{z\in X}\sigma_z,
\end{equation}
where  $\sigma = \set{\sigma_z}_{z \in \L}$ denotes the spin configuration and $\sigma_z \in \{+1,-1\}$ denotes the local spin variable; if $z=(\lceil\frac{L-1}2\rceil,(z)_2)$ for some $(z)_2\in\{1,\ldots,M\}$, then we interpret $\sigma_{z+\hat e_1}$ as 
being equal to $ \sigma_{(\lceil-\frac{L-1}2\rceil,(z)_2)}$ (periodic boundary conditions in the horizontal direction); if $z=((z)_1,M)$ for some $(z)_1\in\{1,\ldots,L\}$, then we interpret $\sigma_{z+\hat e_2}$ as being equal to zero (open boundary conditions in the vertical direction); $V(X)$ is a finite range, translation invariant, even function; finally $\lambda$ is a parameter representing the strength of the interaction, which can be of either sign and, for most of the discussion below, the reader can think of it as being independent of the system size and small compared to the ferromagnetic nearest-neighbor interactions $J_1$ and $J_2$. 

\Cref{eq:HM} with $\lambda= 0$ describes a nearest-neighbor Ising model, which is integrable and which is also known in the context of the class of models under consideration as the \textit{non-interacting} model, while with $\lambda \ne 0$ it describes a perturbed nearest-neighbor Ising model, which is generically non-integrable, known as the \text{interacting} model. These terminological conventions, along with several others we use in this article, are motivated by analogy with quantum field theory (in particular, by the Grassmann representation of the model discussed in Section~\ref{sec:gen} below).

\paragraph{Boundary spin correlations.} 

We fix once and for all an interaction $V$ with the properties spelled out after \eqref{eq:HM} and we assume $J_1/J_2$ belongs to a compact $K \subset (0,+\infty)$. We let $t:= (t_1(\beta), t_2(\beta))$, with $t_i(\beta):=\tanh (\beta J_i)$, $i=1,2$, and recall that, if $\lambda=0$, 
the critical temperature $\beta_c(J_1,J_2)$ is the unique solution of $t_2(\beta)=(1-t_1(\beta))/(1+t_1(\beta))$. Note that there exists a suitable compact  $K'\subset (0,1)^2$ such that whenever $J_1/J_2 \in K$ and $\beta \in [\tfrac12 \beta_c(J_1,J_2),2 \beta_c(J_1,J_2)]$, 
then $t \subset K'$.  

\medskip

Let  $\partial^l\L:=\{z=((z)_1,(z)_2) \in\L: (z)_2=1\}$ be the lower boundary of the cylinder and $\bs{y} = (y_1,\ldots,y_m) \in (\partial^l \L)^m$, $m \in 2 \mathbb{N}$, be a specific sequence of distinct boundary vertices ordered from the right to the left. The multipoint correlation for the spins located in $\bs{y}$ is
\begin{equation}\label{e.1}\media{\sigma_{y_1}\cdots\sigma_{y_{m}}}_{\lambda, t;\Lambda}\,,\end{equation}
where
$\media{\cdot}_{\lambda, t;\L}$ is the average with respect to the Gibbs measure associated with $H_\L$ in \cref{eq:HM} at inverse temperature $\beta >0$; that is, 
given an observable $F:\Omega_\L\to\mathbb R$, 
\begin{equation} \label{media}
\media{F(\s)}_{\lambda, t;\L}:= \frac{\sum_{\sigma\in\Omega_{\L}}e^{-\beta H_{\L}(\sigma)}F(\sigma)}{Z_{\lambda, t;\L} }\,,
\end{equation}
where $Z_{\lambda, t;\L} := \sum_{\sigma\in\Omega_{\L}} e^{-\beta H_{\L}(\sigma)}$ is the corresponding partition function and $\Omega_\L:= \set{\pm 1}^{|\L|}$ is the spin configuration space. Note that $H_{\L}$ in \eqref{eq:HM} is invariant under simultaneous flip of all of the spin values so that \eqref{e.1}, defined as in \eqref{media} using a sum over all spin configurations, is non-zero only for even $m$. 
\medskip 

\paragraph{Main result.} 

Consider the infinite volume limit in which the cylinder $\Lambda$ becomes the discrete upper half-plane $\bH:=\mathbb Z\times\mathbb N$ (with $\mathbb N$ the set 
of positive integers) and let $\lim_{L,M \to \infty} \media{\cdot}_{\lambda,t(\lambda);\L} := \media{\cdot}_{\lambda,t(\lambda);\bH}$, provided this limit exists (the limit $L,M\to\infty$ is performed by keeping $C^{-1}\le L/M\le C$ for some $C>0$).
Our main result can be stated as follows. 

\begin{theorem}
	\label{prop:main} Fix $V$ as discussed above, $J_1,J_2 \in (0,\infty)$, $K\subset (0,\infty)$ compact with $J_1/J_2$ in the interior of $K$, and $\theta\in(0,1)$.
There exist $\lambda_0>0$ and analytic functions
 $\beta_c(\lambda)$, $t_1^*(\lambda)$, $Z_{\Bs}(\lambda)$ 
 defined for $|\lambda| < \lambda_0$, 
such that, for any $\bs y = (y_1,\ldots,y_m) \in( \partial \bH )^m$, $m \in 2\mathbb{N}$, right-to-left ordered sequence of distinct boundary vertices,
	\begin{equation}\begin{aligned}
		&\media{\sigma_{y_1}\cdots\sigma_{y_{m}}}_{\lambda,t(\lambda);\bH}= \big(Z_{\Bs}(\lambda)\big)^{m}\media{\sigma_{y_1}\cdots\sigma_{y_{m}}}_{0,t^*(\lambda);\bH}
		+ R_\bH(\bs y), 
		\label{eq:corr_main_statement}
	\end{aligned}\end{equation}
where $t(\lambda):= (t_1(\lambda),t_2(\lambda))$, $t_i (\lambda) := \tanh(\beta_c(\lambda)J_i)$, $i=1,2$,  $t^*(\lambda) = (t_1^*(\lambda), t_2^*(\lambda))$, $ t_2^*(\lambda):= (1-t_1^*(\lambda))/(1+t_1^*(\lambda))$, and $R_\bH(\bs y)$ is an analytic function of $\lambda$ for $|\lambda|<\lambda_0$, uniformly in the choice of $\bs y$. 
Moreover, there exist positive-valued functions $C(\theta)$ and $\lambda_1(\theta)$, defined on $[0,1)$, monotone increasing and decreasing in $\theta$, respectively, 
such that $\lambda_1(0)=\lambda_0$ and, for any $\theta\in[0,1)$ and $|\lambda|<\lambda_1(\theta)$, 
\begin{equation} |R_\bH(\bs y)|\le  (C(\theta))^{m}  |\lambda| m! d^{-(m/2+\theta)}
\;,\label{10b}\end{equation}
where 
$d$ is the minimum distance between two consecutive elements in $\bs y$.
\end{theorem}

The existence of $\lim_{L,M\to\infty} \media{\sigma_{y_1}\cdots\sigma_{y_{m}}}_{\lambda,t(\lambda);\L}\equiv 
\media{\sigma_{y_1}\cdots\sigma_{y_{m}}}_{\lambda,t(\lambda);\bH}$ is part of the claim of the theorem. Note that
$\beta_c(\lambda)$ has the interpretation of an interacting inverse critical temperature, and $\media{\cdot}_{0,t^*(\lambda);\bH}$ 
plays the role of the `non-interacting' critical Gibbs measure. In fact, \cref{eq:corr_main_statement} tells us that, 
for the purpose of computing the multipoint boundary spin correlations, we can use this non-interacting measure instead of the interacting one, 
up to the finite multiplicative renormalization constants $Z_\Bs(\lambda)$ (which is generically non-trivial, i.e., non-identically 1, see Appendix \ref{ordineL})
and the remainder term $R_\bH$. The reason why we can think of $R_\bH$ as a remainder is the following. It is well known 
that integrable boundary spin correlations have a Pfaffian structure
\cite{GBK78}\footnote{Strictly speaking, in \cite{GBK78} the Pfaffian formula for the boundary spin correlations is discussed only in the case of simply connected domains with free boundary conditions.  The same form holds for spins on a single boundary component of a cylinder (in \cref{sec:gen} we derive a different version of this formula which is more suited for the current context), but there are complications when some of the spins are located on each boundary component, see \cref{ARC} below.}, namely 
\begin{equation} \label{PfForm}
\media{\sigma_{y_1}\cdots\sigma_{y_{m}}}_{0,t^*(\lambda);\bH}= \Pf( M_\bH (\bs y))\,,
\end{equation}
where $M_\bH (\bs y)$ is the $m\times m$ antisymmetric matrix whose elements above the main diagonal ($1 \le i<j  \le m$) are given by $[M_\bH (\bs y)]_{ij}= \langle \sigma_{y_i}\sigma_{y_j}\rangle_{0,t^*(\lambda);\bH}$. Since the two-point boundary spin integrable correlation  $\langle \sigma_{y_i}\sigma_{y_j}\rangle_{0,t^*(\lambda);\bH}$ decays as $\propto |y_i-y_j|^{-1}$ at large separation $|y_i-y_j|$, the Pfaffian appearing in the r.h.s.\ of \eqref{PfForm} decays as $\propto |y_i-y_j|^{-m/2}$ and the term $R_\bH$ in \eqref{10b} is sub-dominant, compared to the first term in the r.h.s.\ of \eqref{eq:corr_main_statement}.

In particular, after appropriate rescaling of the lattice spacing and of the boundary spin observables, as a corollary of 
Theorem~\ref{prop:main} we obtain the scaling limit of the boundary spin correlations, with an explicit rate of convergence. Let $\bH^a:=a\mathbb H$ 
be the discrete upper half-plane with lattice spacing $a$, and $\bH^c=\mathbb R\times[0,+\infty)$ the continuum, closed, upper half-plane, which $\mathbb H^a$ reduces to in the limit $a\to0$.
Given $z$ at the boundary of $\bH^c$, we define the boundary spin observable as follows:
\begin{equation} \sigma^a (z)= a^{-\frac12} \sigma_{a\lfloor a^{-1} z\rfloor }.\end{equation} 
Fix a right-to-left ordered sequence $\bs y=(y_1,\ldots, y_{m})$ of distinct vertices on the lower boundary of $\bH^{c}$.
Theorem \ref{prop:main} tell us that
\begin{equation}\begin{aligned}
		&\media{\sigma^a({y_1})\cdots\sigma^a({y_{m}})}_{\lambda,t(\lambda);\bH^a}= \big(Z_\Bs(\lambda)\big)^{m}\media{\sigma^a({y_1})\cdots\sigma^a({y_{m}})}_{0,t^*(\lambda);\bH^a}
		+ R_{\bH^a}(\bs y), 
		\label{eq:corr_a}
	\end{aligned}\end{equation}
where  $\media{\sigma^a({y_1})\cdots\sigma^a({y_{m}})}_{\lambda,t(\lambda);\bH^a}:=a^{-m/2}\media{\sigma_{\lfloor a^{-1}y_1\rfloor}\cdots\sigma_{\lfloor a^{-1}y_{m}\rfloor}}_{\lambda,t(\lambda);\bH}$, and the remainder $R_{\bH^a}$ can be
bounded as
\begin{equation} |R_{\bH^a}(\bs y)|\le  C_{\theta}^{m}  |\lambda| d^{-(m/2+\theta)} a^\theta\;,\label{10b_a}\end{equation}
where $d$ is the minimum distance between two consecutive elements in $\bs y$. Clearly, for any fixed $\bs y$  the r.h.s.\ of \eqref{10b_a} vanishes as $a\to0$. Moreover, under the 
same assumptions as \cref{prop:main} 
\begin{equation}\begin{aligned}
		&\lim_{a \to 0^+} \media{\sigma^a({y_1})\cdots\sigma^a({y_{m}})}_{\lambda, t(\lambda);\bH^a}= \big(Z_\Bs(\lambda)\big)^{m} \lim_{a \to 0^+}  \media{\sigma^a({y_1})\cdots\sigma^a({y_{m}})}_{0,t^*(\lambda);\bH^a}\,,
		\label{eq:corr_ato0}
	\end{aligned}\end{equation}
and the limit in the r.h.s.\ exists and equals 
\begin{equation}\label{Pfscal}
\lim_{a \to 0^+}  \media{\sigma^a({y_1})\cdots\sigma^a({y_{m}})}_{0,t^*(\lambda);\bH^a} = \Pf(M_{\text{s.l.}})(\bs{y})\,,
		\end{equation}
where $M_{\text{s.l.}}$ is the $m\times m$ antisymmetric matrix, whose elements above the main diagonal ($1\le i<j\le m$) are given by $[M_{\text{s.l.}}(\bs y)]_{ij}= \lim_{a \to 0^+}  \media{\sigma^a({y_i})\sigma^a({y_{j}})}_{0,t^*(\lambda);\bH^a}$ which in the isotropic case ($t_1 = t_2$) is given by $\left( \frac{2}{\pi (\sqrt{2}-1)}\right)^2  |y_i-y_j|^{-1}$.

\paragraph{Overview of the proof and generalizations.}
The proof of Theorem \ref{prop:main} is based on: (1) a representation of the $m$-point boundary spin correlations on a finite cylinder in terms of the $m$-point boundary field  correlations of an effective non-Gaussian Grassmann theory; (2) a multiscale analysis of the corresponding non-Gaussian Grassmann generating function in the thermodynamic limit, via a conceptually straightforward (if somewhat technically involved) modification  of the construction in \cite{AGG_AHP}. 

In \cref{sec:gen}, we derive the effective non-Gaussian Grassmann representation of the generating function for correlations of an even number of spins placed on the lower boundary of the cylinder. 
This starts with the expression for the corresponding generating function for the (integrable) nearest-neighbor Ising model in terms of (Gaussian) Berezin integrals.
As illustrated in \cref{ARC}, we are able to obtain the effective Grassmann representation for the model in \eqref{eq:HM} with both periodic and anti-periodic spin boundary conditions in the horizontal direction and for any even number of spins placed on the boundary of the cylinder: in particular, we can also consider the case in which there are spins placed on both the upper and lower boundaries of the cylinder, regardless of whether it is an even or odd number of spins on each boundary. The derived representation for correlations in the model with periodic/anti-periodic horizontal boundary conditions will also depend in some sense on the anti-periodic/periodic horizontal boundary conditions, so that considering more general correlation functions requires further modifications of the objects defined in \cite{AGG_AHP}.
We expect the necessary modifications to be a bit involved, although conceptually straightforward, so we prefer not to include them here, in order to keep the technicalities to 
a minimum.
Similarly, in order not to overwhelm the technical discussion of the following sections, we limit ourselves to an analysis of the model directly in the half-plane limit, leaving the 
details of the proof of the existence of the limit of the boundary spin correlations as $L,M\to\infty$ to the reader: these are a straightforward adaptation of the analogous estimates worked out in \cite{AGG_CMP} and we refer the interested reader to that paper for details.

In \cref{sec:3} we present the multiscale representation for the generating functional described above, obtaining a convergent diagrammatic expansion for the correlation functions under consideration.  As usual, the key prerequisite for convergence is that ``counterterms'' need to be chosen appropriately via a fixed point construction.  We adapt the approach of previous works (especially \cite{AGG_CMP,AGG_AHP}) to accommodate edge observables, extending estimates on the scaling behavior of the terms in the expansion.
Finally, in \cref{sec:corr} we isolate a ``quasi-free'' part of the correlations which dominates the behavior at long distance and bound the remaining part.  The strategy used again follows previous works (in particular \cite{GGM}), with some modifications, in particular those due to the fact that the observable we consider is represented by odd-order polynomials in the Grassmann variables, unlike the even-order polynomials for the observables considered previously \cite{GGM,AGG_CMP}.

\section{Grassmann representation of the generating function.}
\label{sec:gen}
In this section we derive the Grassmann representation of the generating functional for the correlations of spins placed on the lower boundary of the cylinder. First, in \cref{lambda0}, we write the generating functional for the boundary spin correlations in the integrable case (described by \eqref{eq:HM} with $\lambda = 0$) as the partition function of such a model on a graph with added auxiliary edges, whose weights play the role of sources. Such a partition function can be expressed in terms of Grassmann variables (we use a variant of the calculation for the cylinder in \cite[Chapter VI.3]{MW}) as a `Gaussian' Grassmann integral.  Then, in \cref{lambdane0}, using a construction from \cite{GGM,AGG_AHP}, we write the generating functional in the non-integrable case (\cref{eq:HM} with $\lambda \ne 0$) as a `non-Gaussian' Grassmann integral. Finally, in \cref{xiphi}, we rewrite the generating function by introducing the reference Gaussian measure and the appropriate parameters to renormalize; this final rewriting will be the basis of the multiscale analysis in the following sections. 
Note that in this section we consider the generating function for the correlations of spins placed on the lower boundary of the cylinder: in the next section we will introduce the limit of the (upper) half-plane; the case of spins placed on the upper boundary of the cylinder follows by symmetry (as well as being completely analogous); the case in which spins are present on both boundaries is described in \cref{ARC}.

\subsection{Grassmann representation: integrable models.}\label{lambda0}

Let $\bs y=(y_1,\ldots, y_{m}) \in (\partial^l \L)^m$, $m \in 2\mathbb{N}$, be a right-to-left ordered sequence of distinct boundary vertices, let $\tilde{G}_\L = (\L, \mathfrak{B}_\L \cup \tilde{\mathfrak{B}})$ be the weighted graph, shown in Fig.~\ref{Lambdan}, obtained from the graph $G_\L = (\L, \mathfrak{B}_\L)$ by adding the edge set $\tilde{\mathfrak{B}}$ consisting of $m/2$ auxiliary edges below the lower boundary connecting the sites labelled by $\bs y$ as shown in \cref{Lambdan} and let 
\begin{equation}\label{H_addb}
\tilde{H}_{\L}(\s)=-\sum_{i=1}^2J_i\sum_{\{z, z+\hat{e}_i\} \in \mathfrak{B}_\L}\sigma_{z}\sigma_{z+\hat{e}_i} - \sum_{k=1}^{m/2} {\tilde{J}}_{k}  \s_{y_{2k-1}}\s_{y_{2k}}
\end{equation}
be an auxiliary spin Hamiltonian, depending upon the auxiliary couplings $\tJJ = (\tilde{J}_1,\ldots, \tilde{J}_{m/2})$, describing a nearest-neighbor model defined on the modified graph $\tilde{G}_\L$. Note that the first term in the r.h.s.\ is nothing but the model in \eqref{eq:HM} with $\lambda=0$, so that  Eq.~\eqref{H_addb} describes a modified integrable model with periodic horizontal boundary conditions.
\begin{figure}
\centering
\includegraphics[scale=0.6]{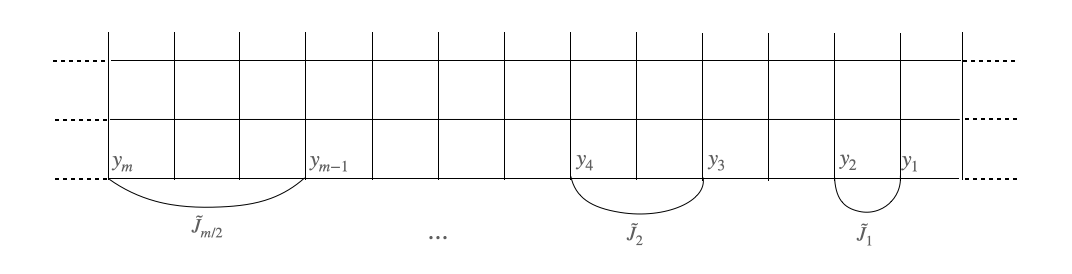}
\caption{The weighted graph $\tilde{G}_\L$ associated with the auxiliary Hamiltonian \eqref{H_addb}, obtained from the weighted graph $G_\L$ 
by adding auxiliary edges below the lower boundary  so that each auxiliary edge connects the pair $\{y_{2k-1}, y_{2k} \}\in \bs y$ and has weight $\tilde{J}_k$, $k = 1, \ldots, m/2$. Note that the auxiliary edges do not intersect each other, do not intersect horizontal or vertical edges and do not surround the horizontal dashed edges that connect the first and last column of the graph.\label{Lambdan}}
 \end{figure}
Let ${Z}_{0,t;\L} (\tJJ)=\sum_{\sigma \in \Omega_\L}e^{- \beta \tilde{H}_{\L} (\sigma)}$, $\beta >0$, be the partition function corresponding to \cref{H_addb}: then 
\begin{equation}\label{evspin}\begin{aligned}
&\media{\sigma_{y_1}\cdots\sigma_{y_{m}}}_{0,t;\L}= \left.\frac{ \beta^{-m/2}}{Z_{0,t;\L} ( \tJJ) } \frac{\partial^{m/2}}{\partial \tilde{J}_1 \cdots \partial \tilde{J}_{m/2}} Z_{0,t;\L} ( \tJJ)  \right|_{\tJJ =\0}\,,\end{aligned}
\end{equation}
so this serves as a generating function for the multipoint correlations in \eqref{e.1} (cf.\ \eqref{media}) with $\lambda = 0$ where $\tJJ$ plays the role of sources.

In order to obtain a Pfaffian formula for $Z_{0,t;\L} ( \tJJ)$ and, correspondingly, its Grassmann representation, we follow the strategy of 
rewriting it as the partition function of a dimer model on a decorated `Fisher's' lattice \cite{Fi66}.  The procedure is standard, but we spell out the main steps below, 
in order to allow the reader to track the effects due to the presence of the additional edges below the boundary.
First of all, we replace each graph element of $\tilde{G}_\L$ consisting of a vertex $z \in \L$ and the four edges exiting from it by a new `expanded' graph element, consisting of 
six new vertices $\{\bar{H}_z, H_z, \bar{V}_z, V_z, \bar{T}_z, T_z\}$ connected among them via `short edges' and to nearest-neighbor vertices via `long edges', as described in Fig.~\ref{fig:subfigA}: consequently, each elementary face surrounded by horizontal and vertical edges is replaced as in Fig.~\ref{fig:subfigB}, each elementary face surrounded by an additional edge is replaced as in Fig.~\ref{fig:subfigC} and we get the decorated planar graph $\tilde{G}_*$ shown in Fig.~\ref{figOK}.
\begin{figure}[ht]\centering
\begin{subfigure}[h]{0.4\textwidth}
      \includegraphics[scale=0.4]{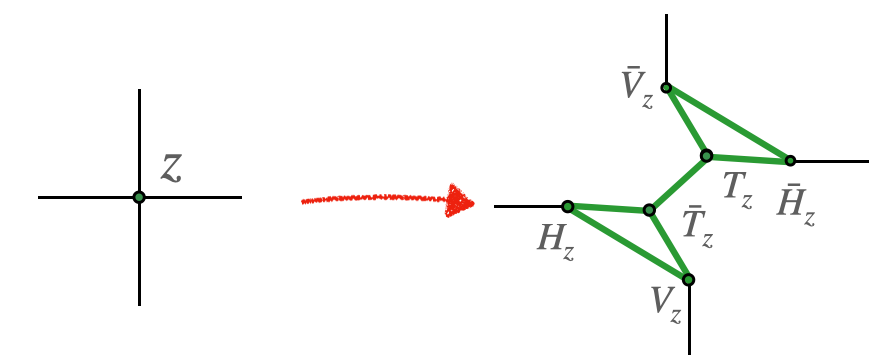}    
\caption{} \label{fig:subfigA}
\end{subfigure}\qquad \qquad\qquad
\begin{subfigure}[h]{0.4\textwidth}
      \includegraphics[scale=0.4]{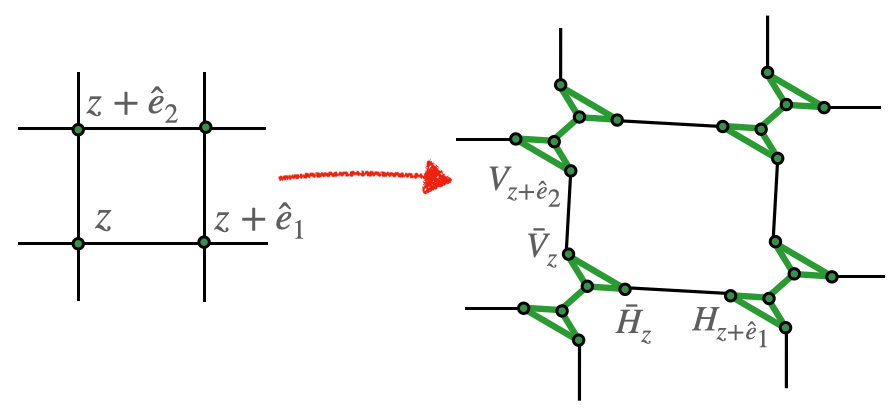}    
 \caption{}   
        \label{fig:subfigB}     \end{subfigure}
 \vfill\centering
\begin{subfigure}[l]{0.5\textwidth}
      \includegraphics[scale=0.8]{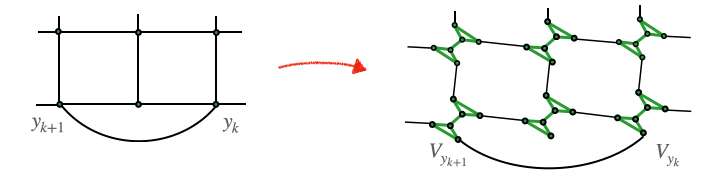}        
    \caption{}
        \label{fig:subfigC}      \end{subfigure}
        \caption{In Fig.~\ref{fig:subfigA}, a vertex  $z \in \L$ (left) corresponds to a cluster of vertices $ \{\bar{H}_z, H_z, \bar{V}_z, V_z, \bar{T}_z, T_z\} $ connected via (green) short edges between them (right); in Fig.~\ref{fig:subfigB}, horizontal and vertical edges between nearest neighbor vertices (left) correspond to long edges between nearest neighbor clusters (right); in Fig.~\ref{fig:subfigC}, additional edges below lower boundary (left) correspond to edges between boundary clusters (right).}\label{replacement123}
   \end{figure}
  \begin{figure}
\centering
\includegraphics[scale=0.5]{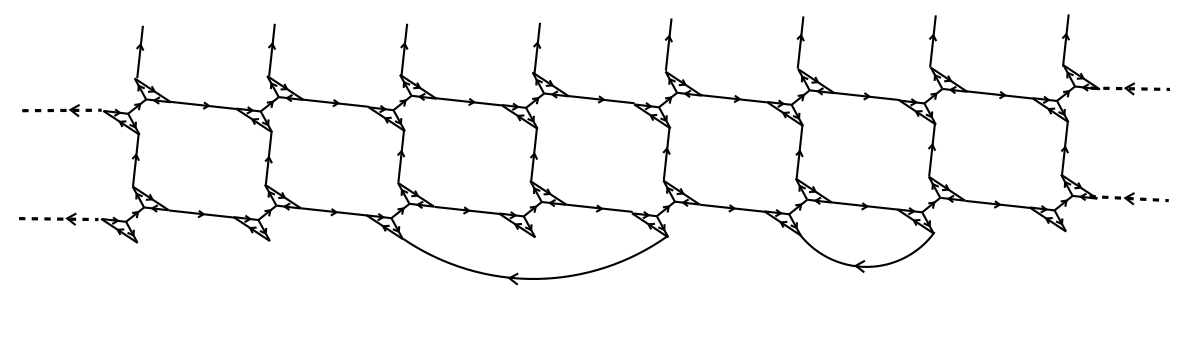}
\caption{The decorated graph $\tilde{G}_*$ in which the arrows correspond to an appropriate clockwise-odd orientation. Starting from the clockwise-odd orientation used in \cite[Chapter 5]{MW} for the graph without the additional edges, the clockwise-odd orientation of $\tilde{G}_*$ is obtained by directing the additional edges from right to left, regardless of the number of vertices surrounded by an additional edge and regardless of the number of additional edges inserted. Note that to get the clockwise-odd orientation, the dashed long horizontal edges connecting the last and the first column of the graph must have a direction opposite to that of all other long horizontal edges of the graph; finally, note that additional edges never surround these horizontal dashed edges.}
\label{figOK}
\end{figure}
Then, we can compute $Z_{0,t;\L} ( \tJJ)$ via Kasteleyn theorem \cite{Kas61}: since the graph $\tilde{G}_*$ is planar (by construction the edges do not intersect, see \cref{Lambdan}), $Z_{0,t;\L} ( \tJJ)$ can be expressed in terms of the Pfaffian of the corresponding  Kasteleyn matrix, thus getting
\begin{equation}\label{eq:Z'Pf}
Z_{0,t;\L} ( \tJJ)=  {C}_\beta( J_1,J_2, \tJJ) \cdot
\Pf{A_{K}}\,,
\end{equation}
where $C_\beta( J_1,J_2, \tJJ) := (2\cosh{\beta J_1})^{LM} (\cosh{\beta J_2})^{L(M-1)} \Big(\prod_{k=1}^{m/2} \cosh{\beta \tilde{J}_k}\Big)$ and $A_{K}$ is an antisymmetric $6LM \times 6LM$ matrix, whose entries are labelled by pairs of vertices of $\tilde{G}_*$ connected by an edge and 
are equal to the weight of the corresponding edge (that is, $1$ for short edges, $t_{i}(\beta):=\tanh \beta J_i$, $i=1,2$, for the long edges and 
$\tilde{t}_k(\beta) := \tanh \beta \tilde{J}_k$, $k=1,\ldots, m/2$, for the additional edges below the boundary), 
times a suitable sign. In order to compute this sign, we need to define a `{clockwise-odd} orientation': that is, we need to assign an orientation to all the edges of $\tilde{G}_*$, in such a way that each elementary face of $\tilde{G}_*$ is surrounded by an odd number of edges directed clockwise. A possible 
choice of the edge orientations compatible with the clockwise-odd rule is shown in Fig.~\ref{figOK} (cf.\ \cite[Chapter~V, Figure~5.4]{MW}). The orientation of the edge $e$ of vertices $z,z'$ is then interpreted as a sign associated with the ordered pair $(z,z')$, equal to $+$ (resp.\ $-$) if the edge is oriented from $z$ to $z'$ (resp.\ $z'$ to $z$). Note that although not made explicit by the notation, $A_K$ (and therefore $\Pf A_K$ in the r.h.s.\ of \eqref{eq:Z'Pf}) depends on $\tJJ$ via $\tilde{t}:= \{\tilde{t}_k(\beta)\}_{k=1}^{m/2} = \{ \tanh \beta \tilde{J}_{k}\}_{k=1}^{m/2}$, the set of weights associated with additional edges on $\tilde{G}_*$.
Moreover, rewriting the Pfaffian in \eqref{eq:Z'Pf} in the form of a Grassmann Gaussian integral, we obtain
\begin{equation}
Z_{0,t;\L} ( \tJJ)=   C_\beta (J_1,J_2, \tJJ) \cdot \int  \mathcal D \Phi'\, e^{\cS'_{t}(\Phi') + \cS_{\tilde{t}}(\Phi')},
\label{PfA'}
\end{equation}
where $\Phi'=\{(\bar H_{z},H_{z},\bar V_{z},V_{z}, \bar T_z,  T_z)\}_{z\in \L}$ is a collection of $6LM$ Grassmann variables, $\mathcal D \Phi'$ denotes the Grassmann `differential' 
$$\mathcal D \Phi'=\prod_{z\in\L}d\bar H_{z} dH_{z}  d\bar V_{z} dV_{z} d\bar T_z d  T_z\,,$$
\begin{equation} 
\cS'_{t}(\Phi') := \sum_{z\in\L}(t_1(\beta)  \bar H_{z} H_{z+\hat e_1}+t_2 (\beta)  \bar V_{z} V_{z+ \hat e_2}+\bar V_{z} \bar H_{z}+\bar H_{z} T_{z}+ V_{z}  H_z+ H_{z}\bar T_{z}+ T_{z} \bar V_{z}+ \bar T_{z} V_{z} + \bar T_{z} T_{z}) \,,
\label{eq:cS_def_T}
\end{equation}
where: $t = (t_1(\beta),t_2(\beta))\equiv (\tanh \beta J_{1}, \tanh\beta J_2)$; for any  $(z)_1 \in \set{1,\ldots,L}$, $V_{((z)_1,M+\hat{e}_2)}$ should be interpreted as representing $0$; for any $(z)_2 \in \set{1,\ldots,M}$, $H_{(L+\hat{e}_1,(z)_2)}$ should be interpreted as representing $- H_{(1,(z)_2)}$\footnote{The identification of $H_{(L+1,(z)_2)}$ with $-H_{(1,(z)_2)}$ corresponds to anti-periodic boundary conditions in the horizontal direction for the Grassmann variables: in fact, as shown in Fig.~\ref{figOK}, to obtain a clockwise-odd orientation the horizontal edges that connect the last and first column are directed from right to left. Note that we are therefore representing the Ising model with periodic horizontal boundary conditions for spin variables (cf.\ \eqref{H_addb})  in terms of a model with anti-periodic horizontal boundary conditions for Grassmann variables. Ising models with anti-periodic horizontal boundary conditions, which can be associated with a negative coupling of the horizontal edges connecting the first and last column, have Grassmann representation with periodic horizontal boundary conditions.}; and 
\begin{equation} 
\cS_{\tilde{t}}(\Phi') :=  \sum_{k=1}^{m/2} \tilde{t}_k(\beta) V_{y_{2k} }V_{y_{2k-1} } \,,
\label{eq:cS_def_addb}
\end{equation}
with $\tilde{t} = \{\tilde{t}_k(\beta)\}_{k=1}^{m/2} = \{ \tanh \beta \tilde{J}_{k}\}_{k=1}^{m/2}$.
The $\bar{T},T$ variables appear only in the diagonal elements of \eqref{eq:cS_def_T} and they can be easily integrated out (see \cite[Appendix A]{GM05}), so we can rewrite Eq.~\eqref{PfA'} as
\begin{equation}
Z_{0,t;\L}( \tJJ)= (-1)^{LM}  C_\beta( J_1,J_2, \tJJ) \cdot \int  \mathcal D \Phi \, e^{\cS_{t}(\Phi) + \cS_{\tilde{t}}(\Phi)}\,,
\label{eq:Z_Grass_addb2}
\end{equation} 
where $\Phi=\{(\bar H_{z},H_{z},\bar V_{z},V_{z})\}_{z\in \L}$ is a collection of $4LM$ Grassmann variables, $\mathcal D \Phi$ is the Grassmann `differential'
\begin{equation}\label{diffPhi}\mathcal D \Phi=\prod_{z\in\L}d\bar H_{z} dH_{z}  d\bar V_{z} dV_{z}\,,\end{equation}
\begin{equation} 
 \cS_{t}(\Phi) := \sum_{z\in\L}(t_1(\beta) \bar H_{z} H_{z+\hat e_1}+t_2(\beta) \bar V_{z} V_{z+ \hat e_2}+\bar H_{z} H_{z}+\bar V_{z} V_{z}+\bar V_{z} \bar H_{z}+ V_{z}\bar H_{z}+ H_{z} \bar V_{z}+ V_{z} H_{z})\,,
	\label{cS_def2}  \end{equation}
with the same definitions given below \eqref{eq:cS_def_T} and, with a little abuse of notation, $\cS_{\tilde{t}}$ denotes the same function as in \eqref{eq:cS_def_addb} (which actually contained no $\bar{T},T$ variables). Note that the Grassmann integral in the r.h.s.\ of \eqref{eq:Z_Grass_addb2} depends on $\tJJ$ via the dependence of $\mathcal{S}_{\tilde{t}}$ on $\tilde{t}$ (cf.\ \eqref{eq:cS_def_addb}).

If we now plug Eq.~\eqref{eq:Z_Grass_addb2} into Eq.~\eqref{evspin}, we obtain the following expression for the multipoint boundary spin correlations: 
\begin{equation}\label{spin_corr_V}\begin{aligned}
&\braket{ \s_{y_1}\s_{y_2} \cdots \s_{y_{m-1}} \s_{y_{m}} }_{0,t;\L} =   \frac{\int \mathcal D \Phi \, e^{\cS_{t}(\Phi) } V_{y_1}V_{y_2} \cdots V_{y_{m-1}} V_{y_{m}}  }{\int \mathcal D \Phi \, e^{\cS_{t}(\Phi) }}, \end{aligned}
\end{equation}
where the r.h.s.\ is the expectation of the Grassmann monomial $ V_{y_1}V_{y_2} \cdots V_{y_{m-1}} V_{y_{m}}$ with respect to the \textit{Gaussian Grassmann measure} $\mathcal D \Phi e^{\cS_{t}(\Phi)}$ (this terminology is motivated by the presence of only quadratic Grassmann monomials in the exponent, see Eq.~\eqref{cS_def2}).

The Grassmann expectation in the r.h.s.\ of \eqref{spin_corr_V} can also be obtained on the original graph in the following way. To each $z \in \partial^l \L$ we associate an auxiliary Grassmann variable (or source field) $\varphi_z$ and we denote $\bs{\varphi} := {\set{\varphi_z}_{z \in \partial^l \L}}$: these $\bs{\varphi}$ variables anti-commute with each other as well as with the $\Phi$ variables. Let 
\begin{equation}\label{eq:Bs}
 \cB_\L(\Phi,\bs{\varphi}) := \sum_{z \in \partial^l\L} V_{z} \varphi_{z} \,\end{equation}
and let
\begin{equation}
	 \Xi_{0,t;\L} (\bs{\varphi}) := \int \cD \Phi \; e^{\cS_{t}(\Phi) 
		+ \cB_\L(\Phi,\bs{\varphi})} \,,
	\label{XiL0}
			\end{equation}
with $\cD \Phi$ as in \eqref{diffPhi} and $\cS_{t}(\Phi)$ as in Eq.~\eqref{cS_def2}. For $y_1,\ldots, y_m$ distinct boundary vertices, $m \in 2 \mathbb{N}$, we get
\begin{equation}  \left.  \frac{\partial^{m}}{\partial  \varphi_{y_1}\partial \varphi_{y_2} \cdots \partial  \varphi_{y_{m-1}}\partial \varphi_{y_{m}}}
\, \log   \Xi_{0,t;\L} (\bs{\varphi}) \right|_{\bs{\varphi}=\0} =  \frac{ \int \cD \Phi \,\,  e^{\cS_{t}(\Phi)} V_{y_1}V_{y_2} \ldots V_{y_{m-1}}V_{y_m} }{ \int \cD \Phi \,\,  e^{\cS_{t}(\Phi) }}\,,
\end{equation}
from which it is immediately clear that we can use $\Xi_{0,t;\L}(\bs{\varphi})$ as the generating function of the boundary spin correlations in \eqref{spin_corr_V}  (paying attention to the order in which the Grassmann variables appear: $V_y$, $\varphi_y$ and $\frac{\partial}{\partial \varphi_y}$ anti-commute with each other). Finally, we note that $\Xi_{0,t;\L}(\bs{\varphi})$ is a function of the auxiliary sources associated with the (boundary) vertices, and no longer involves the auxiliary edges; in this sense it is a representation for the model defined on the original graph $G_\L$ (recall that $G_\L$ and $\tilde{G}_\L$ have the same vertex set).

\subsection{Grassmann representation: non-integrable models.}\label{lambdane0}

In the following proposition we state the representation of the generating function for non-integrable correlations, i.e.\ for the model \eqref{eq:HM} with $\lambda \neq 0$. 
\begin{proposition}\label{GGMprop3}
There exists $\l_0>0$ such that, if $|\l|\le \l_0$, then for any  $\bs y=(y_1,\ldots, y_{m}) \in \partial^l\L$ right-to-left ordered sequence of distinct boundary vertices, $m \in 2 \mathbb{N}$, the multipoint boundary spin correlations can be expressed as
\begin{equation}\label{evspinQ}\begin{aligned}
&\braket{\s_{y_1} \cdots \s_{y_{m}} }_{\l,t;\L} = \left. \frac1{\Xi_{\l,t;\L} (\bs{0})}  \frac{\partial^{m}}{\partial  \varphi_{y_1} \cdots \partial  \varphi_{y_{m}}}
\,  \Xi_{\l,t;\L} (\bs{\varphi})  \right|_{\bs{\varphi}=\0}\,,\end{aligned}
\end{equation}
where $\Xi_{\l,t;\L} (\bs{\varphi})$ is the Grassmann generating functional
\begin{equation}
	 \Xi_{\l,t;\L} (\bs{\varphi}) := \int \cD \Phi \; e^{\cS_{t}(\Phi) 
		+ \cB_\L(\Phi,\bs{\varphi})+ \cV_\L^{\rm int}(\Phi)} \,,
	\label{XiL}
			\end{equation}
with $\cD \Phi$  defined as in \eqref{diffPhi}, $\cS_{t}(\Phi)$ as in \eqref{cS_def2},  $\cB_\L(\Phi,\bs{\varphi})$ as in \eqref{eq:Bs} 
 and
	\begin{equation}
\begin{split}	\cV_\L^{\rm int} (\Phi)
			&:= \sum_{\substack{ {B} \subset \mathfrak{B} _\L: \\ {B} \neq \emptyset}}
			W_\L^{{\rm int}} ({B})\prod_{e \in{B}} E_{e}
					\end{split}
		\label{eq:B_expansion_bis}
		\end{equation}
where \begin{itemize} \item if $e$ is an horizontal edge with endpoints $z,z+\hat e_1$, $E_e:=\bar H_z H_{z+\hat e_1}$; if $e$ is a vertical edge with endpoints $z,z+\hat e_2$, $E_e:=\bar V_z V_{z+\hat e_2}$; \item $W_\L^{\rm int}$ is a function that, for any $n \in \mathbb N$ and suitable positive constants $C,c,\kappa$, satisfies the following bound:
		\begin{equation}
			\sup_{\bar{e}\in \mathfrak{B}_\L}\sum_{\substack{{B} \subset \mathfrak{B}_\L:\,  \bar{e}\ni {B} \\ |{B}|=n}}|W_\L^{\rm int}({B})| e^{c \delta({B})}	\le
			C^{n}|\lambda|^{\max (1,\kappa n)}\,,
			\label{eq:B_base_decay}
		\end{equation}
		where $\delta({B})$, for ${B} \subset \mathfrak{B}_\L$, denotes the size of the smallest ${B}' \supseteq {B}$ which is  the edge set of a connected subgraph of $G_\L$. $W_\L^{\rm int}$, considered as a function of $\lambda$, $t_1$ and $t_2$, can be analytically continued to any complex $\lambda, t_1, t_2$ such that $|\lambda| \le \lambda_0$ and $|t_1|, |t_2|\in K'$, with $K'$ the same compact set introduced before the statement of Theorem \ref{prop:main}, and the analytic continuations satisfies the same bounds above. \end{itemize}\end{proposition}	
From now on we will refer to $\cB_\L (\Phi,\bs{\varphi})$, which is the only term dependent on $\bs{\varphi}$, as the \textit{boundary spin source} term, to $\cV_\L^{\rm int}(\Phi)$, which is the only term dependent on $\lambda$, as the \textit{interaction} term and to $W_\L^{\rm int}$ as the \textit{kernel} of the interaction. Moreover, $\cV_\L^{\rm int}(\Phi)$, as well as $W_\L^{\rm int}$, being independent of  $\bs{\varphi}$, are exactly the same functions introduced in \cite[Proposition 3.1]{AGG_AHP} (without auxiliary energy sources); in particular $W_\L^{\rm int}$ will satisfy the same properties listed after \cite[Proposition 3.1]{AGG_AHP} which for convenience are also reported below (see Remark~\ref{symmCyl} and the comments below Eq.~\eqref{inte}).  The most surprising property is that, unlike the related expressions for generating functionals of the energy correlations, the terms involving the auxiliary sources are $\lambda$-independent, and so have exactly the same form as in the integrable model.  As we shall see, this follows from the observation that the integrability-breaking term in \cref{eq:HM} can be written without the auxiliary edges added above, and so without the specific Grassmann variables appearing in $\cB_\Lambda$.
 
\medskip

\begin{proof}[Proof of Proposition~\ref{GGMprop3}] 
The proof of Proposition~\ref{GGMprop3} is a corollary of the proof of \cite[Proposition~3.1]{AGG_AHP} (and therefore of the proof of \cite[Proposition~1]{GGM}). In fact, the first step of the proof consists in a rewriting of the interaction potential in terms of edge variables: once it has been verified that this involves only the edges of the original graph (and not the additional edges), one can proceed exactly as described there. 

The starting point is the non-integrable analogue of \eqref{evspin}: 
\begin{equation}\label{evspinlambda}\begin{aligned}
&\media{\sigma_{y_1}\cdots\sigma_{y_{m}}}_{\lambda,t;\L} =  \left. \frac{\beta^{-m/2}}{ Z_{\lambda,t;\L}(\tJJ)}\frac{\partial^{m/2}}{\partial \tilde{J}_1 \cdots \partial \tilde{J}_{m/2}} Z_{\lambda,t;\L} ( \tJJ)  \right|_{ \tJJ =\0}\,,\end{aligned}
\end{equation}
where
\begin{equation}\label{Z_tJA}
Z_{\l,t;\L}( \tJJ) =\sum_{\s \in \Omega_\L} e^{-\beta \tilde{H}_{\L}(\s) + \beta \lambda \sum_{X\subset \L} V(X) \sigma_X}\,,
\end{equation}
with $\tilde{H}_\L(\s)$ as in \eqref{H_addb}, the $\l$-dependent term as in \eqref{eq:HM} and $\sigma_X:=\prod_{z\in X}\s_z$. Eq.~\eqref{Z_tJA} is then the partition function associated with the Hamiltonian of a non-integrable model similar to the one in \eqref{eq:HM} but defined on $\tilde{G}_\L$, i.e.\ the nearest-neighbour interactions include those between boundary vertices that are connected with additional edges.

As already noted in \cite[Proposition~3.1]{AGG_AHP}, any even interaction of the form $V(X) \sigma_X$ with $X  \subset \L$ can always be rewritten in terms of edge variables: now we note that it is always possible to consider only the variables associated with horizontal and/or vertical edges and avoid the variables associated with additional edges.
Let $e := \{ z, z+\hat{e}_i\} \in \mathfrak{B}_\L$ denote an edge of the original graph (horizontal or vertical depending on whether it is $i=1$ or $i=2$) and let $\sigma_e := \sigma_z \sigma_{z+\hat{e}_i} \in \{\pm 1\}$ denote the associated `edge spin'\footnote{The product of two spins placed at the endpoints of an edge is also called \textit{energy operator}: in \cite{AGG_AHP} it is denoted by $\epsilon_x$ ($x$ denotes $\{z, z+\hat{e}_i\}$), in \cite{GGM} it is denoted by $\varepsilon_b$ ($b$ denotes $\{z, z+\hat{e}_i\}$).}. If $X = \{z_1, \ldots, z_{2n}\}$, $\sigma_X = \sigma_{z_1}  \cdots \sigma_{z_{2n}}$ can always be rewritten in terms of $\prod_{e \in \mathcal{C}(X)} \sigma_e$ where $\mathcal{C}(X) \subset \mathfrak{B}_\L$ is a special path made up of edges that connects all the vertices of $X$. In Fig.~\ref{fig10} the simple case $|X|= 2$ is shown, which is the case of the pair interaction studied in \cite{GGM}: as discussed after \cite[Equation~(2.8)]{GGM}, we can write $$\sigma_X =\sigma_{z_1}\sigma_{z_2} = \frac12 \Big(\prod_{e \in\CC_U(z_{1},z_{2})} \sigma_e+\prod_{e\in\CC_D(z_{1},z_{2})}\sigma_e\Big)\,,$$
where $\CC_{U/D} (z_{1},z_{2}) \subset \mathfrak{B}_\L$ is the upmost/downmost shortest path connecting $z_{1}$ with $z_{2}$ (see Fig.~\ref{fig10}). In this example we consider an average between two edge paths to obtain a set of figures invariant under the basic symmetries on the cylinder, namely horizontal translations, and horizontal and vertical reflections; if $|X| > 2$ we can always obtain a similar set of more general figures, made of horizontal and vertical edges, which are invariant under the basic symmetries on the cylinder.

Here we want to emphasize that even if we consider the vertices at the endpoints of additional edges, we can always choose a path made only of horizontal and vertical edges, i.e.\ edges of the original graph. 

\begin{figure}[h]
\centering
\begin{tikzpicture}[scale=0.40]
\draw [help lines, dashed, gray] (-12,1) grid (4,9);
\draw[dashed, gray] (-13,8)--(-12,8);
\draw[dashed, gray] (-13,9)--(-12,9);
\draw[very thick, blue] (-6,6)--(-6,8)--(-11,8);
\draw[very thick, purple] (-6,6)--(-11,6)--(-11,8);
\node [above left] at (-10.5,8){${z'_2}$};
\node [above right, blue] at (-9,8){$C_U(z_1',{z'_2})$};
\node [below right] at (-6.5,6){${z'_1}$};
\node [below right,purple] at (-12,6){$C_D({z'_1},z_2')$};
\draw[very thick, blue] (-3,1)--(-3,3)--(2,3);
\draw[very thick, purple] (-3,1)--(2,1)--(2,3);
\node [above, blue] at (-2,3){$C_U{(z_1,z_2)}$};
\node [above, purple] at (-0.5,1){$C_D{(z_1,z_2)}$};
\node [below left] at (-3,1){${z_1}$};
\node [above right] at (2,3){${z_2}$};
\draw[dashed, gray] (-13,7)--(-12,7);
\draw[dashed, gray] (-13,6)--(-12,6);
\draw[dashed, gray] (-13,5)--(-12,5);
\draw[dashed, gray] (-13,4)--(-12,4);
\draw[dashed, gray] (-13,3)--(-12,3);
\draw[dashed, gray] (-13,2)--(-12,2);
\draw[dashed, gray] (-13,1)--(-12,1);
\draw[dashed, gray] (5,9)--(4,9);
\draw[dashed, gray] (5,8)--(4,8);
\draw[dashed, gray] (5,7)--(4,7);
\draw[dashed, gray] (5,6)--(4,6);
\draw[dashed, gray] (5,5)--(4,5);
\draw[dashed, gray] (5,4)--(4,4);
\draw[dashed, gray] (5,3)--(4,3);
\draw[dashed, gray] (5,2)--(4,2);
\draw[dashed, gray] (5,1)--(4,1);
\draw[dashed, gray] (-12,1) to  [out= -75, in =-75] (-10,1) ;
\draw [dashed, gray] (-9,1) to [out= -75, in =-75] (-7,1);
\draw [dashed, gray](-3,1) to [out= -75, in =-75] (2,1);
\draw [dashed, gray] (3,1) to [out= -75, in =-75] (4,1);
\end{tikzpicture}
\caption{An example of pair interaction: paths $C_U$ and $C_D$ connecting two sites, whether they are far from the boundary ($z_1'$ and $z_2'$) or on the boundary ($z_1$ and $z_2$), can always be formed only by horizontal and vertical edges (avoiding involving auxiliary edges). Similarly, for generic even interactions, one can always consider figures consisting only of horizontal and vertical edges that respect symmetry properties on the cylinder.}\label{fig10}
\end{figure}
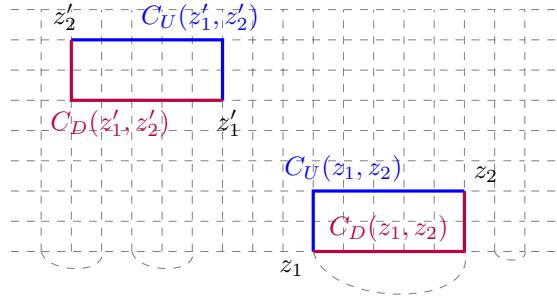
Since $\sigma_X = \prod_{e \in \mathcal{C}(X)} \sigma_e$ is always equal to $\pm1$, we can use $e^{\pm a} = \cosh a [1 \pm \tanh a]$ to write
\begin{equation}\prod_{X\subset \L} e^{\beta \lambda V(X) \sigma_X} =  \prod_{X\subset \L} \cosh{(\beta \lambda V(X))} \prod_{X\subset \L}\Big[1+ \tanh(\beta \lambda V(X)) \prod_{e \in \mathcal{C}(X)} \sigma_e\Big]\,. \label{VG}\end{equation}
Expanding the second product in the r.h.s.\ we obtain a sum of terms: each term contains a product $\prod_{e \in  \mathcal{F}_n }\sigma_e$, where $\mathcal{F}_n \subset \mathfrak{B}_\L$ is a figure composed of $n \ge 0$ edge paths (if $n = 0$, $\prod_{e \in  \mathcal{F}_0 }\sigma_e=1$). Note that since $\sigma_e^2 = 1$, in each term there is a product of edge spin variables associated with distinct edges. 
Furthermore the correspondence in \cite[Eq.~(2.7)]{GGM} remains valid even on the modified graph $\tilde{G}_\L$ and even if the energy sources are set equal to zero: namely, for a set $\mathcal{F}_n \subset \mathfrak{B}_\L$ of distinct edges, we get
$$ \sum_{\s \in \Omega_\L}  e^{-\beta \tilde{H}_\L} \prod_{e \in \mathcal{F}_n} \sigma_{e}
 = (-1)^{LM}C_\beta(J_1,J_2, \tJJ)\cdot \int \cD \Phi \,\,  e^{\cS_{t}(\Phi) + \cS_{\tilde{t}}(\Phi)} \, \prod_{e \in \mathcal{F}_n}  \big[t_{i}(\beta)+(1-t_{i}^2(\beta)) E_{e}\big]$$
where in the r.h.s.\ we use the same notation used for \eqref{eq:Z_Grass_addb2} and if $e$ is an horizontal (resp.\ vertical) edge $i = 1$ (resp.\ $i = 2$)  and $E_e = \bar{H}_z H_{z+\hat{e}_1}$ (resp.\ $E_e= \bar{V}_z V_{z+\hat{e}_2})$. We can therefore refer to the representation obtained in \cite{GGM} and \cite{AGG_AHP}, which in the modified graph becomes
\begin{equation}\label{ZV}
 Z_{\l,t;\L}( \tJJ) \propto C_\beta( J_1,J_2, \tJJ)  \cdot  \int \cD \Phi \,\,  e^{\cS_{t}(\Phi) + \cS_{\tilde{t}}(\Phi)+\mathcal V_\L^{\rm int}(\Phi)}, \end{equation}
where $\propto$ means `up to a multiplicative constant independent of $\tJJ$' and $\mathcal V_\L^{\rm int}(\Phi)$ is the {\it same} interaction potential as \cite[Eq.~(2.5)]{GGM} (cf.\ \cite[Eq.~(3.2)]{AGG_AHP}). Then, if we now plug \eqref{ZV} into \eqref{evspinlambda}, we get  the non-integrable analogue of \eqref{spin_corr_V}, namely
\begin{equation}\label{eqq}\begin{split}
\media{\sigma_{y_1}\sigma_{y_2} \cdots \sigma_{y_{m-1}}\sigma_{y_m}}_{\lambda, t ;\L} &=  \frac{ \int \cD \Phi \,\,  e^{\cS_{t}(\Phi) +\mathcal V_\L^{\rm int}(\Phi)} V_{y_1}V_{y_2} \ldots V_{y_{m-1}}V_{y_m} }{ \int \cD \Phi \,\,  e^{\cS_{t}(\Phi) +\mathcal V_\L^{\rm int}(\Phi)}} \end{split} \end{equation}
where the r.h.s.\ is the expectation of $V_{y_1}V_{y_2} \ldots V_{y_{m-1}}V_{y_m} $ with respect to the \textit{non-Gaussian Grassmann measure} $\mathcal D \Phi \, e^{\cS_{t}(\Phi) + \mathcal V_\L^{\rm int}(\Phi)}$ (in fact $\mathcal V_\L^{\rm int}(\Phi)$ contains not only quadratic Grassmann monomials). Since $\mathcal V_\L^{\rm int}(\Phi)$ is independent of $\tilde{\bs{J}}$, the derivatives with respect to the auxiliary weights $\tJJ$ only act on $\cS_{\tilde{t}}(\Phi)$ so we can proceed exactly as below \eqref{spin_corr_V} and easily verify that the Grassmann expectation in the r.h.s.\ of \eqref{eqq} can be obtained by \eqref{evspinQ}. \end{proof}

\subsection{Grassmann representation: massive and massless variables.}\label{xiphi}
Before we start the multiscale computation, let us make a final rewriting of the Grassmann 
generating functional $\Xi_{\l,t;\L} (\bs{\varphi})$. As discussed in \cite{AGG_AHP, AGG_CMP}, at criticality the interaction term has, among others, the effect of modifying 
(`renormalizing') the large distance behavior of the bare propagator (covariance of the Gaussian Grassmann measure) by effectively rescaling it by $1/Z$ (where $Z=Z(\lambda)$ plays the role of the `wave function renormalization', in a QFT analogy), and by changing the unperturbed critical temperature $\beta_c|_{\lambda=0}$ and the bare couplings 
into `dressed' ones, $\beta_c(\lambda)$ and $t^*:= (t_1^*,t_2^*)$ respectively (here $t_1^*=t_1^*(\lambda)$ and $t_2^*=\frac{1-t_1^*}{1+t_1^*}$).
In order to take this effect into account, it is convenient to write the Grassmann generating function in terms of a reference Gaussian Grassmann integration with dressed parameters  $Z(\lambda),\beta_c(\lambda), (t_1^*(\lambda),t_2^*(\lambda))$ with $Z,\beta_c,t_1^*$ the same {\it analytic} functions of $\lambda$ used in \cite{AGG_AHP,AGG_CMP}, see in particular \cite[Proposition 4.11]{AGG_AHP}.
Choosing the same values of $Z,\beta_c,t_1^*$ as in \cite{AGG_AHP,AGG_CMP} works because these values do not depend upon the presence of the external sources, 
and the $\varphi$-independent part of the action at exponent, under the integral sign in the right hand side of Eq.~\eqref{XiL}, that is $\mathcal S_t(\Phi)+
\mathcal V^{\text{int}}_\Lambda(\Phi)$, has the same expression as the $\boldsymbol{A}$-independent part of the action in \cite[Eq.~(3.1)]{AGG_AHP} (denoted by 
$\mathcal S_{t_1,t_2}(\Phi)+\mathcal V(\Phi,\boldsymbol{0})$ there). 

Therefore, in order to set up the multiscale analysis, we follow the procedure described in \cite[Section~3]{AGG_AHP} taking into account the presence of boundary spin sources. We perform a change of variables with an explicit invertible linear transformation, see \cite[Section~2.1.2]{AGG_AHP}, passing from collection of Grassmann variables $\{(\bar H_{z},H_{z},\bar V_{z},V_{z})\}_{z\in \L}$ to two collections of Grassmann  variables $\xi = \set{ (\xi_{+,z},\xi_{-,z}) }_{z \in \L}$ and $\phi = \set{ (\phi_{+,z}, \phi_{-,z})}_{ z \in \L}$, which are called `massive' and `massless' variables. 
In terms of the new variables the quadratic action in Eq.~\eqref{cS_def2} separates into a part that depends only on $\xi$ and one that depends only on $\phi$, namely 
$\cS_{t}( \xi,\phi) = \cS_{m}(\xi) +\cS_{c}(\phi)$, cf.\ \cite[Equations~(2.1.12)-(2.1.13)]{AGG_AHP} and preceding lines. Next we rescale the Grassmann variables as 
$(\xi,\phi)\to Z^{-1/2}(\xi,\phi)$, and we isolate from $\cS_{t}( \xi,\phi)$ a `dressed' quadratic part $\cS^*(\xi,\phi)=\cS_{m}^*(\xi) +\cS_{c}^*(\phi)\equiv \frac12(\xi,A^*_m\xi)+\frac12(\phi,A^*_c\phi)$, where $\cS_{m}^*(\xi)$ and $\cS_{c}^*(\phi)$ are defined via the same expressions as $\cS_{m}(\xi)$ and $\cS_{c}(\phi)$, respectively, with $(t_1,t_2)$ 
replaced by $(t_1^*,t_2^*)$. We thus obtain the analogue of \cite[Eq.(3.22)-(3.23)]{AGG_AHP}, that is 
 \begin{equation} 
	\Xi_{\l,t;\L}(\bs{\varphi})\propto \int  P^*_c (\cD\phi) P^*_m (\cD\xi) \, e^{\cV_\L^{(1)}(\xi,\phi)+ \cB_\L(\phi,\bs{\varphi})}, \label{eq:startfrom}\end{equation}
where $P^*_c (\cD\phi)=\mathcal D\phi\, e^{\cS^*_c(\phi)}/{\textrm{Pf}(A^*_c)}$, $P^*_m (\cD\xi)=\mathcal D\xi\, e^{\cS^*_m(\xi)}/{\textrm{Pf}(A^*_m)}$, and 
\begin{equation}\cV_\L^{(1)}(\xi,\phi) :=  Z^{-1}\cS_{t}(\xi,\phi) - \cS_{t^*}(\xi,\phi)
+ \widetilde{\cV}_\L^{\rm int}(Z^{-1/2}\xi,Z^{-1/2}\phi),\label{qualequesta}\end{equation}
where  by 
$\widetilde{\cV}_\L^{\rm int}(\xi,\phi)$ we denote $\mathcal V^{\text{int}}_\Lambda(\Phi)$, expressed in terms of the new variables $(\xi,\phi)$ (note that 
$\mathcal V^{(1)}_\Lambda(\xi,\phi)$ in \eqref{eq:startfrom} is the same as the function $\mathcal V^{(1)}(\xi,\phi,\boldsymbol{0})$ in \cite[Eq.~(3.22)]{AGG_AHP}).
Moreover, \begin{equation}\label{eq:BsPhi}
 \cB_\L(\phi,\bs{\varphi}) := Z^{-1/2}\sum_{z \in \partial^l\L} \phi_{-,z} \varphi_{z}\,.\end{equation} 
\begin{remark}[Reflection symmetries]\label{symmCyl}
The quadratic actions $\cS_{t}(\xi)$ and $\cS_{t}(\phi)$ are invariant under the following horizontal and vertical reflections: $\xi_{\omega,z} \to i \xi_{-\omega, (-(z)_1,(z)_2)}$, $\phi_{\omega,z} \to   i\omega \phi_{\omega, (-(z)_1,(z)_2)}$, and $\xi_{\omega,z} \to -i \omega$ $\xi_{\omega, ((z)_1,M+1-(z)_2)}$, $\phi_{\omega,z} \to   i \phi_{-\omega, ((z)_1,M+1-(z)_2)}$ (see \cite[Section~2.4]{AGG_AHP}); the boundary spin source term in Eq.~\eqref{eq:BsPhi}, where only the $\phi_{-, z}$ field appears and $(z)_2=1$, is invariant under the following horizontal reflections: $\phi_{-,((z)_1,1)} \to   i\omega \phi_{-, (-(z)_1,1)}$ and $ \varphi_{((z)_1,1)} \to i \varphi_{(-(z)_1,1)}$.\end{remark}

For later purposes, we will also need to compute the averages of arbitrary monomials in the massive and critical Grassmann variables as $\int  P_\#^*(\cD f) f_{\omega_1,z_1} \ldots f_{\omega_n,z_n} = \Pf G_\#$ where $\#\in\{m,c\}$ and, if $n$ is even, $G_\#$ is a $n\times n$ matrix with entries $[G_\#]_{jk} = \int  P_\#^*(\cD f) f_{\omega_j,z_j} f_{\omega_k,z_k}  = -[(A^*_\#)^{-1}]_{(\omega_j,z_j),(\omega_k,z_k)}$ (if $n$ is odd the averages are zero). The covariance of the Gaussian Grassmann measure $P_m^*(\cD \xi)$ (resp. $P_c^*(\cD \phi)$), i.e.\ the  two-point function $\int  P_m^*(\cD \xi)  \xi_{\omega,z} \xi_{\omega',z'}$ (resp. $\int  P_c^*(\cD \phi)  \phi_{\omega,z} \phi_{\omega',z'}$), is the $(\omega,\omega')$ element of the \textit{massive propagator}  $\mathfrak{g}_{\L,m}$ (resp. \textit{critical propagator}  $\mathfrak{g}_{\L,c}$). The explicit expression of the propagators has been computed in \cite[Section~2.1.2]{AGG_AHP} via Fourier diagonalization of $A^*_\xi$ and $A^*_\phi$, see \cite[Eqs.~(2.1.16),(2.1.19)]{AGG_AHP} or \cite[Eqs.~(2.1.13),(2.1.15)]{AGG_CMP}\footnote{Note that in \cite{AGG_AHP,AGG_CMP} to which we refer, the notation for the cylindrical propagator is without subscript $\L$.}. We use the same decompositions as \cite[Section~2.2.1]{AGG_AHP} and the same estimates as \cite[Section~2.2.2]{AGG_AHP}: in Section~\ref{propagatori} below, we recall the corresponding expressions in the half-plane limit. 

Moreover, we express $\cV_\L^{(1)}(\xi,\phi)$ of Eq.~\eqref{eq:startfrom} as in \cite[Eq.~(3.24)]{AGG_AHP} (cf.\ \cite[Eq.~(2.2.18)]{AGG_CMP}), namely 
\begin{equation}\label{inte}
\cV_\L^{(1)}(\xi,\phi) = \sum_{n \in 2 \mathbb{N}} \sum_{(\bs{\omega},\bs{z})\in \{\pm,\pm i\}^n\times \L^n}W_\L^{(1)}(\bs{\omega},\bs{z}) \phi (\bs{\omega},\bs{z}) \,,
\end{equation} where $(\bs{\omega},\bs{z}) =((\omega_1, z_1) \ldots,(\omega_{n}, z_{n}))$, $\phi (\bs{\omega},\bs{z}) :=\prod_{k=1}^{n} \phi_{\omega_k,z_k}$ and  $\phi_{\pm i,z}$ denotes $\xi_{\pm, z}$; we assume that $W_\L^{(1)}(\bs{\omega},\bs{z})$, the \textit{initial scale} kernel function of the interaction, is antisymmetric under simultaneous permutations of $\bs{\omega}$ and $\bs{z}$, invariant under horizontal translations of $\bs{z}$ (with anti-periodic horizontal boundary conditions) and invariant under the reflection symmetries induced by the substitutions that leave the quadratic actions unchanged mentioned in Remark~\ref{symmCyl}. The interaction kernel in \eqref{inte} satisfies the same estimate as \cite[Eq.~(3.15)]{AGG_AHP} and the same bulk-edge decomposition as \cite[Eq.~(3.17)]{AGG_AHP}: in Remark~\ref{Wint} below, we recall the corresponding expressions adapted to the half-plane limit given by
\begin{equation}\label{W1HH}
W_\bH^{(1)}(\bs{\omega}, \bs{z}):= \lim_{L, M \to \infty}  W^{(1)}_\L( \bs{\omega}, \bs{z})\,.\end{equation}
Similarly, we represent also $\cB_\L(\phi,\bs{\varphi})$ (cf.\ Eq.~\eqref{eq:BsPhi}) as
\begin{equation}\label{BL}
\cB_\L(\phi,\bs{\varphi}) = \sum_{(\omega, z) \in \set{\pm}\times\partial^l\L}B^{(0)}_\L ((\omega, z), y) \phi_{\omega,z} \varphi_y\,,
\end{equation}
where $B^{(0)}_\L ((\omega, z), y):= Z^{-\frac12} \delta_{\omega,-}\delta_{z,y}$ is the \textit{initial scale} kernel function of the boundary spin source. Note that what we call the \textit{initial scale} kernel is associated with superscript $(1)$ when referring to the interaction kernel while it is associated with superscript $(0)$ when referring to the boundary spin source kernel: the reason for this apparent inconsistency will become clear shortly, at the beginning of Section \ref{sec:3}. In the half-plane limit we let 
\begin{equation}\label{BH0}B^{(0)}_\bH ((\omega, z), y) := \lim_{L, M \to \infty} B^{(0)}_\L ((\omega, z), y)\,;
\end{equation} 
note that $|B^{(0)}_\bH ((\omega, z), y)| = |B^{(0)}_\L ((\omega, z), y)| = |Z^{-\frac12}|\le C$.

\begin{remark}[Half-plane symmetries] \label{item:symm} We assume $W_\bH^{(1)}(\bs{\omega}, \bs{z})$ antisymmetric under simultaneous permutations of $\bs{\omega}$ and $\bs{z}$, invariant under horizontal translations of $\bs{z}$ and invariant under the reflection symmetries induced by the horizontal reflections $\phi_{\omega,z} \to \Theta_1\phi_{\omega,z}:=  i\omega \phi_{\omega, (-(z)_1,(z)_2)}$ and $\xi_{\omega,z}  \to i \xi_{-\omega, (-(z)_1,(z)_2)}$. Finally, we assume $B^{(0)}_\bH ((\omega, z), y)$  invariant under the reflection symmetries induced by $\phi_{\omega,z} \to \Theta_1\phi_{\omega,z}$ and $\varphi_{z}  \to i \varphi_{(-(z)_1,1)}$. %
\end{remark}

\subsubsection{Propagators: decompositions and dimensional bounds.}\label{propagatori}

\textit{The multiscale decomposition.} With the aim of performing a multiscale analysis in the following sections, we let $\mathfrak g_\L^{(1)} := \fg_{\L,m}$ denote the massive propagator as in \cite[Eq.(2.1.16)]{AGG_AHP}), whose explicit expression in the half-plane limit is
\begin{equation}\label{gH1}
\begin{aligned}
\mathfrak{g}^{(1)}_{\bH}(z,z') & := \lim_{L,M \to \infty}  \mathfrak{g}^{(1)}_{\L}(z ,z')  \\
&  =  \delta_{(z)_2,(z')_2}\frac{1}{2\pi} \int_{[-\pi,\pi]} dk_1 e^{-ik_1((z)_1-(z')_1)}  \begin{pmatrix}  0 & (1+t^*_1 e^{ik_1})^{-1} \\ - ( 1+t^*_1 e^{-ik_1})^{-1} & 0 \end{pmatrix}\,,\end{aligned}\end{equation}
and we introduce the \textit{multiscale decomposition} of the critical propagator $\fg_{\L,c}$ (cf.\ \cite[Section~(2.2.1)]{AGG_AHP}):  letting $h^*:= -\lfloor \log_2(\min \{L,M\})\rfloor$, for any $h^* \le h <0$ we write
\begin{equation} 
\fg_{\L,c}(z,z')=\mathfrak g_\L^{(\le h)}(z, z')+\sum_{j=h+1}^0\mathfrak g_\L^{(j)}(z, z'),\label{gch*h}
\end{equation}
where \begin{equation}\label{g0h}
\begin{aligned}
\fg_\L^{(0)} (z,z')& := \int_0^1 \fg_{\L}^{[\eta]} (z,z') d\eta\,,\\
\fg_\L^{(h)} (z,z')& := \int_{2^{-2h-2}}^{2^{-2h}} \fg_{\L}^{[\eta]} (z,z') d\eta\qquad \mbox{ if } h^*<h < 0\,,\\
\fg_\L^{(\le h)} (z,z')& := \int_{2^{-2h-2}}^{\infty} \fg_{\L}^{[\eta]} (z,z') d\eta\,,
\end{aligned}\end{equation} 
and $\fg_{\L}^{[\eta]}$ is the \textit{cutoff propagator} as in \cite[Eq.~(2.2.7)]{AGG_AHP}, whose explicit expression in the half-plane limit is
\begin{equation}\begin{aligned}\label{gceta}
\mathfrak{g}^{[\eta]}_{\bH} (z,z') & := \lim_{L,M \to \infty}  \mathfrak{g}^{[\eta]}_{\L}(z +(\lfloor L/2\rfloor, 0),z'+(\lfloor L/2\rfloor, 0))  : = \begin{pmatrix} g^{[\eta]}_{ \bH,++} (z,z') &g^{[\eta]}_{\bH, +-} (z,z') \\ g^{[\eta]}_{\bH, -+} (z,z') &g^{[\eta]}_{ \bH,--} (z,z')  \end{pmatrix} \\ 
&  :=  \frac{1}{(2\pi)^2} \int_{[-\pi,\pi]^2} dk_1dk_2  e^{-ik_1((z)_1-(z')_1)} e^{-\eta D(k_1,k_2)} \\
 & \times \left\{e^{-ik_2((z)_2-(z')_2)}  \begin{pmatrix}  -2it^*_1\sin{k_1} & -(1-(t^*_1)^2)(1-B(k_1)e^{-ik_2}) \\ (1-(t^*_1)^2)(1-B(k_1)e^{ik_2}) &2it^*_1\sin{k_1} \end{pmatrix} \right.\\
& \left. -  \,e^{-ik_2((z)_2+(z')_2)}  \begin{pmatrix}-i2t^*_1\sin{k_1}&  -(1-(t^*_1)^2)(1-B(k_1)e^{ik_2}) \\  (1-(t^*_1)^2)(1-B(k_1)e^{ik_2})& \frac{1-B(k_1)e^{-ik_2}}{1-B(k_1)e^{ik_2}}2it^*_1\sin{k_1}\end{pmatrix}  \right\}\,,
 \end{aligned}
\end{equation}
where 
\begin{equation}\nonumber\begin{aligned}
&D(k_1,k_2) = 2(1-t^*_2)^2(1-\cos{k_1})+2(1-t^*_1)^2(1-\cos{k_2})\,, \qquad B(k_1) = \frac{|1+t^*_1 e^{ik_1}|^2}{(1+t^*_1)^2} \,.
 \end{aligned}
\end{equation}
Note that, for the critical case we can use $t^*_2 = \frac{1-t^*_1}{1+t^*_1}$ so that $B(k_1)$ reported here coincides with $B(k_1)$ in \cite[Eq.~(2.1.22)]{AGG_AHP}; moreover,  $\frac{1-B(k_1)e^{-ik_2}}{1-B(k_1)e^{ik_2}}$ in the last line of \eqref{gceta} coincides with $e^{2(M+1){k_2}}$ in the last line of \cite[Eq.~(2.1.19)]{AGG_AHP} (to verify this just use \cite[Eq.~(2.1.23)]{AGG_AHP}) therefore the expressions in curly brackets in \eqref{gceta} and in  \cite[Eq.~(2.1.19)]{AGG_AHP} coincide.
Combining \eqref{gceta} with \eqref{g0h} and \eqref{gch*h}, we get, for any $h < 0$, $ \fg_{\bH,c}(z,z')=\mathfrak g_\bH^{(\le h)}(z, z')+\sum_{j=h+1}^0\mathfrak g_\bH^{(j)}(z, z')$.

\textit{The cancellation property.} In view of \cite[Remark~2.2]{AGG_AHP}, the definition in Eq.~\eqref{gceta} can be extended to all $z, z' \in \mathbb{R}^2$ so that, using the symmetries mentioned in Remark~\ref{item:symm}, and letting $\bar{z} := ((z)_1, 0)$ be the projection of $z$ on the row at height $0$ (i.e.\ immediately below the lower boundary), we see that some components of $\fg_{\bH,c}(z,z')$ respect the following \textit{cancellation property}:
\begin{equation}
g_{\bH,c,++} (\bar{z}, z') = g_{\bH,c,++} ( z,\bar{z}')=g_{\bH,c,+-} (\bar{z}, z') =g_{\bH,c,-+} ( z,\bar{z}')=0\,, 
\label{gc0}
\end{equation}
and analogously for the same components of $\fg_\bH^{(h)}$ and $\fg_\bH^{(\le h)}$ (cf.\ \cite[Eq.~(2.2.8)]{AGG_AHP}). As in \cite[Remark~2.3]{AGG_CMP}, we can introduce, for any $z \in \bH$, two fictitious Grassmann variables $\phi_{\omega,\bar{z}}$, $\omega \in \set{\pm}$: if $\omega = +$,  $\phi_{+,\bar{z}}$ is identically zero in the sense that the two-point functions $g_{\bH,c,+\omega'} (\bar{z}, z')$ and  $g_{\bH,c,\omega+} ({z}, \bar{z}')$ vanish by \eqref{gc0}.

\textit{Dimensional bounds.} The propagators $\fg_\bH^{(0)}$, $\fg_\bH^{(h)}$ and $\fg_\bH^{(\le h)}$, for any $h < 0$, as well as the propagator $\mathfrak g_\bH^{(1)}$, satisfy the dimensional estimates stated in \cite[Proposition~2.3]{AGG_AHP} which remain valid for the infinite volume limits (see \cite[Remark~2.6]{AGG_AHP}), namely, for any $\bs{r} = (r_{1,1}, r_{1,2}, r_{2,1}, r_{2,2}) \in \mathbb{Z}^4_+$ and any $z,z' \in \bH$,
\begin{equation}\begin{aligned}\label{gdecay}
& || \bs{\partial}^\bs{r}\fg_\bH^{(h)}(z,z')|| \le C^{1+ |\bs{r}|_1} \bs{r}! 2^{(1+|\bs{r}|_1)h } e^{-c_02^h \|z-z'\|_1} \\
&|| \bs{\partial}^\bs{r}\fg_\bH^{(\le h)}(z,z')|| \le C^{1+ |\bs{r}|_1} \bs{r}! 2^{(1+|\bs{r}|_1)h } \end{aligned}\end{equation}
where in the l.h.s.\ the matrix norm is the max norm, i.e.\ the maximum over the matrix elements, $ \bs{\partial}^\bs{r} : = \prod_{i,j = 1}^2 \partial_{i,j}^{r_{i,j}}$, $\bs{r}! = \prod_{i,j=1}^2 r_{i,j}!$ and $\partial_{1,j}$  (resp.\ $\partial_{2,j}$) is the discrete derivative in direction $j$ with respect to the first (second) argument.

\textit{The bulk-edge decomposition.} We let 
 \begin{equation} \label{ginfty} \mathfrak{g}^{[\eta]}_\infty (z,z') := \lim_{M \to \infty} \mathfrak{g}^{[\eta]}_{\bH} (z+(0, \lfloor M\rfloor),z'+(0, \lfloor M\rfloor))\,,\end{equation}
be the infinite cutoff propagator, whit explicit expression given by \cite[Eq.~(2.2.9)]{AGG_AHP}: note that it is the same expression in the r.h.s.\ of Eq.~\eqref{gceta} with only the first term in curly brackets. Then, we can write Eq.~\eqref{gceta} as $\mathfrak{g}^{[\eta]}_{\bH} := \mathfrak{g}^{[\eta]}_\infty+ (\mathfrak{g}^{[\eta]}_{\bH}-\mathfrak{g}^{[\eta]}_{\infty})\equiv 
\mathfrak{g}^{[\eta]}_{\bH}+\mathfrak{g}^{[\eta]}_{\bH,E}$, we can define $\fg_\infty^{(h)}$ and $\fg_{\bH,E}^{(h)}$ via the analogues of \eqref{g0h} and we get, for any $h <0$, $\fg_{\bH,c}(z,z') = \fg^{(\le h)}_\bH (z,z')+ \sum_{j=h+1}^0 (\fg_\infty^{(j)}(z,z') +\fg_{\bH,E}^{(j)} (z,z'))$: in analogy with \cite[Section~2.2.1]{AGG_AHP} we call it \textit{bulk-edge} decomposition (even if here it is actually a decomposition into an infinite plane propagator and an edge propagator). The infinite plane propagators satisfy the same estimates as \eqref{gdecay} (cf.\ \cite[Remark~2.6]{AGG_AHP}, the edge propagators satisfy similar estimates provided that instead of $\|z-z'\|_1$ in the first line of \eqref{gdecay} we consider $|(z)_1-(z')_1|+ |(z)_2+(z')_2|$ (cf.\ \cite[Eq.~(2.2.16)]{AGG_AHP}). 

To prove these estimates it is sufficient to retrace the proof in \cite[Appendix~B]{AGG_AHP}: taking into account the comments below \eqref{gceta}, it is easy to verify that $\bs{\partial}^\bs{r}\fg^{(h)}_{\bH,E}$ corresponds to the contribution in the first line in square brackets of \cite[Eq.~(B.21)]{AGG_AHP}.

\section{Multiscale expansion.}\label{sec:3}

In this section, we will show that for every $J_1/J_2 \in K$, $t\in K'$, with $K,K'$ the compact sets introduced before the statement of Theorem \ref{prop:main}, 
$|\lambda|$ sufficiently small, and the appropriate choice of $Z$, $\beta_c$, $t_1^*$,  the derivatives of $\log \Xi_{\l,t;\L}(\bs{\varphi})$ of order $m$, 
with no repetitions, at $\bs{\varphi}=\bs 0$  (and then the $m$-point boundary spin correlations, see Eq.~\eqref{evspinQ}) admit an 
expansion as a uniformly convergent sum. Such an expansion is based on the following iterative evaluation of $\Xi_{\l,t;\L}(\bs{\varphi})$: 
starting from Eq.~\eqref{eq:startfrom}, we first define
\begin{equation}
	e^{\cV_\L^{(0)}(\phi)}\propto 
    \int P_m^* (\cD\xi) \, e^{ \cV_\L^{(1)}(\xi,\phi)},   
	\label{eq:Vcyl_N}
\end{equation}
where $\cV_\L^{(0)}(\phi)$ is the \textit{effective interaction on scale $0$} and $\propto$ means `up to a multiplicative constant independent of $\bs{\varphi}$'\footnote{Note that $\bs{\varphi}$ does not appear in \eqref{eq:Vcyl_N}: except for a slight difference in notation, it is the same expression in \cite[Equation~(3.1)]{AGG_CMP} with the auxiliary energy sources set to zero.};
then, we can write Eq.~\eqref{eq:startfrom} as
\begin{equation}
	\Xi_{\l,t;\L}( \bs{\varphi})\propto	
	\int  P_c^* (\cD\phi) e^{\cV_\L^{(0)}(\phi) + \cB_\L(\phi,\bs{\varphi}) }\,.
	\label{eq:after_step_1}
\end{equation}
Corresponding to the multiscale decomposition of critical propagator in Eqs.~\eqref{gch*h}-\eqref{g0h}, we introduce $P^{(\le h)}_\L(\cD\phi)$ and $P^{(h)}_\L(\cD\phi)$ which are the Gaussian Grassmann measures with covariance $\fg_\L^{(\le h)}$ and $\fg_\L^{(h)}$ respectively; we also introduce the massless field multiscale decomposition as $\phi = \phi^{(\le h)} + \sum_{j=h+1}^0 \phi^{(j)}$, $h^*+1 \le h < 0$. Then, in light of the addition formula for Grassmann integrals (see e.g.~\cite[Proposition 1]{GMT17}), we can rewrite the r.h.s.\ of Eq.~\eqref{eq:after_step_1} as
\begin{equation}\label{intms}\begin{aligned}
 \int  P_\L^{(\le -1)} (\cD\phi^{(\le -1)} )   \int  P_\L^{(0)} (\cD\phi^{(0)} ) e^{\cV_\L^{(0)}(  \phi^{(\le -1)}+\phi^{(0)}) + \cB_\L( \phi^{(\le -1)}+\phi^{(0)},\bs{\varphi} )}\,,
 \end{aligned}
\end{equation} 
integrate out  $\phi^{(0)}$ and iteratively use the same procedure to perform a step-by-step integration; at each step we rearrange the result of the integration so as to iteratively define sequences of functions $\cW_\L^{(h-1)}$, $\cV_\L^{(h-1)}$ and $\cB_\L^{(h-1)}$, for all $h \le 0$, via
\begin{equation}\begin{aligned}
	&e^{
		\cW_\L^{(h-1)}( \bs{\varphi})
		+ 
		\cV_\L^{(h-1)}(\phi^{(\le h-1)})
		+
		\cB_\L^{(h-1)}(\phi^{(\le h-1)},\bs{\varphi})
	}\\
	&\qquad\qquad\qquad
	\propto
	\int P_\L^{(h)}(\cD {\phi^{(h)}})
	e^{
		\cV_\L^{(h)}(\phi^{(\le h-1)}+{\phi^{(h)}})
		+
		\cB_\L^{(h)}(\phi^{(\le h-1)}+{\phi^{(h)}} ,\bs{\varphi})}\,;
	\label{eq:eff_inter_iter2} \end{aligned}
\end{equation}
that is, after integrating out $\phi^{(h)}$, the result is rearranged so that the \textit{single-scale contribution} to the generating function $\cW_\L^{(h-1)}( \bs{\varphi})$ depends only on $\bs{\varphi}$, the \textit{effective interaction} function  $\cV_\L^{(h-1)}(\phi^{(\le h-1)})$ depends only on ${\phi}^{(\le h-1)}$ and the \textit{effective  boundary spin source} function $\cB_\L^{(h-1)}(\phi^{(\le h-1)},\bs{\varphi})$ depends on both ${\phi}^{(\le h-1)}$ and $\bs{\varphi}$. At each step, $\cW_\L^{(h-1)}$,  $\cV_\L^{(h-1)}$ and $\cB_\L^{(h-1)}$ are fixed in such way that $\cW_\L^{(h)}(\bs 0) = \cV_\L^{(h)}(0) = \cB_\L^{(h)}(0,\0) =0$.

Note that $h =0$ is the first scale on which the $\bs{\varphi}$-sources appear: therefore if $h =0$, in the r.h.s.\ of \eqref{eq:eff_inter_iter2} we let  $\cB_\L^{(0)} := \cB_\L$ as in \eqref{BL} and $\cW_\L^{(0)}:=0$ (the $\cW_\L^{(h)}$ functions exist starting from the scale $h =- 1$).
The iteration continues until the scale $h^* = -\lfloor \log_2(\min\set{L,M})\rfloor$ is reached,
at which point we let 
\begin{equation}\begin{aligned}
	e^{
		\cW_\L^{(h^*-1)}( \bs{\varphi})}
	\propto
	\int P_\L^{(\le h^*)}(\cD {\phi^{(\le h^*)}})
	e^{
		\cV_\L^{(h^*)}(\phi^{(\le h^*)})
		+
		\cB_\L^{(h^*)}(\phi^{(\le h^*)})}\,,
	\label{hstar} \end{aligned}
\end{equation}
so that, eventually,
\begin{equation}
	\Xi_{\l,t;\L} (\bs{\varphi}) \propto \exp\left(\sum_{h=h^*-1}^{-1}\cW_\L^{(h)}(\bs{\varphi})\right)\,.\label{eventually}
\end{equation}
The basic tool for the evaluation of Eqs.~\eqref{eq:Vcyl_N}-\eqref{hstar} is the following formula, which we spell out in detail only for the effective boundary spin source term and for \eqref{eq:eff_inter_iter2}, i.e.\ the iteration step in which we integrate out the $h$-scale fields. Suppose that, for all $h^*<h \le -1$, $\cB^{(h)}$ can be written in a way analogous to the one on scale $h=0$, i.e.\ as in Eq.~\eqref{BL}:
\begin{equation} \label{expgen}
\cB^{(h)}_\L(\phi,\bs{\varphi}) = \sum_{\substack{\bs{\Psi} \in  \mathcal M_\L\\  \bs{y}\in \mathcal{Q}(\partial^l \L)}} B^{(h)}_\L (\bs{\Psi}, \bs{y}) \phi (\bs{\Psi}) \varphi (\bs{y}) 
\end{equation}
for a suitable real \textit{h-scale effective} kernel function $B^{(h)}_\L (\bs{\Psi}, \bs{y}): \mathcal M_\L\times \mathcal{Q}(\partial^l \L) \to \mathbb{R}$ where $\mathcal M_\L$ denotes the set of $\bs{\Psi}= (\bs{\omega},\bs{z}) \in (\set{\pm} \times \L)^{n}$,  $\mathcal{Q}(\partial^l \L)$ denotes the set of $\bs{y} \in (\partial^l \L)^{m}$, $n+m \in 2 \mathbb{N}$, $\phi(\bs{\Psi})= \phi_{\omega_1,z_1} \cdots \phi_{\omega_{n},z_{n}}$ and  $\varphi(\bs{y})=\varphi_{y_1}\cdots \varphi_{y_m}$.  Then  $\cB^{(h-1)}_\L$ , as computed from Eq.~\eqref{eq:eff_inter_iter2}, admits an expansion analogous to \eqref{expgen},
with $B_\L^{(h)}(\bs{\Psi},\bs y)$ replaced by
\begin{equation}\label{Bh-1L}\begin{aligned}
&B_\L^{(h-1)}( \bs{\Psi},\bs{y})  =  \sum_{s=1}^\infty \frac{1}{s!} \sum_{\substack{\bs{\Psi}_1,\ldots,\bs{\Psi}_s \in \mathcal{M}_{\L}\\\bs{y}_1,\ldots,\bs{y}_s \in \mathcal{Q}(\partial^l \L)}}^{(\bs{\Psi},\bs{y})}  \alpha_s(\bs{\Psi}, \bs{y})\,  \mathbb{E}_\L^{(h)}(\phi(\bar{\bs{\Psi}}_1); \ldots;\phi(\bar{\bs{\Psi}}_s)) \prod_{j=1}^s  B_\L^{(h)}(\bs{\Psi}_j,\bs{y}_j) \,,\end{aligned}
\end{equation}
where: 
\begin{itemize}
\item the superscript $(\bs{\Psi},\bs{y})$ on the second sum indicates that the sum runs over all ways of representing $\bs{\Psi}$ as an ordered sum of $s$ (possibly empty) tuples, $\bs{\Psi}_1'\oplus\cdots\oplus\bs{\Psi}'_s = \bs{\Psi}$, over all tuples $\mathcal{M}_{\L}\ni \bs{\Psi}_j \supseteq \bs{\Psi}'_j$ and over the (possibly empty) tuples $\bs{y}_1,\ldots,\bs{y}_s$ such that $\bs{y}_1\oplus \cdots \oplus \bs{y}_s = \bs{y}$;
\item we let $\bar{\bs{\Psi}}_j:= \bs{\Psi}_j\setminus \bs{\Psi}'_j$ be the multilabels of the contracted fields and $\alpha_s(\bs{\Psi},\bs{y})$ be the sign of the permutation from $\bs{\Psi}_1 \oplus \bs y_1 \oplus \cdots \oplus \bs{\Psi}_s \oplus \bs y_s \to \bs{\Psi}\oplus \bs y \oplus \bar{\bs{\Psi}}_1 \oplus \cdots \oplus \bar{\bs{\Psi}}_s$;
\item  $\mathbb{E}_\L^{(h)}(\phi(\bar{\bs{\Psi}}_1); \cdots;\phi(\bar{\bs{\Psi}}_s))$ is the truncated expectation of the monomials $\phi(\bar{\bs{\Psi}}_1),\ldots,\phi(\bar{\bs{\Psi}}_s)$ with respect to the Grassmann Gaussian integration with propagator $\mathfrak{g}_\L^{(h)}$ of \eqref{g0h}; the truncated expectation vanishes if one of the sets $\bar{\bs{\Psi}}_i$ is empty, 
unless $s=1$, in which case $\mathbb{E}_\L^{(h)}(\emptyset)=1$.
\end{itemize}
The explicit expression for $\mathbb{E}_\L^{(h)}(\phi(\bar{\bs{\Psi}}_1); \cdots;\phi(\bar{\bs{\Psi}}_s))$ will be introduced in Remark~\ref{BFF} below (directly in the half-plane limit), for now it can be thought of as a polynomial in the $\mathfrak{g}_\L^{(h)}$ propagators. 
The single-scale contribution to the generating function $\cW_\L^{(h)}$ and the effective interaction $\cV_\L^{(h)}$ admit expressions similar to Eqs.~\eqref{expgen}-\eqref{Bh-1L}: those for $\cW_\L^{(h)}$ (resp.\ for  $\cV_\L^{(h)}$) on scale $h^*\le h \le -1$ (resp.\ $h^*<h \le 1$) are obtained by replacing $\bs{\Psi}$ (resp.\ $\bs{y}$) with the empty set.

As a result, we note that Eqs.~\eqref{expgen}-\eqref{Bh-1L} and their analogous expressions allow us to represent the $(h-1)$-scale effective kernels, i.e.\ the coefficients of $\cB^{(h-1)}_\L, \cW^{(h-1)}_\L$ and $\cV^{(h-1)}_\L$, purely in terms of $h$-scale propagators and $h$-scale effective kernels. This observation allows us to iteratively obtain for each $h$ the half-plane limit of the $h$-scale effective kernels using the half-plane limits of $h$-scale propagators (see \eqref{gH1} and \eqref{gceta} combined with \eqref{g0h}) and initial scale kernels (see \eqref{W1HH} and \eqref{BH0}). In terms of these half-plane kernels we will be able to define the half-plane limits $\cB^{(h-1)}_\bH, \cW^{(h-1)}_\bH$ and $\cV^{(h-1)}_\bH$ (with expressions similar to the one in \eqref{expgen} but with positions on the half-plane) and therefore of the half-plane limit of \eqref{eventually} as $ \Xi_{\l,t;\bH} (\bs{\varphi}) \propto \exp\left(\sum_{h\le-1}\cW_\bH^{(h)}(\bs{\varphi})\right)$. In the following we shall discuss the solution to the recursive equations for the kernels in the half-plane limit. Given this solution, convergence of the finite volume kernels to the solution of the recursive equations for the kernels in the half-plane limit, with explicit bounds on the norm of the finite size corrections, can be proved as in \cite[Section~3]{AGG_CMP}, and the interested reader is referred there for this technical aspect of the construction.

The half-plane analogues of Eqs.~\eqref{expgen}-\eqref{Bh-1L}, by systematically using the definition of the truncated expectations with respect to the Grassmann Gaussian integration with propagator $\mathfrak{g}^{(h)}_\bH$, in combination with the decay bounds on $\mathfrak{g}^{(h)}_\bH$ in \eqref{gdecay}, allow one to get bounds on the kernels of $\cB^{(h)}_\bH, \cW^{(h)}_\bH$ and $\cV^{(h)}_\bH$ and, consequently, on those of $\log \Xi_{\l,t^*;\bH} (\bs{\varphi})$. In order to obtain bounds on the half-plane kernels that are uniform in the scale label $h$, so that the resulting expansions for the correlations are convergent, at each step it is necessary to isolate from $\cV_\bH^{(h)}$ and $\cB_\bH^{(h)}$ the contributions that tend to ‘expand’ (in an appropriate norm) under iterations: these, in the RG terminology, are associated with the so-called  `relevant' and `marginal' contributions. We introduce a \textit{localization procedure} of the kernels to obtain, at each step of the iterative procedure, decompositions as $\cV_\bH^{(h)}(\phi) = \cL\cV_\bH^{(h)}(\phi)+ \cR \cV_\bH^{(h)}(\phi)$ and $\cB_\bH^{(h)} (\phi,\bs{\varphi}) = \cL  \cB_\bH^{(h)} (\phi,\bs{\varphi}) + \cR  \cB_\bH^{(h)} (\phi,\bs{\varphi})$: $\cL\cV_\bH^{(h)}$ and $\cL  \cB_\bH^{(h)} $ are \textit{local parts}, collecting the potentially divergent contributions, parametrized at any given scale, by a finite number of $h$-dependent `running coupling constants'; $\cR \cV_\bH^{(h)}$ and $\cR \cB_\bH^{(h)}$ are \textit{irrelevant parts} or \textit{remainders}, containing all the other contributions, which are not the source of any divergence. To show that the irrelevant contributions satisfy `improved dimensional bounds', it is necessary to rewrite their kernels  in an appropriate, interpolated, form, involving the action of discrete derivatives on the Grassmann fields.
In \cite[Section~3.1]{AGG_CMP} the authors studied in detail the effective interaction $\cV_\L^{(h)}$ and its local-irrelevant decomposition: as already mentioned, here we are considering the same effective interaction (in the half-plane limit and without energy sources) therefore we do not repeat the technical discussion but directly illustrate how to derive the local-irrelevant decomposition of $\cB_\bH^{(h)}$. An important observation (see also Remark~\ref{Wgen} below): for ease of writing we use the symbols $\cL$ and $\cR$ in both the decomposition of $\cV_\bH^{(h)}$ and $\cB_\bH^{(h)}$ while they denote different operators depending on whether they act on $\cV_\bH^{(h)}$ or $\cB_\bH^{(h)}$.

The plan of the incoming subsections is the following: first, in Section~\ref{kernrep} we introduce the representation of the effective functions in the half-plane limit. Then, in Section~\ref{LeR} we define the action of the operators $\cL$ and $\cR$ on the effective kernels of $\cB_\bH^{(h)}$ and derive dimensional bounds on the remainders. Finally, in Section~\ref{sec:trees}, we derive the solution to the recursive equations for the effective kernels in terms of a tree expansion.

\subsection{Effective kernel functions.}\label{kernrep}

In this section we proceed as in section \cite[Section~4.1]{AGG_AHP}. We represent each half-plane effective function with an expression similar for example to that of Eq.~\eqref{BL}, that is, as a formal Grassmann polynomial in which the half-plane kernel plays the role of coefficient function: this points of view avoids defining an infinite-dimensional Grassmann algebra. As mentioned, we want this representation to be such that we can keep track of when and how some of the Grassmann fields are collected into discrete derivatives (corresponding to the local-remainder decomposition which will be introduced in the next section) therefore, exploiting the equivalence relations defined in \cite[Section~4.1]{AGG_AHP} (which also apply to half-plane kernels), we introduce the field multilabels that index the half-plane kernels as follows.

Let $\bar{\bH}= \bH \cup \set{z \in \mathbb{Z}^2 : (z)_2=0}$, where the additional row at vertical height $0$ is the one on which we interpret $\phi_{+,\bar{z}}$ as a vanishing field (see below Eq.~\eqref{gc0}). Let ${{\cM}}_{\bar{\bH}}$ be the set of \textit{field multilabels}, i.e.\ the set of tuples of the form $$\bs{\Psi} =(\bs{\omega},\bs{D},\bs{z})=((\omega_1,D_1, z_1), \ldots, (\omega_{n }, D_{n }, z_{n})) \in (\set{\pm} \times \set{0,1,2}^2\times \bar{\bH})^{n},$$ with $||D_i||_1 \le 2$, 
for some $n\ge 1$; let ${{\cM}}_{\bH}$ be the restriction of set ${{\cM}}_{\bar{\bH}}$ to the case of $z_i$ 
in $\bH$ rather than in $\bar{\bH}$\footnote{We are using a slightly different notation than \cite{AGG_AHP,AGG_CMP}, where the sets are denoted as $\cM$ and $\cM^\circ$ respectively and  there is a constraint that $n$ is even.}; any  $\bs{\Psi} \in \cM_{\bar{\bH}}$ indexes a formal Grassmann monomial $\phi(\bs{\Psi}) := \partial^{D_1} \phi_{\omega_1,z_1} \cdots  \partial^{D_{n }}\phi_{\omega_{n },z_{n }}$, with $\partial^{D_j} =\partial^{(D_j)_1}_1\partial^{(D_j)_2}_2 $, $(D_j)_1,(D_j)_2 \in \set{0,1,2}$, $j = 1,\ldots,n$.\footnote{We let $\partial_j \phi_{\omega, z} := \phi_{\omega, z+\hat{e}_j}-\phi_{\omega, z}$ be the right discrete derivative in direction $j=1,2$.} Let $\mathcal{Q}(\partial \bH)$ be the set of boundary spin source positions, i.e.\ the set of tuples $\bs{y} = (y_1,\ldots, y_{m}) \in (\partial \bH)^{m}$, $m \in 2\mathbb{N}$; any $\bs{y} \in \mathcal{Q}(\partial \bH)$ indexes $\varphi(\bs{y})=\varphi_{y_1}\cdots \varphi_{y_m}$.

In terms of these definitions we can rewrite the initial boundary spin source $\cB_\bH^{(0)}(\phi,\bs{\varphi})$ (cf.\ \eqref{BL}, see below \eqref{eq:eff_inter_iter2}) in the half-plane limit as
\begin{equation}\label{Bh0}
\cB_\bH^{(0)}(\phi,\bs{\varphi}) = \sum_{\substack{\bs{\Psi} \in \cM_{\bar{\bH}}:\\ n=1, \bs{D} = \0}}  \sum_{\substack{\bs{y} \in \mathcal{Q}(\partial \bH):\\ m=1}}  B_\bH^{(0)}(\bs{\Psi},\bs{y})  \phi(\bs{\Psi})\varphi(\bs{y}) \,,
\end{equation}
where $B_\bH^{(0)}(\bs{\Psi},\bs{y})$ is exactly the same kernel introduced in Eq.~\eqref{BH0}. 
We will inductively prove that the effective boundary spin source $\cB_\bH^{(h)}(\phi,\bs{\varphi})$ for any $h \le -1$ can be written with a similar expression, namely 
\begin{equation}\label{Bkernel}
\cB_\bH^{(h)}(\phi,\bs{\varphi})  =\sum_{\substack{\bs{\Psi} \in \cM_{\bar{\bH}}}} \sum_{\substack{\bs{y} \in \mathcal{Q}(\partial \bH)}}  B_\bH^{(h)}( \bs{\Psi},\bs{y}) \phi(\bs{\Psi}) \varphi(\bs{y})\,,
\end{equation}
for a suitable $B_\bH^{(h)} (\bs{\Psi}, \bs{y}): \mathcal{M}_{\bar{\bH}}  \times  \mathcal{Q}(\partial \bH) \to \mathbb{R}$, which is the $h$-scale effective kernel function of the half-plane boundary spin source. $B_\bH^{(h)} (\bs{\Psi}, \bs{y})$ is antisymmetric under permutations of the elements of $\bs{\Psi}$ and/or of
$\bs{y}$ and can be chosen in such a way as to respect the symmetry properties mentioned in Remark~ \ref{item:symm}.  Similarly, we can represent the half-plane limit of the single-scale contribution to the generating function, for any $h \le -1$, as 
\begin{equation}\label{wh-1} \mathcal{W}_\bH^{(h)}(\bs{\varphi}) = \sum_{\substack{\bs{y} \in \mathcal{Q}(\partial \bH)}} w_\bH^{(h)}(\bs{y}) {\varphi}(\bs{y})\,,\end{equation} 
for a suitable  $w_\bH^{(h)} (\bs{y}):  \mathcal{Q}(\partial \bH) \to \mathbb{R}$ which is the $h$-scale effective kernel function.

\begin{remark}[Equivalent kernels]\label{equivK} We say that two kernels on $\cM_{\bar{\bH}} \times \mathcal{Q}(\partial \bH)$, such as the kernel $B_\bH^{(h)}$ from Eq.~\eqref{Bkernel} and the kernel $\tilde{B}_\bH^{(h)}$ from  $\mathcal{\tilde{B}}_\bH^{(h)}(\phi,\bs{\varphi})  = \sum_{\bs{\Psi} \in \cM_{\bar{\bH}}}\sum_{\bs{y}\in \mathcal{Q}(\partial \bH)} \tilde{B}_\bH^{(h)}(\bs{\Psi},\bs{y}) \phi(\bs{\Psi}) \varphi(\bs{y})$, are `equivalent kernels', or we write $B_\bH^{(h)} \sim \tilde{B}_\bH^{(h)}$, if the corresponding boundary spin sources $\mathcal{B}_\bH^{(h)}(\phi,\bs{\varphi})$ and $\mathcal{\tilde{B}}_\bH^{(h)}(\phi,\bs{\varphi})$ are equal: the definition of equivalence relation is given below \cite[Eq.~(4.1.2)]{AGG_AHP}. In simple words, it  generalizes the equivalence between different ways of writing the coefficients of a given polynomial and corresponds to manipulations allowed by the anti-commutativity of the Grassmann variables and by the definition of discrete derivative.  \end{remark}
  
\begin{remark}[Interaction kernels in the half-plane limit] \label{Wint}
As already mentioned, $\cV_\bH^{(h)}(\phi)$ is the interaction potential studied in \cite{AGG_AHP, AGG_CMP} adapted to the half-plane limit. By using Eq.~\eqref{W1HH}, we can express the half-plane limit of Eq.~\eqref{inte}, as  
\begin{equation}\label{V1H} 
\cV_\bH^{(1)}(\phi) = \sum_{\bs{\Psi} \in \cM_{1,{\bH}}} W_\bH^{(1)}(\bs{\Psi})\phi(\bs{\Psi})\,, 
\end{equation}
where $\cM_{1,\bH}$ denotes the set of the tuples $\bs{\Psi} = ( \bs{\omega},\bs{z}) \in (\pm, \pm i)^n \times \bH^n$, $W_\bH^{(1)}(\bs{\Psi})$ denotes $W_\bH^{(1)}(\bs{\omega},\bs{z})$ of Eq.~\eqref{W1HH} and $\phi_{\pm i,z}$ denotes $\xi_{\pm, z}$. Note that, as a result of Eq.~\eqref{eq:B_base_decay}, for any $n \in \mathbb{N}$ and suitable positive constants $C, c_0, \kappa$, the half-plane initial scale kernel $W_\bH^{(1)}(\bs{\omega},\bs{z})$ of Eq.~\eqref{W1HH} satisfies the following bound: 
\begin{equation}\label{WintL}
\sup_{\bs{\omega} \in \set{\pm,\pm i}^n} \sup_{z_1 \in \bH} \sum_{z_2,\ldots,z_n \in \bH} e^{2c_0\delta(\bs{z})} \left| {W_\bH^{(1)} (\bs{\omega},\bs{z})} \right| \le C^n |\lambda|^{\max(1,\kappa n)}\,
\end{equation}
where $\delta(\bs{z})$ denotes the tree distance of $\bs{z}$, i.e.\ the cardinality of the smallest connected set of edges that have endpoints in $\bs{z}$; note that the analyticity property mentioned after \eqref{eq:B_base_decay} still holds. Moreover, for any fixed $n$-tuple $\bs{z} \in (\mathbb{Z}^2)^n$, $n \in 2 \mathbb{N}$,
we let \begin{equation} \label{Winfty}W^{(1)}_\infty (\bs{\omega},\bs{z}) := \lim_{M \to \infty} W^{(1)}_\bH(\bs{\omega}, \bs{z}+(0, \lfloor M\rfloor))\end{equation}
be the infinite plane initial scale kernel of the interaction (note that \eqref{Winfty} coincides with \cite[Eq.~(3.19)]{AGG_AHP}  or \cite[Eq.~(2.2.11)]{AGG_CMP} with $\bs{x} =\emptyset$, i.e.\ in the absence of energy sources). The kernel $W^{(1)}_\infty(\bs{\omega},\bs{z})$, besides being antisymmetric under simultaneous permutations of $\bs{\omega}$ and $\bs{z}$, is invariant under translations of $\bs{z}$ in both directions and invariant under the infinite plane reflection symmetries induced by $\phi_{\omega,z} \to  i \omega \phi_{\omega,(-(z)_1,(z)_2)}$, $\xi_{\pm,z} \to i \xi_{\mp,(-(z)_1,(z)_2)}$ (horizontal reflections) and by $\phi_{\pm,z} \to i \phi_{\mp, ((z)_1,-(z)_2)}$, $\xi_{\pm,z} \to \mp i \xi_{\pm, ((z)_1,-(z)_2)}$ (vertical reflections). Taking advantage of the fact that the infinite volume limit of the kernels is reached exponentially fast, we perform the bulk-edge decomposition\footnote{Note that, as in the previous section, we use the terminology `bulk-edge' in analogy with the references, even if the bulk part of the half-plane kernel coincides with the infinite plane kernel (factors that take into account the finiteness of $L$ in first line of \cite[Eq.~(3.17)]{AGG_AHP} in the half-plane limit are equal to $1$). } (cf.\ \cite[Eq.~(3.17)]{AGG_AHP} or \cite[Eq.(2.2.12)]{AGG_CMP}) as: 
\begin{equation}
\label{WE}
\begin{aligned} 
W^{(1)}_\bH( \bs{\omega},\bs{z}) = W^{(1)}_\infty ( \bs{\omega},\bs{z})+W^{(1)}_{\bH,E}( \bs{\omega},\bs{z})\,,  \text{ with }\,\,
W^{(1)}_{\bH,E}( \bs{\omega},\bs{z})&:= W_\bH^{(1)}( \bs{\omega},\bs{z})-W^{(1)}_\infty( \bs{\omega},\bs{z})\,.\end{aligned}
\end{equation}
where: $W_{\bH,E}^{(1)}$ respects the same properties mentioned in Remark~\ref{item:symm} for $W_\bH^{(1)}$; $W_\infty^{(1)}$ satisfies the same bound as $W_\bH^{(1)}$ in Eq.~\eqref{WintL}; and $W^{(1)}_{\bH,E}$, for any $n \in \mathbb{N}$, satisfies the following improved bound:
\begin{equation}\label{WEbound}
\sup_{\bs{\omega}\in \set{\pm,\pm i}^n}\sum_{\bs{z} \in \bH^n} e^{2c_0\delta_E(\bs{z})} \left| W_{E,\bH}^{(1)} (\bs{\omega},\bs{z}) \right| \le C^n |\l|^{\max(1,\kappa n)}\,, 
\end{equation} 
where $C,c_0,\kappa$ are as in \eqref{WintL} and $\delta_E(\bs{z})$ denotes the edge tree distance of $\bs{z}$ in the half-plane, i.e.\ the cardinality of the smallest connected set of edges that have endpoints in $\bs{z}$ and in $\partial \bH$. To prove these estimates it is sufficient to retrace the proof of \cite[Lemma~3.2]{AGG_AHP}.

As already mentioned, the half-plane limit of the recursive definitions in Eqs.~\eqref{expgen}-\eqref{Bh-1L} allows us to derive the half-plane limit as well as the bulk-edge decomposition  of each $h$-scale effective kernel using the half-plane  limit and the bulk-edge decomposition of the propagator (\eqref{gceta} and below \eqref{ginfty}) and of the initial kernel (\eqref{W1HH} and \eqref{WE}) and then to represent the half-plane limit of the effective interaction, for any $h \le 0$, as
\begin{equation}\label{Vheff}
\cV_\bH^{(h)}(\phi)= \sum_{\bs{\Psi} \in \mathcal{M}_{{\bar{\bH}}}} W_\bH^{(h)}(\bs{\Psi})\phi(\bs{\Psi}) = \sum_{\bs{\Psi} \in \mathcal{M}_{{\bar{\bH}}}} (W_\infty^{(h)}(\bs{\Psi})+ W_{\bH,E}^{(h)}(\bs{\Psi}))\phi(\bs{\Psi})
\,.\end{equation} 
Eq.~\eqref{Vheff} is the half-plane analogue of \cite[Equations (3.1.2)-(3.1.4)]{AGG_CMP}. 
\end{remark}
\begin{remark}[Scaling dimension with boundary spin sources] \label{sd}
	We let $(B_\bH^{(h)})_{n,p,m}$ denote the restriction of $B_\bH^{(h)}$ to field multilabels $\bs{\Psi}$ of length $n$, i.e.\ of the form 
$\bs{\Psi}=((\omega_1,D_1, z_1), \ldots, (\omega_{n }, D_{n }, z_{n})$, with $||\bs D ||_1 = \sum_{j=1}^{n }||D_j||_1=p$ and to tuples $\bs{y}$ of length $m$. Note that  $(B_\bH^{(h)})_{n,p,m}$ is different from zero only if $m \ge 1$ and $n+m \in 2 \mathbb{N}$. 

	Although there are no divergences associated with observables in a lattice theory, classifying them by scaling dimension like the terms in the effective interaction (cf.\ \cite[Remark 3.5]{AGG_CMP}) gives a way of separating terms in the correlation functions according to their long-distance asymptotic. Following the same logic as \cite[Section 2.2]{GMRS24}, 
we assign $(B_\bH^{(h)})_{n,p,m}$ the scaling dimension $1-\frac{n}{2} -p-\frac{m}{2}$ for $m \ge 1$, where $1$ corresponds to the spatial dimension of the support
of the kernel (that is, the boundary of the half-plane), $n/2$ to $n$ times the scaling dimension of the critical field $\phi$, and $m/2$ to $m$ times the scaling dimension 
of the source field $\varphi$ (fixed as in \cite[Section 2.2]{GMRS24}, in a setting, like ours, where there is no anomalous dimension).
Note in particular that, if $m\ge 1$ the combination $1-\frac{n}{2} -p-\frac{m}{2}$ will always be negative unless $m= n= 1$ and $p=0$: only the kernels $(B_\bH^{(h)})_{1,0,1}$ will be associated with the so-called local part, while all the other kernels will be associated with the remainders.
	\end{remark}

\subsection{Local-irrelevant decomposition.}
\label{LeR}

In this section we define the action the operators $\cL$ and $\cR$\footnote{{$\cL$ and  $\cR$ are linear operators, which, if applied to kernels that are invariant under the symmetries described in Remark~\ref{item:symm}, act as projection operators, orthogonal to each other; their action on the interaction kernels is defined in \cite[Section~3.1.2]{AGG_CMP}.}} to the kernel functions $B_\bH^{(h)} : \cM_{ \bar{{\bH}}} \times \mathcal{Q}(\partial {\bH}) \to \mathbb{R}$, depending on the different values of $n, p$ and $m$ (see Remark~\ref{sd}): the action of these operators corresponds to a decomposition $\cB_\bH^{(h)}= \cL \cB_\bH^{(h)}+ \cR \cB_\bH^{(h)}$, separating the dominant long-distance contributions of the correlation functions from the remainder.
If $(n,p,m) = (1,0,1)$ we let
\begin{equation}\label{locHP}
	(\cL B_\bH^{(h)} )_{1,0,1} ( (\omega,0, z), y)
	:= \delta_{\omega,-}\delta_{z,y}
	  \sum_{z'\in {\bH}}
	  (B_\bH^{(h)} )_{1,0,1} ( (-,0,z'), y)\,, 
	\end{equation}
and if $(n,p,m) \ne (1,0,1)$  we let $(\cL B_\bH^{(h)} )_{n,p,m}:= 0$. Note that this definition maintains the property that $\cL(B_\bH^{(h)} )_{1,0,1} = (\cL B_\bH^{(h)} )_{1,0,1}$, i.e.\ that the action of $\cL$ commutes with the operation of restricting to a certain order.
By translation invariance, this has the form
\begin{equation}\label{ZbsL} (\cL B_\bH^{(h)})_{1,0,1}((\omega,0, z), y) \equiv Z_{\Bs,h}  \delta_{\omega,-}\delta_{z,y}\,,
\end{equation} 
which defines a real quantity $Z_{\Bs,h}$, which is called \textit{boundary spin running coupling constant} on scale $h$;
note that  $(B_\bH^{(0)})_{1,0,1}$ as expressed in \cref{Bh0} has the same form (so for $h=0$, $Z_{\Bs,0}= Z^{-\frac12}$).

Now we associate the sub-leading contributions with kernels $[(B_\bH^{(h)} )_{1,0,1} -(\cL B_\bH^{(h)} )_{1,0,1}]$, and show that even if these kernels are expressed with index $p=0$, they are equivalent to kernels with index $p=1$ (equivalent in the sense of \cref{equivK}).
First of all, due to the presence of $\delta_{\omega,-}$ in r.h.s.\ of Eq.~\eqref{locHP} we distinguish two cases:
\begin{enumerate}
\item  if $\omega = -$, we can consider the following expression
\begin{equation}
\begin{aligned}
\sum_{z \in \bar{\bH}} &[(B_\bH^{(h)} )_{1,0,1} -(\cL B_\bH^{(h)} )_{1,0,1} ] ( (-,0, z), y) \phi_{-,z} \varphi_y\\
&\qquad = \sum_{z \in \bar{\bH}} (B_\bH^{(h)} )_{1,0,1}( (-,0, z), y) \left[\phi_{-,z}-\phi_{-,y} \right] \varphi_y\,,
\label{int101oy}
\end{aligned}\end{equation}
\item  if $\omega = +$, $(\cL B_\bH^{(h)} )_{1,0,1}((+,0,z),y)= 0$, so that we can consider
\begin{equation}
\begin{aligned}
\sum_{z \in \bar{\bH}} &(B_\bH^{(h)} )_{1,0,1} ( (+,0, z), y) \phi_{+,z} \varphi_y\\
&\qquad = \sum_{z \in \bar{\bH}} (B_\bH^{(h)} )_{1,0,1}( (+,0, z), y) \left(\phi_{+,z}- \phi_{+,\bar{z}}\right) \varphi_y\,,
\label{int101-oy}
\end{aligned}\end{equation}
where $\phi_{+,\bar{z}}$ is the vanishing field defined below Eq.~\eqref{gc0}.
\end{enumerate}
We can rewrite, in the last line of both \eqref{int101oy} and \eqref{int101-oy}, the difference of fields as a suitable interpolation of the differences of adjacent fields, obtaining
\begin{equation}
\begin{aligned}
&\sum_{z \in \bar{\bH}}  (B_\bH^{(h)} )_{1,0,1}( (\omega,0, z), y) \sum_{\substack{D\in\{0,1\}^2:\\||D||_1=1}} \
\sum^*_{\substack{z' \in \gamma_{\omega}(z)}} \sigma \partial^D \phi_{\omega,z'} \varphi_y
\label{intgamma}
\end{aligned}\end{equation}
where, for any $z \in \bar{\bH}$ and $y \in \partial {\bH}$, \begin{enumerate}
\item if $\omega= -$, $\gamma_{-}(z)$ is the shortest path obtained by going from $z$ to $y$ first vertically and then horizontally; 
\item  if $\omega= +$, $\gamma_{+}(z)$ is the path obtained by going vertically from $z$ to $\bar{z}= ((z)_1, 0)$;
\end{enumerate}
 $*$ denotes the constraint that the sum is over $z' \in \bar{\bH}$ such that $z', z'+ D \in \gamma_{\omega}(z)$ with $D\in\{(1,0),(0,1)\}$, and $\sigma = +/-$ if $z'$ follows/precedes $z'+ D$ in the sequence defining $\gamma_{\omega}(z)$. Finally, by exchanging the order in which we sum the coordinates, we can rewrite Eq.~\eqref{intgamma} as
\begin{equation}
\begin{aligned}
& \sum_{z' \in \bar{\bH}} \sum_{\substack{D\in\{0,1\}^2:\\ ||D||_1=1}}\  \sum^*_{\substack{z :\\z'\in \gamma_\omega(z)}}\sigma (B_\bH^{(h)} )_{1,0,1}( (\omega,0, z), y) \partial^D \phi_{\omega,z'} \varphi_y\\ 
&\qquad =:\sum_{z' \in \bar{\bH}}   \sum_{\substack{D\in\{0,1\}^2:\\ ||D||_1=1}}  (\tilde{\cR} B_\bH^{(h)} )_{1,1,1} ((\omega, D,z'),y) \partial^D \phi_{\omega,z'} \varphi_y\,,\label{tcRL}
\end{aligned}\end{equation}
where in the last line we are defining the action of $\tcR $ on $(B_\bH^{(h)} )_{1,0,1}$ as $\tcR(B_\bH^{(h)} )_{1,0,1} = (\tcR B_\bH^{(h)} )_{1,1,1}$, i.e.\  the action of $\tcR$  “adds a derivative”, taking a kernel supported on multilabels without derivatives ($p=0$) to one which is supported on multilabels containing a single derivative ($p=1$); note that ${\tcR}$, unlike ${\cL}$, does not commute with the operation of restricting to a certain order. 
Eq.~\eqref{tcRL} is similar to the definition given, for example, in  \cite[Eq.(3.1.13)]{AGG_CMP} (here adapted to the half-plane and with a slightly different and simpler notation, since we can localize at the positions of the boundary spin sources). %

Using this construction, if $(n,p,m)=(1,1,1)$ we let 
\begin{equation}\label{cRHH}\begin{aligned}
&(\cR B_\bH^{(h)})_{1,1,1} :=  (\tilde{\cR}   B_\bH^{(h)})_{1,1,1} +  (B_\bH^{(h)})_{1 ,1,1}\,,\end{aligned}
\end{equation}
if $(n, p, m) = (1,0,1)$, we let $(\cR   B_\bH^{(h)})_{1,0,1} := 0$ and, for all other values of $(n, p, m)$, we let $(\cR   B_\bH^{(h)})_{n, p, m}:=(B_\bH^{(h)})_{n, p, m}$. 
\begin{remark}\label{RL=0}
From the definitions in Eq.~\eqref{locHP} and in  Eq.~\eqref{tcRL} it follows that $\cL(\cL B_\bH^{(h)} )_{1,0,1}= (\cL B_\bH^{(h)} )_{1,0,1}$ and, since $(\cL B_\bH^{(h)} )_{1,0,1}$ is supported on arguments with $\omega = -$ and $z=y$ in which case $\gamma_{-}(y)= \emptyset$, that $\tcR (\cL B_\bH^{(h)} )_{1,0,1}=0$. Moreover, $\cL (\cR B_\bH^{(h)}) = 0$ and $\cR (\cR B_\bH^{(h)}) = (\cR B_\bH^{(h)})$ (see Eqs.~\eqref{locHP} and \eqref{cRHH} and subsequent lines).\end{remark}
\begin{remark}[Norm bounds on the remainders] \label{normcRHH} As in \cite[Equation~(3.1.40)]{AGG_CMP} \footnote{Since we consider functions of sources placed on the half-plane boundary, we do not introduce a bulk-edge decomposition of the kernel of the boundary spin sources, and therefore we do not introduce the edge norm.}, we let, using a slightly different notation, 
\begin{equation}\label{normHH}
||(B_\bH^{(h)})_{n,p,m}(\bs{y})||_{h} :=\sup_{\bs{\omega} \in \set{\pm}^n} \sum_{\bs{z} \in \bar{\bH}^n} e^{\frac{c_0}2 2^h \delta(\bs{z},\bs{y})} \sup_{\substack{{\bs{D}}: \\||\bs{D}||_1=p}} | (B_\bH^{(h)})_{n,p,m}((\bs{\omega},\bs{D},\bs{z}),\bs{y})|\,,
\end{equation}
where $c_0$ is the same constant as in \eqref{gdecay} and \eqref{WintL}, and the summand in r.h.s.\ is interpreted as zero if $((\bs{\omega},\bs{D},\bs{z}),\bs{y}) \notin \cM_{\bar{\bH}} \times \mathcal{Q}(\partial \bH)$. 
Note, as a consequence of the invariance under horizontal translations, that the norm in the l.h.s.\  of Eq.~\eqref{normHH} actually depends on the $m-1$ differences between the (horizontal) components of $\bs{y}$ (so that, if $m = 1$, it does not depend on $\bs y$).

We can then proceed as in in \cite[Lemma 4.3]{AGG_AHP} to bound the norm of $(\cR B_\bH^{(h)})_{1,1,1}$ defined in Eq.~\eqref{cRHH} as
\begin{equation}\label{boundRHH}
||(\cR   B_\bH^{(h)})_{1,1,1}||_{h}  \le C 2^{-h} ||(B_\bH^{(h)})_{1,0,1}||_{ h+1} + ||(B_\bH^{(h)})_{1,1,1}||_{h}\,,
\end{equation}
where the first term in the r.h.s.\ of Eq.~\eqref{boundRHH} is the norm bound for $(\tcR   B_\bH^{(h)})_{1,1,1}$ as defined in Eq.~\eqref{tcRL}.
\end{remark}
Finally, using the representation in  Eq.~\eqref{Bkernel}, note that the {\it kernel} decomposition $B_\bH^{(h)} \sim \cL B_\bH^{(h)} +  \cR B_\bH^{(h)}$ induces the decomposition  for the boundary spin {\it sources} $\cB_\bH^{(h)} = \cL \cB_\bH^{(h)} +  \cR \cB_\bH^{(h)}$(see Remark~\ref{equivK}).

\begin{remark}[Local-irrelevant decomposition of the effective interaction]\label{Wgen}
	As already mentioned, we refer to \cite{AGG_AHP} and \cite{AGG_CMP} for the study of the effective interaction $\cV_\L^{(h)}(\phi)$: we can introduce the local part and the irrelevant part of $W_\bH^{(h)}$ (see \eqref{Vheff}) as the half-plane limit of their cylindrical versions in \cite[Eqs.~(3.1.52)-(3.1.53)]{AGG_CMP} (without energy sources), which preserves the relation $W_\bH^{(h)} \sim \cL W_\bH^{(h)}+\cR W_\bH^{(h)}$. Once again, we use $\cL$ and $\cR$ to indicate different operators depending on whether they act on $B_\bH^{(h)}$ or $W_\bH^{(h)}$: in particular, given the bulk-edge decomposition of $W_\bH^{(h)}$ (last equality of \eqref{Vheff}), it is necessary to take into account the different definitions in \cite[Sections~(3.1.2),(3.1.4)]{AGG_CMP}.\end{remark}

\subsection{The renormalized expansion.}
\label{sec:trees}

Now that the action of the operators $\cL$ and $\cR$ on $B_\bH^{(h)}$ has been defined, we are ready to define the recursion formulas for the effective kernels as the half-plane analogues of Eqs.~\eqref{expgen}-\eqref{Bh-1L}. Note that \eqref{Bh-1L} was introduced to define the strategy to follow but we had not provided all the details: we will provide them now in the half-plane limit. In particular, for any $(\bs{\Psi},\bs y)\in \mathcal{M}_{\bar{\bH}}\times \mathcal{Q}(\partial{\bH})$, an expansion analogous to the one introduced in \eqref{Bkernel} for $\mathcal{B}_\bH^{(h)}(\phi,\bs{\varphi})$ will also be valid for $\mathcal{B}_\bH^{(h-1)}(\phi,\bs{\varphi})$ with $B_\bH^{(h)}(\bs{\Psi},\bs y)$ replaced by
\begin{equation}\label{Bh-1}\begin{aligned}
B_\bH^{(h-1)}( \bs{\Psi},\bs{y}) &= 
\sum_{s=1}^\infty \frac{1}{s!} \sum_{\substack{\bs{\Psi}_1,\ldots,\bs{\Psi}_s \in \mathcal{M}_{\bH}\\\bs{y}_1,\ldots,\bs{y}_s \in \mathcal{Q}(\partial \bH)}}^{(\bs{\Psi},\bs{y})}  \alpha_s (\bs{\Psi}, \bs{y})\,  \mathbb{E}_\bH^{(h)}(\phi(\bar{\bs{\Psi}}_1); \ldots;\phi(\bar{\bs{\Psi}}_s)) \\
& \times  \left( \prod_{j=1}^s \left( \cL B_\bH^{(h)}(\bs{\Psi}_j,\bs{y}_j) + \cR B_\bH^{(h)}(\bs{\Psi}_j,\bs{y}_j)+ \cL  W_\bH^{(h)}(\bs{\Psi}_j) + \cR W_\bH^{(h)}(\bs{\Psi}_j) \right)\right)\,,\end{aligned}
\end{equation}
where: 
\begin{itemize}
\item the superscript $(\bs{\Psi},\bs{y})$ on the second sum indicates that the sum runs over all ways of representing $\bs{\Psi}$ as an ordered sum of $s$ (possibly empty) tuples, $\bs{\Psi}_1'\oplus\cdots\oplus\bs{\Psi}'_s = \bs{\Psi}$, over all tuples $\mathcal{M}_{\bH}\ni \bs{\Psi}_j \supseteq \bs{\Psi}'_j$ and over the (possibly empty) tuples $\bs{y}_1,\ldots,\bs{y}_s$ such that $\bs{y}_1\oplus \cdots \oplus \bs{y}_s = \bs{y}$;
\item we let $\bar{\bs{\Psi}}_j:= \bs{\Psi}_j\setminus \bs{\Psi}'_j$ be the multilabels of the contracted fields and $\alpha_s(\bs{\Psi},\bs{y})$ be the sign of the permutation from $\bs{\Psi}_1 \oplus \bs y_1 \oplus \cdots \oplus \bs{\Psi}_s \oplus \bs y_s \to \bs{\Psi}\oplus \bs y \oplus \bar{\bs{\Psi}}_1 \oplus \cdots \oplus \bar{\bs{\Psi}}_s$;
\item  $\mathbb{E}_\bH^{(h)}(\phi(\bar{\bs{\Psi}}_1); \cdots;\phi(\bar{\bs{\Psi}}_s))$ is the truncated expectation of the monomials $\phi(\bar{\bs{\Psi}}_1),\ldots,\phi(\bar{\bs{\Psi}}_s)$ with respect to the Gaussian Grassmann measure with propagator $\mathfrak{g}_\bH^{(h)}$ (see Remark~\ref{BFF} below); the truncated expectation vanishes if one of $\bar{\bs{\Psi}}_i =\emptyset$, unless $s=1$, in which case $\mathbb{E}_\bH^{(h)}(\emptyset)=1$;
\item  in the last line we distinguish the terms that depend on boundary spin sources from those that do not depend on them: $\cL B_\bH^{(h)} + \cR B_\bH^{(h)}$ denotes the local-irrelevant decomposition of the kernel $B_\bH^{(h)}$ introduced in Section~\ref{LeR} and $\cL  W_\bH^{(h)}+ \cR  W_\bH^{(h)}$ denotes the local-irrelevant decomposition of the kernel $W_\bH^{(h)}$ (see Remark~\ref{Wgen}).
\end{itemize}
\begin{remark}\label{LB0} Note that at $h=0$ $$\cL B_\bH^{(0)} (\bs{\Psi}_j,\bs{y}_j) = B_\bH^{(0)}(\bs{\Psi}_j,\bs{y}_j) = Z^{-\frac12} \delta_{\bs{\Psi}_j,(-,0,y)}\delta_{\bs{y}_j,y}$$ and  $\cR B_\bH^{(0)}(\bs{\Psi}_j,\bs{y}_j) = 0$, so we can rewrite the right hand side of \cref{Bh-1} accordingly.\end{remark}
Similarly we define the recursion formulas for the kernels of the single-scale contributions to the generating function for any $h \le -1$: an expansion analogous to the one introduced in \eqref{wh-1}  for $ \mathcal{W}_\bH^{(h)}(\bs{\varphi})$ holds for $\mathcal{W}_\bH^{(h-1)}$ where $w_\bH^{(h-1)}(\bs{y})$ is defined by the counterpart of Eq.~\eqref{Bh-1} with $\bs{\Psi} = \emptyset$ (note that we have not included $\mathcal{W}_\L^{(h)}(\bs{\varphi})$ in r.h.s.\ of Eq.~\eqref{eq:eff_inter_iter2}, so there is no term with $s=1$ and $\bs{\Psi}_1=\emptyset$).

\begin{remark}\label{BFF}
In Eq.~\eqref{Bh-1}, the truncated expectation can be evaluated explicitly in terms of the following Pfaffian formula, originally due to Battle, Brydges and Federbush \cite{Bry86,BF78}, later improved and simplified \cite{AR98,BK87} and rederived in several review papers \cite{GM01,Giu10,GMR21}, see e.g.\ \cite[Lemma 3]{GMT17} and \cite[Eq.(3.17)]{AGG_CMP}. 

Let $\mathcal{S}(\bar{\bs{\Psi}}_1,\ldots,  \bar{\bs{\Psi}}_s)$ be the set of all the `spanning trees' which has as elements all the sets $T$ formed by ordered pairs of $(f,f')$, with $f \in \bar{\bs{\Psi}}_i$, $f' \in \bar{\bs{\Psi}}_j$ and $i<j$, such that the graph of $T$, with vertices $\set{1,\ldots,s}$ and edges $\set{(i,j) \in \set{1,\ldots,s}^2: \exists (f,f') \in T \mbox{ with } f \in \bar{\bs{\Psi}}_i, f' \in \bar{\bs{\Psi}}_j}$ is a tree graph. Then
\begin{equation} \mathbb{E}_\bH^{(h)}(\phi(\bar{\bs{\Psi}}_1); \ldots;\phi(\bar{\bs{\Psi}}_s)) = \sum_{T \in \mathcal{S}(\bar{\bs{\Psi}}_1,\ldots,  \bar{\bs{\Psi}}_s)} \mathfrak{G}_{T,\bH}^{(h)}
(\bar{\bs{\Psi}}_1,\ldots,  \bar{\bs{\Psi}}_s) \end{equation}
with 
  \begin{equation}  \label{fGT}\begin{aligned} 
	\fG_{T,\bH}^{(h)}(\bar{\bs{\Psi}}_{1},\ldots,\bar{\bs{\Psi}}_{s}) &= \alpha_{T} (\bar{\bs{\Psi}}_{1},\ldots,\bar{\bs{\Psi}}_{s}) \left( \prod_{{(f,f')} \in T} g_{{(f,f')}}^{(h)} \right)\\
	&\times \int P_{\bar{\bs{\Psi}}_{1},\ldots,\bar{\bs{\Psi}}_{s},T} (d \bs{t}) \Pf (G_{{\bar{\bs{\Psi}}_1,\ldots,\bar{\bs{\Psi}}_s},T} (\bs{t}))\,,
		 \end{aligned} \end{equation} 
where \begin{itemize}
\item  $\alpha_{T} (\bar{\bs{\Psi}}_{1},\ldots,\bar{\bs{\Psi}}_{s})$ is the sign of permutation from $ \bar{\bs{\Psi}}_{1}\oplus \cdots \oplus \bar{\bs{\Psi}}_{s}$ to $ T \oplus (\bar{\bs{\Psi}}_{1} \setminus T)\oplus \cdots \oplus( \bar{\bs{\Psi}}_{s}\setminus T)$; 
\item ${(f,f')} = ((\omega_i,D_i,z_i),(\omega_j,D_j,z_j))$, $g^{(h)}_{(f,f')} := \partial_{z_i}^{D_i}\partial_{z_j}^{D_j} g^{(h)}_{\omega_i\omega_j}(z_i,z_j)$ and $g^{(h)}_{\omega_i\omega_j}(z_i,z_j)$ is an element of the propagator $\fg_\bH^{(h)}$ introduced in Section~\eqref{propagatori} (if $h=1$ it is an element of the propagator $\fg_\bH^{(1)}$ and $D_i=D_j=0$, if $h \le 0$ it is an element of $\fg_\bH^{(h)}$, obtained by \eqref{gceta} combined with \eqref{g0h});
\item $\bs{t} = \set{t_{i,j}}_{1\le i,j\le s}$ and $P_{\bar{\bs{\Psi}}_{1},\ldots,\bar{\bs{\Psi}}_{s},T} (d \bs{t})$ is a probability measure with support on a set of $\bs{t}$ such that $t_{i,j} =\bs{u}_i \cdot \bs{u}_j$ for some family of vectors $\bs{u}_i = \bs{u}_i(\bs{t}) \in \mathbb{R}_s$ of unit norm;
\item letting $2q = \sum_{i=1}^s |\bar{\bs{\Psi}}_i|$, $G^{(h)}_{\bar{\bs{\Psi}}_1,\ldots, \bar{\bs{\Psi}}_s,T}(\bs{t})$ is an antisymmetric $(2q-2s+2)\times (2q-2s+2)$ matrix, whose off-diagonal elements are given by $(G^{(h)}_{\bar{\bs{\Psi}}_1,\ldots, \bar{\bs{\Psi}}_s,T}(\bs{t}))_{{\bar{f},\bar{f}'}}= t_{i(\bar{f}),i(\bar{f}')}g^{(h)}_{(\bar{f},\bar{f}')}$, where $\bar{f},\bar{f}'$ are elements of the tuple $(\bar{\bs{\Psi}}_1\setminus T) \oplus \cdots \oplus (\bar{\bs{\Psi}}_s\setminus T)$, and $i(\bar{f})$ (resp.\ $i(\bar{f}'))$ is the integer in $\set{1,\ldots,s}$ such that $\bar{f}$ (resp.\ $\bar{f}'$)  is an element of $\bar{\bs{\Psi}}_i \setminus T$ (resp.\ $\bar{\bs{\Psi}}_j \setminus T$).
\end{itemize}
Note that $\mathfrak{G}_{T,\bH}^{(h)}$ vanishes if one of $\bs{\bar{\Psi}}_i =\emptyset$ unless $s=1$, in which case $\mathfrak{G}_\emptyset^{(h)}(\emptyset)=1$; if $s=1$ and $\bs{\bar{\Psi}}_1 \ne \emptyset$, we let $\mathcal{S}(\bs{\bar{\Psi}}_1)=\set{\emptyset}$ and $\mathfrak{G}_\emptyset^{(h)}(\bs{\bar{\Psi}}_1):=\Pf (G_{\bs{\bar{\Psi}}_1}^{(h)})$ where $(G_{\bs{\bar{\Psi}}_1}^{(h)})_{\bar{f},\bar{f}'}= g^{(h)}_{(\bar{f},\bar{f}')}$ and $\bar{f},\bar{f}'$ are elements of the tuple $\bs{\bar{\Psi}}_1$. 

For later reference, we recall the following estimates, obtained using the estimates in \eqref{gdecay} and those on infinite and edge propagators mentioned below. The following estimate is the half-plane version of \cite[Eq.~(4.4.8)]{AGG_AHP} (or \cite[Lemma~3.15]{AGG_CMP}):
\begin{equation}\label{GTbound}
| \fG_{T,\bH}^{(h)}(\bar{\bs{\Psi}}_{1},\ldots,\bar{\bs{\Psi}}_{s})| \le C^{2q}  2^{(q+ \sum_{f \in \cup_i \bs{\bar{\Psi}}_i }|| D(f)||_1)h} e^{-{c_0} 2^h \sum_{(f,f')\in T} ||z(f)-z(f')||_1}\,,
\end{equation}
where $2q$ is defined in the last item below \eqref{fGT}, {$c_0$} is the same constant as in \eqref{gdecay}
and $\bar{\bs{\Psi}}_{k}= (\bs{\omega}_k, \bs{D}_k, \bs{z}_k)$, for $k \in 1, \ldots, s$, $s \ge 1$, are all non-empty field multilabels. Moreover, we define $\fG_{T,\infty}^{(h)}$ by the analogue of \eqref{fGT} where each propagator (both in the product over the spanning tree lines in the first line and in the Pfaffian in the last line) is replaced by its infinite plane limit as in Eq.~\eqref{ginfty} (and has argument $\bs{z}_k \in \mathbb{Z}^2$ one of the representatives with whom $\bs{z}_k \in \bH$ can be identified). $|\fG_{T,\infty}^{(h)}|$ satisfies the same estimate as $|\fG_{T,\bH}^{(h)}|$ in \eqref{GTbound} while $|\fG_{T,\bH}^{(h)}- \fG_{T,\infty}^{(h)}|$ satisfies a similar estimate where $\sum_{(f,f') \in T}||z(f)-z(f')||_1 + \dist (\bs{z},\partial \bH)$ replaces $\sum_{(f,f') \in T}||z(f)-z(f')||_1$. 
\end{remark}

\subsubsection{Trees and tree expansion.}\label{T}

In this section, we describe the expansion for the kernels of the effective boundary spin sources, as it arises from the iterative application of the half-plane limit of Eq.~\eqref{eq:eff_inter_iter2}. As in \cite[Section~4.3]{AGG_AHP} (or \cite[Section~3.2]{AGG_CMP}), it is convenient to graphically represent the result of the expansion in terms of Gallavotti-Niccol\`o (GN) trees, to obtain the following expression:
\begin{equation}
B_\bH^{(h)} \sim \sum_{\tau \in \mathcal{T}^{(h)}_{\Bs}}B_\bH[\tau]\,, \,\, \mbox{ with } \,\, B_\bH[\tau] = \sum_{\substack{\underline{P} \in \mathcal{P}(\tau)\\ P_{v_0} \ne \emptyset}} \sum_{\underline{T} \in \mathcal{S}(\tau,\underline{P})}  \sum_{\underline{D} \in \mathcal{D}(\tau,\underline{P})} B_\bH[ \tau, \underline{P},\underline{T},\underline{D} ]\,,\label{Btau}
\end{equation}
where the set $\mathcal{T}^{(h)}_{\Bs}$ is a family of GN trees. In the following we provide the definitions of trees and tree values, but we do not repeat the details of the iterative construction leading to the tree expansion, as it is completely analogous to that described in \cite{AGG_AHP, AGG_CMP}.
A generic element of $\mathcal{T}^{(h)}_{\Bs}$ is a tree $\tau$: it consists of vertices, elements of $V(\tau)$, connected by edges as in Figure~\ref{FigTauBar}; each vertex $v \in V(\tau)$ has a scale label $h_v \in [h+1,2]\cap \mathbb{Z}$ (vertical dashed lines in the figure). The leftmost vertex is $v_0$, the \textit{root}, always on scale $h_{v_0}=h+1$; the rightmost vertices (which have only one incident edge) are called \textit{endpoints}, elements of $V_e(\tau)\subset V(\tau)$, and can be of five different types represented with different symbols and colors as in  Figure~\ref{FigTauBar}, see the caption for a description and definition of these symbols and colors.
We can refer to an intuitive partial order on $V(\tau)$: if $v\in V(\tau)$ can be connected to $u\in V(\tau)$ by an edge path going only from right to left, then  $v \ge u$ ($v >u$ if they are distinct vertices) and we can say that `$v$ {follows} $u$' and that `$u$ {precedes} $v$'; if the edge path connecting $v$ and $u$ consists of a single edge and $v > u$, we say that `$v$ is an {immediate successor} of $u$'. Note that the definitions imply that any $v \in V(\tau) \setminus \set{v_0}$ can be the immediate successor of a unique vertex, which we denote by $v'$, while any $v \in V(\tau)\setminus V_e(\tau)$ can have more than one immediate successor: we denote by $S_v$ the set consisting of the immediate successors of $v$. Finally, for any $v \in V(\tau)\setminus V_e(\tau)\setminus \set{v_0}$, we let $\tau_v \in \mathcal{T}_{\Bs}^{(h_v-1)}$ be the \textit{subtree} with root $v$ on scale $h_v$, whose vertex set is $V(\tau_v):= \set{w \in V(\tau) : w \ge v}$.
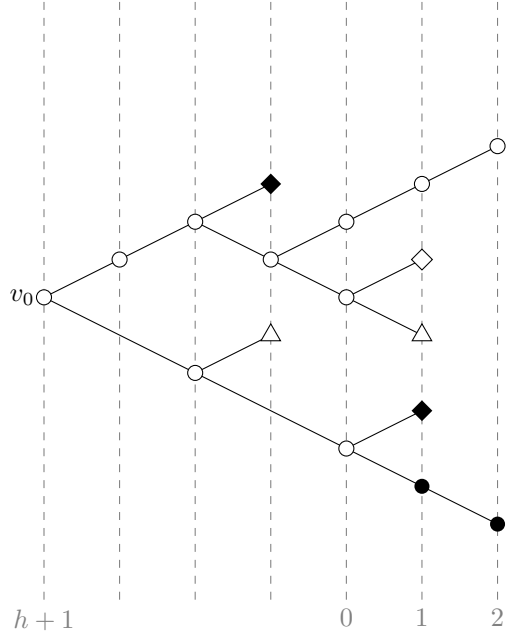
\begin{figure}[ht]
\centering
\begin{tikzpicture}[baseline=-0.4em]
\draw[dashed, gray] (0,0)node[below]{$h+1$} -- (0,8)  ;
\draw[dashed, gray] (1,0) -- (1,8) ;
\draw[dashed, gray] (2,0) -- (2,8) ;
\draw[dashed, gray] (3,0)  -- (3,8) ;
\draw[dashed, gray] (4,0) node[below]{$0$} -- (4,8) ;
\draw[dashed, gray] (5,0) node[below]{$1$}-- (5,8) ;
\draw[dashed, gray] (6,0) node[below]{$2$} -- (6,8) ;
\draw (0,4) -- (1,4.5);
\draw (1,4.5)--(2,5)-- (3,5.5);
\draw (2,5) -- (3,4.5) --(4,4)--(5,3.5);
\draw (2,3)--(3,3.5);
 \draw (4,4) -- (5,4.5);
\draw(3,4.5)--(4,5)--(5,5.5)--(6,6);
\draw (0,4)--(1,3.5)--(2,3)--(3,2.5)--(4,2)--(5,1.5)--(6,1);
\draw (4,2)--(5,2.5);
\node[left] at (0,4) {$v_0$};
\draw[color=black]   (0,4) node[vertex,E] {};
\draw[color=black]   (1,4.5) node[vertex,E] {};
\draw[color=black]   (2,5) node[vertex,E] {};
\draw[color=black]   (3,4.5) node[vertex,E] {};
\draw[color=black]   (5,1.5) node[vertex] {};
\draw[color=black]   (6,1) node[vertex] {};
\draw[color=black]    (5,2.5) node[ctVertex] {};
\draw[color=black]   (2,3) node[vertex,E] {};
\draw[color=black]   (4,2) node[vertex,E] {};
\node[black] at  (3,5.5) {\pgfuseplotmark{*}};
\draw[color=black]  (3,5.5) node[ctVertex] {};
\draw[color=black]   (4,5) node[vertex,E] {} ;
\draw[color=black]   (4,4) node[vertex,E] {};
\draw[color=black]   (5,5.5) node[vertex,E] {};
\draw[color=black]   (6,6) node[vertex,E] {};
\draw[color=black]  (5,4.5)  node[ctVertex,E] {} ;
\draw[color=black]  (5,3.5)  node[FSSpinEP] {} ;
\draw[color=black]  (3,3.5) node[FSSpinEP] {} ;
\end{tikzpicture}
\caption{Example of a tree $\tau \in \mathcal{T}_{\Bs}^{(h)}$ with all possible types of endpoints: endpoints \protect\tikzvertex{FSSpinEP}\,\,\protect  graphically represent the local parts of the boundary spin sources, endpoints \protect\tikzvertex{vertex}\protect, \protect\tikzvertex{ctVertex}\,\protect, \protect\tikzvertex{ctVertex,E}\,\protect and \protect\tikzvertex{vertex,E}\,\,\protect  graphically represent contributions from the effective interactions and are the ones studied in \cite{AGG_CMP} (here adapted to the half-plane). Endpoints \protect\tikzvertex{vertex} and \protect\tikzvertex{vertex,E} can only be on scale $h_v=2$; endpoints \protect\tikzvertex{FSSpinEP}\,\,\protect, \protect\tikzvertex{ctVertex}\,\protect and  \protect\tikzvertex{ctVertex,E}\,\protect  can be on any scale $h_v \in [h+2,1]\cap \mathbb{Z}$ but are always preceded by a branching point on scale $h_v-1$. The root $v_0$, the leftmost vertex, can be a branching point, as in this figure, in which case it is represented \textit{dotted}, or it can not be, in which case it can be represented dotted or undotted; every other $v \in V(\tau)\setminus V_e(\tau) \setminus \set{v_0}$ is always represented with a dotted symbol. Note that the white vertices are `hereditary': if a vertex is white, then all the dotted vertices preceding it are white. \label{FigTauBar}}
\end{figure}

In addition to the scale label $h_v$, we assign two other labels to each $v \in V(\tau)$: $E_v$, which is $1$ if $v$ is white and $0$ otherwise (if $v_0$ is not dotted, we let $E_{v_0}$ be the same as $E_{w_0}$, with $w_0$ the immediate successor of $v_0$), and $m_v$\footnote{Note that $m_v$ in \cite{AGG_CMP} was associated with energy sources while here it is associated with boundary spin sources.}, which is the number of \tikzvertex{FSSpinEP} endpoints following $v$, i.e.\ if $v$ is a \tikzvertex{FSSpinEP} endpoint, $m_v=1$ and if $v \in V(\tau) \setminus V_e(\tau)$, $m_v := \sum_{w \in V_e(\tau)}^{w > v} m_w$. Note that since \tikzvertex{FSSpinEP} is white, and a white vertex must be preceded by a white vertex, all vertices with $m_v\ge 1$ are white and have $E_v=1$. 

Finally, since trees as the one in Figure~\ref{FigTauBar} are obtained from iterative graphical expansions, by iterating graphical equations analogous to the ones in \cite[Figures~1-2]{AGG_AHP}, for example, in principle there should be a label $\cR$ at all the intersections of branches with vertical dashed lines, except when $v_0$ or one of the endpoints is on the intersection; however, by convention, in order not to overwhelm the figures, we prefer not to indicate these labels explicitly since the correct placement of these labels can always be reconstructed. This label would indicate the action of a $\cR$ operator which acts on the kernel associated with the vertex on the intersection or, if there is no vertex on the intersection, associated with the first vertex located on the right on the same branch. For this reason, a \protect\tikzvertex{FSSpinEP} endpoint on scale $h$, which graphically represent the local part of the effective source, must be  preceded by a branching point on scale $h-1$: if this were not the case, then on scale $h-1$ there would be an $\cR$ operator acting on a local kernel, but this would annihilate it, see Remark~\ref{RL=0}.

Given $\tau \in \cup_{h\le 0} \mathcal{T}_{\Bs}^{(h)}$, we introduce the following notation needed to explain the second term in \eqref{Btau}: 
\begin{itemize}
\item  $\mathcal{P}(\tau)$ is the set of the allowed field labels, concerning the Grassmann fields $\phi$ (and $\xi$ if $v$ is on scale $h_v=2$), whose generic element $\underline{P}= \set{P_v}_{v \in V(\tau)} \in \mathcal{P}(\tau)$ is characterized by the following properties: \begin{itemize}

\item if $v \in V_e(\tau)$ we assign it a label $j_v \in \set{1,\ldots, |V_e(\tau)|}$ to each $v \in V_e(\tau)$ so that the endpoints as in the figure are ordered from top to bottom, for example; then, if $v \in V_e(\tau)$, we let $P_v = \set{f_1,\ldots,f_n}$ be the set of the $n$ fields associated with $v$: for any $k = 1,\ldots, n$, $f_k:= (o(f_k), \omega(f_k))$, where $o(f_k):=(j_v, k)$ represents the position of $f_k$ in an ordered list of field indices associated with the endpoints, 
and $\omega(f_k)  \in\set{\pm,\pm i}$ (resp.\ $\omega(f_k) \in \set{\pm}$) if $h_v = 2$ (resp.\  if $h_v \le 1$) is a field index such that $\pm i$ values correspond to $\xi$ fields and $\pm$ values to $\phi$ fields;  $n = 1$ if $v$ is a \tikzvertex{FSSpinEP} endpoint, $n = 2$ (resp.\ $n \in \set{2,4}$) if $v$ is a \tikzvertex{ctVertex} (resp.\ \tikzvertex{ctVertex,E}) endpoint and $n \in 2 \mathbb{N}$ if  $v$ is a \tikzvertex{vertex} or a \tikzvertex{vertex,E} endpoint;

\item if $v \in V(\tau)\setminus V_e(\tau)$, $P_v \subset \cup_{w \in S_v} P_w$, and for any $v \in V(\tau)\setminus\set{v_0}$,  $|P_v|$ is positive and $|P_v|+ m_v$ is even. Additionally, 
if $v \in V(\tau)\setminus V_e(\tau)$, we let $Q_v := \{\left( \cup_{w \in S_v} P_w\right)\setminus P_v\}$ be the set of fields `contracted on scale $h_v$', so that if $v$ is dotted, $|Q_v|\ge 2$ and $|Q_v| \ge 2(|S_v|-1)$, while $Q_{v_0}$ is the empty set if $v_0$ is not dotted. Note that the definitions imply that for $v,w\in V(\tau)$ such that neither vertex precedes the other (for example when $v' = w'$ but $v \ne w$), $P_v$ and $P_w$ are disjoint, as are $Q_v$ and $Q_w$. Finally, since the contracted fields on scale $1$ are all and only the $\xi$ fields, if $h_v=1$,  $Q_v =  \cup_{w \in S_v} \set{f \in P_w: \omega(f) \in \set{\pm i}}$, while if $h_v<1$, $\omega(f) \in \set{\pm}$ for all $f \in Q_v$. As a consequence, any $v \in V_e(\tau)$ with $h_{v'} < 1$ has $\omega(f) \in \set{\pm}$ for all $f \in P_v$.
 
\end{itemize}
\item given $\underline{P} \in \mathcal{P}(\tau)$, $\mathcal{S}(\tau, \underline{P})$ is the set of anchored trees, that is the set of $\underline{T}= \set{T_v}_{v \in V(\tau)\setminus V_e(\tau)}$ where $T_v = \set{(f_1,f_2),\ldots, (f_{{2|S_v|-3}},f_{{2|S_v|-2}})}\subset Q_v^2$, such that no $f$ appears more than once and the graph obtained by connecting pairs $v,w$ such that $f_{2j-1} \in P_v, \ f_{2j} \in P_w$ for some $j$ is a tree graph.
	Alternatively, the elements of $\cS(\tau,\ul P)$ can be generated by taking a tree with vertex set $V(\tau)$ and ``anchoring'' it by associating each edge $\{v,w\}$ with a pair in $P_v \times P_w$ or $P_w \times P_v$, again with no repetitions.\footnote{In \cite{AGG_AHP,AGG_CMP} the elements of $\cS(\tau, \ul P)$ are called ``spanning trees''.}
\item given $\underline{P} \in \mathcal{P}(\tau)$, $\mathcal{D}(\tau, \underline{P})$ is the set of allowed derivative maps $\underline{D}= \set{D_v}_{v \in V(\tau)}$ where for each $v \in V(\tau)$, $D_v$ denotes a map $D_v: P_v \to \set{D \in \set{0,1,2}^2: \|D\|_1 \le 2}$ (the reader should think that $\partial^{D_v(f)}$ acts on the field with label $f$). \end{itemize}
Given these along with a map $z: P_v \to \bH$, we let $\bs{\Psi}_v = \bs{\Psi} (P_v):= (\bs{\omega}_v, \bs{D}_v, \bs{z}_v)$ be the field multi-label associated with the tuples $\bs{\omega}_v, \bs{D}_v$ and $\bs{z}_v$ of components $\omega(f), D(f)$, and $z(f)$, with  $f \in P_v$, respectively. If $v \in V(\tau)\setminus V_e(\tau)$ and also the maps $z: P_w \to \bH$, for all $w \in S_v$, are assigned, for each $w \in S_v$ we let $\bar{\bs{\Psi}}_w=\bs{\Psi}(P_w \setminus P_v) =(\left. \bs{\omega}_w \right|_{P_w \setminus P_v},\left. \bs{D}_w \right|_{P_w \setminus P_v},\left. \bs{z}_w \right|_{P_w\setminus P_v})$  be the 
restriction of $\bs{\Psi}_w$ to the subset $(P_w \setminus P_v)\subset P_v$. Finally, each $v\in V(\tau)$ with $m_v \ge 1$ is associated with a tuple $\bs{y}_v \in \mathcal{Q}(\partial \bH)$: if $v\in V_e(\tau)$ (therefore $m_v = 1$), $\bs{y}_v$ is simply the position of the boundary spin source associated with the  \tikzvertex{FSSpinEP} endpoint, if $v \in V(\tau) \setminus V_e(\tau)$,  $\bs{y}_v = \oplus_{w \in S_v} \bs{y}_w$ where the elements of $S_v$ are ordered in the way induced by the ordering of the endpoints.

In terms of these definitions, as in \cite[Eq.~(4.3.5)]{AGG_AHP} or \cite[Eq.~(3.2.7)]{AGG_CMP},  we can recursively define the kernel $B_\bH[\tau, \ul P, \ul T, \ul D]$ in 
Eq.~\eqref{Btau}  as
\begin{equation}
	\begin{split}
		B_\bH[\tau, \ul P, \ul T, \ul D]\big((\bs \omega_0,\bs D_0,\bs z_0),\bs y_{v_0}\big)
		&=  \mathds 1(\bs \omega_0=\bs \omega_{v_0}) \mathds 1(\bs D_0=\bs D_{v_0}= \bs D_{v_0}')\, \frac{\alpha_{v_0}}{ |S_{v_0}|!} 
		\\ & \times
		\sum_{\substack{z: \cup_{v\in S_{v_0}}\!P_{v}\to\bar\bH:\\ 
		\bs z_0= \bs z_{v_0}}}\hskip-.2truecm
		\fG_{T_{v_0},\bH}^{(h_{v_0})}(\bar{\bs{\Psi}}_{v_1},\ldots,\bar{\bs{\Psi}}_{v_{|S_{v_0}|}})\, \prod_{v \in S_{v_0}}K_v(\bs{\Psi}_v, \bs y_{v}),	
	\end{split}
	\label{eq:W_Lambda_gray_iter}
\end{equation}
where $\bs D_{v_0}':= \oplus_{v \in S_{v_0}} \left.\bs{D}_v \right|_{P_{v_0}}$ (which equals $0$ if $h_{v_0}=1$), $\alpha_{v_0} := \alpha (\bs{\Psi}_{v_0},{\bs{y}_{v_0}}, |S_{v_0}|)$ is defined as below \eqref{Bh-1}, in the second line the summand must be interpreted as zero if $(\bs{\Psi}_v,\bs{y}_v) \notin \cM_{\bar{\bH}} \times \mathcal{Q}(\partial \bH)$ for some $v\in S_{v_0}$, $\fG_{T_{v_0},\bH}^{(h_{v_0})}$ is defined as in Eq.~\eqref{fGT} and if $v \in S_{v_0}$, $K_v(\bs{\Psi}_v, \bs y_{v}) := K [\tau_v, \ul P, \ul T, \ul D] (\bs{\Psi}_v, \bs y_{v})$ is defined as follows: if $m_v \ge 1$, then\footnote{By construction, in order for a vertex $v$ to have $m_v\ge 1$, it must be on scale $h_v\le 1$; moreover, if $v$ with $m_v\ge 1$ is on scale $h_v=1$, then it is an endpoint of type \tikzvertex{FSSpinEP} (and, therefore, it has $m_v=1$).}
\begin{equation}
	K_v(\bs{\Psi}_v, \bs y_{v})
	:=
	\begin{cases}
	 Z_{\Bs,h_v-1} \delta_{\bs{\Psi}_v,(-,0,y)}  \delta_{\bs{y}_v,y}
		&
		\textup{if } v \in V_e(\tau) \textup{ (so that $v$ is of type  \tikzvertex{FSSpinEP} and $h_v=h_{v_0}+1$)},
		\\
		\bar B_\bH[\tau_v, \ul P_v, \ul T_v, \ul D_v] (\bs{\Psi}_v,  \bs y_{v})
		&
		\textup{if } v \in V(\tau)\setminus V_e(\tau),
	\end{cases}
	\label{eq:K_gray_def}
\end{equation}
where $\ul P_v$ (resp.\ $\underline{T}_v$, resp.\ $\underline{D}_v$) denotes the restriction of $\ul P$ (resp.\ $\underline{T}$, resp.\ $\underline{D}$) to the subtree $\tau_v$ and 
\begin{equation}
	\bar B_\bH[\tau_v, \ul P_v, \ul T_v, \ul D_v]
	:=
	\left.
	\cR B_\bH[\tau_v, \ul P_v, \ul T_v, \ul D_v']
	\right|_{\bs D_v}
	\label{eq:Wbar_gray}
\end{equation}
where  $\ul D_v' := \set{D'_v}\cup \set{D_w}_{w \in V(\tau):w>v}$ ($D'_v$ is the map such that $\bs D_v'  := \left. \oplus_{w \in S_{v}} \bs{D}_w \right|_{P_{v}})$ and if $\bs D_v = (D_1,\ldots, D_n)$ with $n= |P_v|$ and $|| \bs D_v||_1= p$, $\left.\cR B_\bH[\tau_v, \ul P_v, \ul T_v, \ul D_v']\right|_{\bs D_v}$ denotes the restriction of $(\cR (B_\bH[\tau_v, \ul P_v, \ul T_v, \ul D_v']))_{n,p}$ to that specific choice of derivative label; if $m_v = 0$, we let $K_v(\bs{\Psi}_v, \emptyset)  = K_v(\bs{\Psi}_v)$, which is given by the expressions in \cite[Eqs.~(3.2.14)--(3.2.17)]{AGG_CMP} in the half-plane limit, namely: if $h_{v_0} = 1$, then $v$ is necessarily an endpoint either of type \tikzvertex{vertex} or \tikzvertex{vertex,E}, with value 
$ K_v(\bs{\Psi}_v):=  W_\infty^{(1)}(\bs{\Psi}_v)$ or $W_{\bH,E}^{(1)}(\bs{\Psi}_v)$ (see \cref{WE}), respectively;
if $h_{v_0} \le 0$, we set
 \begin{equation}
	 K_v(\bs{\Psi}_v)
	:=
\begin{cases}
	 v_{h_v-1} \cdot F_{B} (\bs{\Psi}_v)
		&
		\textup{if } v \in V_e(\tau) \textup{ is of type  \tikzvertex{ctVertex} (so that $h_v = h_{v_0}+1$)},
		\\
		\cR_{E} C^{(h_v-1)}_{\bH,E} (\bs{\Psi}_v)&
		\textup{if } v \in V_e(\tau) \textup{ is of type  \tikzvertex{ctVertex,E} (so that $h_v = h_{v_0}+1$)},\\
		(\cR_{\infty} W_{\infty}^{(1)}) (\bs{\Psi}_v)    &
		\textup{if } v \in V_e(\tau) \textup{ is of type \tikzvertex{vertex} (so that $h_v = 2$)},\\
		\cR_{E} (W_{\bH,E}^{(1)}+ \mathcal{E}_\bH W^{(1)}_{\infty})(\bs{\Psi}_v)  &
		\textup{if } v \in V_e(\tau) \textup{ is of type \tikzvertex{vertex,E} (so that $h_v = 2$)},\\

		\bar W_\bH[\tau_v, \ul P_v, \ul T_v, \ul D_v] (\bs{\Psi}_v)
		&
		\textup{if } v \in V(\tau)\setminus V_e(\tau),
	\end{cases}
	\label{KPsi}
\end{equation}
where the values in the r.h.s.\ are obtained by adapting the definitions in \cite[Eq.~(3.2.15)]{AGG_CMP} and  \cite[Eq.~(3.2.17)]{AGG_CMP}  to the half-plane limit. In particular, in the 
first line of \cref{KPsi}, $v_{h} \cdot F_{B} := 2^h \nu_h F_{\nu, B} + \zeta_h F_{\zeta,B} + \eta_h F_{\eta,B}$, where $\nu_h, \zeta_h,\eta_h$, called the \textit{running coupling 
constants}, only depend on $W^{(h)}_\infty$ and are the same defined in \cite[Eqs. (4.2.30), (4.5.2)]{AGG_AHP} and studied in \cite[Section 4.5]{AGG_AHP}, and 
$F_{\nu, B}, F_{\zeta,B}, F_{\eta,B}$ are the kernels of the \textit{local potentials} in the first two lines of \cite[Eq.(3.1.47)]{AGG_CMP} adapted to the half-plane; in the second line of \cref{KPsi}, $C_{E,\bH}^{(h)} = \sum_{k = h}^0 \sum^{\tikzvertex{vertex}}_{\substack{\tau \in \mathcal{T}^{(k)}}} \mathcal{E}_\bH W_\infty[\tau]$
(cf.\ the first expression in  \cite[Eq.~(3.2.18)]{AGG_CMP}) where the sum runs over trees with black dotted root, $W_\infty[\tau]$ is defined via \cite[Eqs.~(3.2.6)-(3.2.7)]{AGG_CMP}
and, letting $I_\bH(\boldsymbol{\Psi}):=\mathds 1(\boldsymbol{\Psi}\in\mathcal M_\bH)$, $\mathcal{E}_\bH W_\infty[\tau] = \cL_B(I_\bH\, W_\infty[\tau])-(\cL_\infty W_\infty[\tau])I_\bH+
\cR_B(I_\bH\, W_\infty[\tau])-(\cR_\infty W_\infty[\tau])I_\bH$, where $\cL_\infty,\cR_\infty$ are the infinite plane localization and renormalization operators defined in \cite[Sect.4.2]{AGG_AHP}, while $\cL_B,\cR_B$ are the analogous operators adapted to the half-plane (defined as in \cite[Sect.3.1.2]{AGG_CMP}, with the cylinder $\Lambda$ used there replaced by the half-plane); in the fourth line, $\mathcal R_E$ is the edge renormalization operator, defined as in \cite[Sect.3.1.2]{AGG_CMP}, with the cylinder $\Lambda$ used there replaced by the half-plane, $W^{(1)}_{\bH,E}$ was defined in \eqref{WE}, and $\mathcal E_\bH W^{(1)}_\infty= \cL_B(I_\bH\, W_\infty^{(1)})-(\cL_\infty W_\infty^{(1)})I_\bH+
\cR_B(I_\bH\, W_\infty^{(1)})-(\cR_\infty W_\infty^{(1)})I_\bH$; in the last line $\bar W_\bH[\tau_v, \ul P_v, \ul T_v, 
\ul D_v]$ is defined via the analogue of \eqref{eq:Wbar_gray}, with $\cR B_\bH$ in the right side replaced by a similar contribution, whose specific definition depends on whether the 
vertex is black or white and, 
more precisely: 
(1) if $E_v=0$ (i.e.\ if $v$ is black vertex), then $\bar W_\bH[\tau_v, \ul P_v, \ul T_v, \ul D_v] = \cR_\infty W_\infty[\tau_v, \ul P_v, \ul T_v, \ul D_v']\big|_{{\boldsymbol{D}_v}}$; 
(2) if $E_v=1$ (i.e.\ if $v$ is a white vertex), then $\bar W_\bH[\tau_v, \ul P_v, \ul T_v, \ul D_v] = \cR_E W_\bH[\tau_v, \ul P_v, \ul T_v, \ul D_v']\big|_{{\boldsymbol{D}_v}}$, 
with $W_\bH$ defined as in \cite[Eq.~(3.2.12)]{AGG_CMP} or as in \cite[Eq.~(3.2.13)]
{AGG_CMP}, depending on whether $v$ is followed only by black vertices or not (once again, the definitions of \cite{AGG_CMP}, written there for the case of a finite cylinder, should be adapted here to the half-plane by taking the appropriate limit).

Furthermore, a tree expansion analogous to \eqref{Btau} also holds for the kernel of the single-scale contribution to the generating function in Eq.~\eqref{wh-1}, namely
\begin{equation}\begin{aligned}
&w_\bH^{(h)} \sim \sum_{\substack{\tau \in \mathcal{T}^{(h)}_{\Bs}:\\ {m_{v_0}\ge1} } }^{*}w_\bH[\tau]\,, \,\, \mbox{ with  }\,\, w_\bH[\tau] = \sum_{\substack{\underline{P} \in \mathcal{P}(\tau):\\ P_{v_0} = \emptyset}} \sum_{\underline{T} \in \mathcal{S}(\tau,\underline{P})}  \sum_{\underline{D} \in \mathcal{D}(\tau,\underline{P})} B_\bH[ \tau, \underline{P},\underline{T},\underline{D} ]\,,\label{wtree}
\end{aligned}
\end{equation}
where in the first sum $*$ indicates that the sum is over the trees with dotted $v_0$  (cf.\ second line of \cite[Equation~(3.2.5)]{AGG_CMP}). The analogous expansion for the interaction kernel in Eq.~\eqref{Vheff} is similar to the one in the first line of \cite[Eq.~(3.2.5)]{AGG_CMP}, in which the trees must be considered without energy sources and with coordinates in the half-plane rather than in the cylinder: they coincide with the trees introduced in this section if the root has no boundary spin sources.

As in \cite[Remark~3.13]{AGG_CMP} and \cite[Remark~4.4]{AGG_AHP},  we let $R_v=\| \bs D_v\|_1-\sum_{w\in S_v}\| \bs D_w\big|_{P_v}\|_1\equiv \| \bs D_v\|_1-\| \bs D_v'\|_1$.
\begin{remark}
\label{Rv}
From the definitions, $R_{v_0}=0$, while if $v \in V(\tau)\setminus V_e(\tau)\setminus \set{v_0}$, $R_v$ is non vanishing and equal to $1$ only if $|P_v| = m_v = 1$ and  $||\bs{D}'_v|| = 0$, and, if $m_v=0$, it takes the values specified in \cite[Remark~3.13]{AGG_CMP} at $m_v=0$ (note the change in notation: in \cite{AGG_CMP} $m_v$ indicates the 
number of energy source endpoints following $v$ on $\tau$, while here it indicates the number of boundary spin source endpoints following $v$ on $\tau$; of course, for the vertices such that $\tau_v$ has no source endpoints, the values of $R_v$ here and in \cite{AGG_CMP} are the same). If we let 
\begin{equation}
	d_v
	:=
	2 - \tfrac12 |P_v| - \|\bs D_v \|_1  - \tfrac12 m_v -E_v
	\label{eq:scaling_d}\,,
\end{equation}
be the scaling dimension of $v$, then these conventions are exactly such that 
\begin{equation}
	R_v
	=  
	\max \set{d_v+1, 0}\,. \label{Rvdv}
\end{equation}
Note that if $m_v \ge 1$, $E_v=1$ so that \eqref{eq:scaling_d} gives $d_v = 1- \frac{|P_v|}{2}-||\bs{D}_v||_1 - \frac{m_v}{2}$, which is the same definition we introduced in Remark~\ref{sd}; if $m_v = 0$,  \eqref{eq:scaling_d} coincides with the expression in \cite[Eq.~(3.2.24)]{AGG_CMP} at $m_v=0$.

From the definitions, in particular from the definitions of the allowed $\ul D$ (see \cite[Remark~4.4]{AGG_AHP}), for all $v \in V(\tau) \setminus \set{ v \in V_e(\tau): h_v = h_{v'}+1}$, $d_v \le -1$: this will play an important role in the uniformity in $h_{v_0}$ of the expansion in GN trees.
 \end{remark}

\subsubsection{Bounds on the kernels of the boundary spin source.}
In this section, we state two key bounds on the kernels $B_\bH[\tau, \ul P, \ul T, \ul D]$, which will be used later in order to estimate the logarithmic derivatives of the generating function of boundary spin correlations, once we express them in the form of sums over GN trees. These bounds are obtained by adapting those of \cite[Section~3.3]{AGG_CMP} to the presence of boundary spin sources in the  half-plane  limit.

We can repeat the same argument as in the proof of \cite[Proposition~3.19]{AGG_CMP} to obtain the following proposition\footnote{Proposition 3.9 of \cite{AGG_CMP} was formulated and proved for $\theta=3/4$: this restriction is unnecessary, cf.~\cite[Footnote 19]{AGG_CMP}, and was made only to mildly simplify the explicit choice of a few constants 
here and there in its proof. It can be easily checked that the proof of \cite[Proposition~3.19]{AGG_CMP} can be generalized to any $\theta\in (0,1)$ at the cost of 
having a $\theta$-dependent constant $C=C(\theta)$ in the right hand side of \cite[Eq.(3.3.31)]{AGG_CMP}, which may diverge as $\theta\to 1$. Similar considerations apply to Prop.~3.17, Prop.~3.23 and to the discussion of Sect.~4 of \cite{AGG_CMP}. We take the occasion to mention that the way in which \cite[Theorem 1.1]{AGG_CMP} is formulated is 
slightly mistaken: the remainder $R_\Lambda(\bs y)$, which is well defined and analytic in $\lambda$ for $|\lambda|<\lambda_0$, uniformly in $\Lambda,\bs y$, 
satisfies the bound \cite[Eq.(1.6)]{AGG_CMP} only in an a priori smaller domain, $|\lambda|<\lambda_1(\theta,\varepsilon)$, with $\lambda_1:(0,1)\times(0,\frac12)\to (0,+\infty)$ 
a function which may vanish in the limit as $\theta\to 1$ and/or $\varepsilon\to 0$. Similar considerations are valid for \cite[Theorem 1.1]{AGG_AHP} and 
\cite[Theorems 1.1 and 1.2]{GGM}.}
\begin{proposition} Let $B_\bH[\tau, \ul P, \ul T, \ul D]$ be inductively defined as in Section~\ref{T}. Given $\theta\in[0,1)$, there exist $C,\kappa,\lambda_0>0$ such that, for any $\tau\in \cup_{h ^*-1 \le h \le -1} \cT^{(h)}_{\Bs}$ with $m_{v_0}\ge 1$,  
$\bs y_{v_0} \in (\mathcal{Q}(\partial \bH))^{m_{v_0}}$,
$\ul P\in \cP(\tau)$, $\ul T\in \cS(\tau,\ul P)$, $\ul D\in \cD(\tau,\ul P)$, and $|\lambda|\le \lambda_0$, 
\begin{equation}\begin{aligned}
	\|B_\bH [\tau,\ul P,\ul T,\ul D](\bs y_{v_0})\|_{h_{v_0}}
	\le \, & C^{m_{v_0} + \sum_{v\in V_e(\tau)}|P_v|}    \Big(\prod_{v \in V(\tau)\setminus V_e(\tau)}
		      \frac{2^{(\frac{|Q_v|}2+ \sum_{w\in S_v}\| \bs D_w |_{Q_v} \|_1 - R_v + 2(1 -  \left|S_v\right|))h_v}}{ |S_v|!}  \Big)  \\
	&\times \Big(\prod_{\substack{v\in V(\tau)\setminus \set{v_0}:\\ m_v=0}}2^{-E_v(h_v-h_{v'})} \Big)  \Big(\prod_{\substack{v\in V(\tau):\\
		    m_v\ge 1}} 2^{2[|S^*_v|-1]_+h_v} e^{-{\frac{c_0}{12}}2^{h_v}\delta(\bs y_v)} \Big) \\
		    &\times \prod_{\substack{v \in V_e(\tau)}}\piecewise{
				|Z_{\Bs,h_v-1}|  & \text{ if $m_v=1$ }\\
		|\lambda|^{\max\{1,\kappa |P_v|\}} 2^{(2-\frac{|P_v|}{2}-||\bs D_v||_1)h_v}2^{\theta h_v}  & \text{ if $m_v=0$}} 
			      \label{BPD} \end{aligned}	 \end{equation}
			      where: in the last product in the second line, $S_v^*=\{w\in S_v: m_w \ge 1\}$, $[\cdot]_+$ is the positive part, and $c_0$ is the same constant as in Eq.~\eqref{gdecay} (and then as in \cite[Equation~(3.8)]{AGG_CMP}). 
\label{prop1}
\end{proposition}

Aside from noting that the norm used here coincides with the one appearing in \cite[Proposition~3.19]{AGG_CMP} for the kernel in question, the only difference between the proof of 
\cite[Proposition~3.19]{AGG_CMP} and that of Proposition \ref{prop1} is due to the presence of the endpoints with external fields (which are now \tikzvertex{FSSpinEP}), whose
value $K_v$ is defined in the first line of \eqref{eq:K_gray_def}. These can be treated analogously to the energy observable endpoints in the proof 
of \cite[Proposition~3.19]{AGG_CMP}, using
$$\sum_{z: P_v \to \bar{\bH}}  e^{\alpha c_02^{h_{v_0}}\delta(\bs{z}_v,\bs{y}_v)}|K_v(\bs{\Psi}_v,\bs{y}_v)|=  | Z_{\Bs,h_v-1}|$$
for any $\alpha>0$ in place of \cite[Eq.~(3.3.37)]{AGG_CMP} and the previous inequality.

We now make a rearrangement similar to the one leading to \cite[Proposition~3.23]{AGG_CMP}, bearing in mind however the different definition of the scaling dimension $d_v$ for kernels with boundary spin sources.
\begin{proposition}\label{prop2} Under the same assumptions as Proposition~\ref{prop1},
\begin{equation}\label{eq:WL_bigbound}
	\begin{aligned}
	& \|B_\bH[\tau,\ul P,\ul T,\ul D](\bs y_{v_0})\|_{h_{v_0}}  \le     
		\frac{C^{m_{v_0}+\sum_{v\in V_e(\tau)}|P_v|}}{|S_{v_0}|!}
		\, 2^{h_{v_0}d_{v_0} }    
		  \Big(
			\prod_{v \in V(\tau)\setminus\set{v_0}}\frac1{|S_{v}|!}
		2^{(h_v-h_{v'}) d_v}  \Big) 
		\\
		&\qquad \times 
		\Big(\prod_{\substack{v\in V(\tau):\\
		    m_v\ge 1}} 2^{[|S^*_v|-1]_+h_v} e^{-{\frac{c_0}{12}}2^{h_v}\delta(\bs y_v)}\Big)
				 \prod_{\substack{v \in V_e(\tau)}}\piecewise{
				 |Z_{\Bs,h_v-1}| & \text{ if $m_v=1$ }\\
		|\lambda|^{\max\{1,\kappa |P_v|\}} (C_{\theta_v})^{\mathds 1(h_v<0)} 2^{\theta_v h_v}  & \text{ if $m_v=0$}}
		\,,
	\end{aligned}
\end{equation}
with $d_v = 2 - \tfrac12 |P_v| - \|\bs D_v \|_1  - \tfrac12 m_v -E_v$ as in Eq.~\eqref{eq:scaling_d}.
\end{proposition}

\begin{proof}
	Comparing \cref{BPD,eq:WL_bigbound}, 
	evidently the result follows immediately if we show that
	\begin{equation}\label{D1}
		\begin{split}
			&
			\sum_{v \in V(\tau)\setminus V_e(\tau)}
			\left(\frac{|Q_v|}2+ \sum_{w\in S_v}\| \bs D_w |_{Q_v} \|_1 - R_v + 2(1 -  \left|S_v\right|)\right)h_v
			-
			\sum_{\substack{v\in V(\tau)\setminus \set{v_0}:\\ m_v=0}}
			 E_v (h_v-h_{v'})
			\\ & \quad
			+
			\sum_{\substack{v\in V(\tau):\\  m_v\ge 1}}
			2[|S^*_v|-1]_+ h_v
			+
			\sum_{\substack{v \in V_e(\tau) \\ m_v = 0}}
			\left(2-\frac{|P_v|}{2}-\|\bs D_v\|_1\right)h_v 
			\\ & =
			d_{v_0} h_{v_0}			
			+
			\sum_{v \in V(\tau)\setminus\set{v_0}}
			d_v(h_v-h_{v'}) 
			+
			\sum_{\substack{v\in V(\tau):\\ m_v\ge 1}}
			[|S^*_v|-1]_+ h_v
			.
		\end{split}
	\end{equation}
Proceeding as described in Eqs.~\cite[Eqs.~(4.4.17)-(4.4.21)]{AGG_AHP} we first note that
\begin{equation}\begin{aligned}\label{rewAGG}
&\sum_{v \in V(\tau)\setminus V_e(\tau)}
\left(\frac{|Q_v|}2+ \sum_{w\in S_v}\| \bs D_w |_{Q_v} \|_1 - R_v + 2(1 -  \left|S_v\right|)\right)h_v \\
&\qquad =\left(2-\frac{|P_{v_0}|}{2}- \|\bs D_{v_0}\|_1\right)h_{v_0} + \sum_{v \in V(\tau)\setminus \set{v_0}} {\left(2-\frac{|P_{v}|}{2}- \|\bs D_{v}\|_1\right)(h_{v}-h_{v'})}\\
&\qquad - \sum_{v \in V_e(\tau)} \left( 2 - \frac{|P_v|}{2} - \|\bs D_v \|_1 \right) h_{v}
\end{aligned}
\end{equation}
and, using also the fact that $m_{v_0} h_{v_0}+	\sum_{v \in V(\tau)\setminus \set{v_0}} m_v (h_{v}-h_{v'})=\sum_{\substack{v \in V_e(\tau) \\ m_v = 1}}  h_{v}$, 
\begin{equation}\label{following}
	\sum_{\substack{v \in V(\tau) \\ m_v \ge 1}}[|S_v^*|-1]_+ h_v
	=\sum_{\substack{v \in V(\tau)\setminus V_e(\tau) \\ m_v \ge 1}}(|S_v^*|-1) h_v
	=-h_{v_0}-	\sum_{\substack{v \in V(\tau)\setminus\{v_0\} \\ m_v \ge 1}} (h_v - h_{v'})+\sum_{\substack{v \in V_e(\tau) \\ m_v = 1}} h_v	,\end{equation}
so the l.h.s. of \cref{D1} reduces to
\begin{equation}\label{rewrit}
	\begin{split}
&\left(1-\frac{|P_{v_0}|}{2}- \|\bs D_{v_0}\|_1\right)h_{v_0} \\
&+ \sum_{v \in V(\tau)\setminus \set{v_0}} \left(2-\frac{|P_{v}|}{2}- \|\bs D_{v}\|_1 - E_v\mathds 1(m_v=0)-\mathds 1(m_v=1) \right)(h_{v}-h_{v'})	\\ 
&-\sum_{\substack{v \in V_e(\tau) \\ m_v = 1}} \left(1 - \frac{|P_v|}{2} - \|\bs D_v \|_1 \right) h_{v} + \sum_{\substack{v\in V(\tau):\\ m_v\ge 1}}
			[|S^*_v|-1]_+ h_v.\end{split}\end{equation}			
Recall that the endpoints with $m_v=1$ have $|P_v|=1$ and $\bs D_v=\bs 0$, so the first term in the last line is equal to $-\frac12\sum_{\substack{v \in V_e(\tau) \\ m_v = 1}}h_{v}$. Using again the identity in the line preceding \eqref{following}, this can be rewritten as $-\frac12{m_{v_0}}h_{v_0}-\frac12	\sum_{v \in V(\tau)\setminus \set{v_0}} m_v (h_{v}-h_{v'})$. Plugging this back into \eqref{rewrit}, and recalling that $E_v=1$ for all vertices with $m_v\ge 1$, including $v_0$, we get the desired identity, \eqref{D1}, with $d_v = 2 - \tfrac12 |P_v| - \|\bs D_v \|_1  - \tfrac12 m_v -E_v$.\end{proof}

\medskip

\cref{prop2} can be used to bound the sums in \eqref{Btau}, in analogy with \cite[Lemma~4.8]{AGG_AHP}, as well as the beta function controlling the flow of the running coupling constants, in particular $Z_{\Bs, h}$, as discussed in the following subsection.  

\subsubsection{Beta Function Equation and analyticity of $Z_\Bs$.}\label{beta}

The definition of the running coupling constant $Z_{\Bs,h}$ in Eq.~\eqref{ZbsL}, combined with the GN tree expansion for the boundary spin source kernels derived in \cref{T} (cf.\ first expression in Eq.~\eqref{Btau}), implies that the running coupling constants $\underline{Z}_{\Bs} := \set{{Z}_{\Bs,h}}_{h \le -1}$  satisfy the following equation, for all $h \le -1$ and for any $y  \in \partial \bH$: 
\begin{equation} \label{Ztau}Z_{\Bs,h} = \sum_{\substack{\tau \in \cT_{\Bs}^{(h)}:\\m_{v_0}=1}}  \sum_{z'\in {\bH}}
	  (B_\bH[\tau])_{1,0,1}( (-,0,z'), y)\,,
\end{equation}
where, in view of the iterative definition of $B_\bH[\tau]$, see Eqs.~\eqref{Btau}-\eqref{eq:W_Lambda_gray_iter}, the right hand side of this equation depends on the restriction of the sequence $\underline{Z}_{\Bs}$ to the scale indices larger than $h$, i.e., upon $\set{{Z}_{\Bs,k}}_{k > h}$. Therefore, it is a recursive equation for the components of $\underline{Z}_\Bs$: given the initial value $Z_{\spin,0} = Z^{-1/2}$ we can reconstruct the values at all the lower scales. We shall see that the result is bounded uniformly in the scale index and 
$Z_{\Bs,h}$ admits a limit as $h \to -\infty$.

Note that in r.h.s.\ of \eqref{Ztau}, the sum is over trees with exactly one \tikzvertex{FSSpinEP} endpoint: among these we can distinguish the trees that have the \tikzvertex{FSSpinEP} endpoint as the only endpoint and the trees that have, in addition to the \tikzvertex{FSSpinEP} endpoint, also other endpoints of other types. In the first case (\tikzvertex{FSSpinEP} as the only endpoint), there is only the `trivial' tree with (undotted) root on scale $h+1$ and the  \tikzvertex{FSSpinEP} endpoint on scale $h+2$ which gives contribution $Z_{\Bs,h+1}$ (the same tree in which the root is on scale $h+1$ and the \tikzvertex{FSSpinEP} endpoint on scale $h+3$ or higher, gives zero contribution since an operator $\cR$ acts on scale $h+2$, and so was excluded from $\cT_\spin$ in \cref{T}). Then, we can write  Eq.~\eqref{Ztau}  as
\begin{equation} \label{Bflow} Z_{\Bs,h} = Z_{\Bs,h+1} + B_{h+1}[\underline{Z}_\Bs]\,, \end{equation}
where 
\begin{equation}
	B_{h+1}[\underline{Z}_\Bs]
	:=
	 \sum^*_{\substack{\tau \in \cT_{\Bs}^{(h)}:\\m_{v_0}=1}}  \sum_{z'\in {\bH}}
	 (B_\bH[\tau])_{1,0,1}( (-,0,z'), y)\,,
	\label{eq:beta_def}
\end{equation}
where the $*$ on the sum indicates the constraint that the sum is over trees with a dotted root.
The function $B_{h+1}[\underline{Z}_\Bs]$ is called \textit{beta function} and Eq.~\eqref{Bflow} is called the \textit{beta function flow equation} for the running coupling constants $\underline{Z}_\Bs$. Eq.\eqref{eq:beta_def} can be bounded via \cref{prop2}: proceeding as in  \cite[Lemma~4.8]{AGG_AHP} (and \cite[Remark~3.24]{AGG_CMP}), we see that under the same assumptions as Proposition~\ref{prop2}, for $h \le -1$ and any $\theta \in (0,1)$, there exists $C_{\theta}>0$ such that 
\begin{equation}\label{shortmemory}
\begin{aligned}
2^{-\theta h}\sum^{*}_{\substack{\tau \in \mathcal{T}^{(h)}_{\Bs}:\\ m_{v_0}= 1}} \sum_{\substack{\underline{P}\in \mathcal{P}(\tau):\\ |P_{v_0}|=1}}\sum_{\substack{\underline{T}\in \mathcal{S}(\tau,\underline{P})}} \sum_{\substack{\underline{D}\in \mathcal{D}(\tau\,\underline{P}):\\ ||\bs D_{v_0}||_1=0}}& || (B_\bH[ \tau,\underline{P},\underline{T},\underline{D}])_{1,0,1}||_{h+1}\\ & \qquad   \le C_\theta |\lambda| \left( \max_{h < h'\le -1}  Z_{\Bs,h'}\right) \,,
\end{aligned}
\end{equation}
where  the norm in the left hand side is the one defined in \eqref{normHH} (note that since $m_{v_0}=1$ there is no dependence on the boundary spin position $y_{v_0} = ((y_{v_0})_1,1)$, due to translational invariance in the horizontal direction). Note that the bound proportional to $|\lambda|$ is due to the fact that the trees contributing to the left hand side must have at least one endpoint associated with the interaction term in addition to the boundary spin endpoint. 

Using the bound \cref{shortmemory} in \eqref{Bflow}-\eqref{eq:beta_def}, we first see iteratively that $|Z_{\spin,h}|$ remains bounded,
and then also that for $k < h \le 0$
\begin{equation}
	\left|Z_{\Bs,k} - Z_{\Bs,h} \right|
	\le
	\sum_{j = k+1}^h \left| B_j [\underline{Z}_\Bs]\right|
	\le C_\theta'|\lambda| 2^{\theta h},
	\label{eq:Zs_cauchy}
\end{equation}
for some $C_\theta'>0$, so that $ \underline{Z}_\Bs$ is a Cauchy sequence and the limit $Z_{\Bs} := Z_{\Bs}(\lambda) = \lim_{h \to -\infty} Z_{\Bs,h}$ exists with
\begin{equation}
	|Z_{\Bs,h} - Z_{\Bs}|
	\le C_\theta''
	  |\lambda| 2^{\theta h}\,,
	\label{eq:Zs_finite_scale}
\end{equation}
for any $\theta\in(0,1)$ and some $C_\theta''>0$. Note also that, from the analyticity of $Z=Z(\lambda)$ (see  \cite[Proposition 4.11]{AGG_AHP}) and that of $B_{j}[\underline{Z}_\Bs]$, thought of as a function of $\lambda$ (which, in turn, follows from the fact that the beta function itself is defined via an absolutely convergent sum over trees, whose values are by definition analytic function of $\lambda$ as well), 
the limiting, renormalized, coupling $Z_\Bs$ is an analytic function of $\lambda$ in a sufficiently small complex circle centered at the origin. 

\section{Correlation functions}\label{sec:corr}

In this section we obtain the expression for the boundary spin correlation functions and conclude the proof of \cref{prop:main}.
The necessary calculations are similar to \cite[Section~4]{AGG_CMP};
note that some of the most prominent changes regard the fact that some of the effective boundary spin source functions have an odd number of massless fields, which loosens an important constraint on the terms appearing in the correlation functions.

We consider $\bs{y} = (y_1, \ldots, y_{m}) \in (\partial \bH)^m$ distinct right-to-left ordered boundary positions corresponding to the locations of the spins $\sigma_{y_1}, \ldots ,\sigma_{y_{m}}$, $m \in 2 \mathbb{N}$, which we intend to compute the correlations of. We let 
\begin{equation}\label{evspinQtrunc}\braket{\s_{y_1}; \cdots; \s_{y_{m}} }_{\l,t;\L} =  \frac{\partial^{m}}{\partial  \varphi_{y_1} \cdots \partial  \varphi_{y_{m}}}
\, \log \Xi_{\l,t;\L} (\bs{\varphi})  \Big|_{\bs{\varphi}=\0}
\end{equation}
be the {\it truncated} spin correlations in finite volume, and as usual $\braket{\s_{y_1}; \cdots; \s_{y_{m}} }_{\l,t;\bH}$ the corresponding half-plane limit. Note that, due 
to the Grassmann nature of the $\bs{\varphi}$ field, Eq.\eqref{evspinQtrunc} defines {\it fermionic} truncated correlations. For example, recalling that in our case the expectation of an 
odd number of spin observables vanishes:  $\braket{\s_{y_1};\s_{y_{2}} }_{\l,t;\bH}= \braket{\s_{y_1}\s_{y_{2}} }_{\l,t;\bH}$, 
\begin{equation}\begin{split} \braket{\s_{y_1};\s_{y_{2}};\s_{y_{3}};\s_{y_{4}}}_{\l,t;\bH}&=  \braket{\s_{y_1}\s_{y_{2}}\s_{y_{3}}\s_{y_{4}}}_{\l,t;\bH}- \braket{\s_{y_1}\s_{y_{2}}}_{\l,t;\bH}
 \braket{\s_{y_3}\s_{y_{4}}}_{\l,t;\bH}\\
 & -\braket{\s_{y_1}\s_{y_{4}}}_{\l,t;\bH} \braket{\s_{y_2}\s_{y_{3}}}_{\l,t;\bH}+\braket{\s_{y_1}\s_{y_{3}}}_{\l,t;\bH}
 \braket{\s_{y_2}\s_{y_{4}}}_{\l,t;\bH},\end{split}\label{eq:4.2}\end{equation}
 and similarly for higher order truncated correlations (notice the plus sign in front of the last term). 
Of course, simple correlations can iteratively be expressed in terms of truncated correlations and vice versa (starting from the two-spins case, for which $\braket{\s_{y_1};\s_{y_{2}} }_{\l,t;\bH}= \braket{\s_{y_1}\s_{y_{2}} }_{\l,t;\bH}$, next inverting \eqref{eq:4.2}, then proceeding iteratively in $m$ for the case of $m$-spins correlations, $m\ge 4$), and in order to prove the theorem we will use this correspondence.

Recall that $\log\Xi_{\l,t;\Lambda}(\bs{\varphi})$ in \eqref{evspinQtrunc} can be expressed via \eqref{eventually}, so that, in the half-plane limit, 
\begin{equation}\braket{\s_{y_1}; \cdots; \s_{y_{m}} }_{\l,t;\bH}=\sum_{h\le -1} \frac{\partial^{m}}{\partial  \varphi_{y_1} \cdots \partial  \varphi_{y_{m}}}
\, \cW_\bH^{(h)}(\bs{\varphi}) \Big|_{\bs{\varphi}=\0},\end{equation}
where $\cW^{(h)}_\bH(\bs{\varphi})$ can be expressed as in Eq.~\eqref{wh-1} with the kernel $w_\bH^{(h)}$ expressed via the GN tree expansion \eqref{wtree}. Assembling all this, we have
\begin{equation}
	\langle \sigma_{y_1}; \cdots; \sigma_{y_{m}} \rangle_{\lambda,t;\bH}
	=(-1)^{m/2}
	\sum_{\pi} (-1)^\pi\sum_{h\le-1}
	\sum_{\substack{\tau \in \cT^{(h)}_{\Bs}: \\ m_{v_0} = m}}^*w_\bH[\tau](\pi(\bs y))
	\label{corr}
\end{equation}
where the first sum runs over the $m!$ permutations $\pi$ of the $m-$tuple $\bs{y}$, $(-1)^{\pi}$ is the corresponding sign, and the $*$ on the last sum indicates the constraint that the root $v_0$ is dotted. 
This is the analogue of \cite[Eq.~(4.1)]{AGG_CMP}. However, comparing with that equation, note that in \eqref{corr} there is no term corresponding to what in \cite[Eq.~(4.1)]{AGG_CMP} was denoted $w_\Lambda(\pi(\bs y))$, the reason being that the interaction term $\cV_\L^{\rm int}$ in Eq.~\eqref{XiL} is independent of $\bs{\varphi}$.

In order to prove \cref{prop:main}, we need to compare \eqref{corr} with the rescaled, critical, integrable version of $\langle \sigma_{y_1}; \cdots; \sigma_{y_{m}} \rangle_{\lambda,t;\bH}$, namely  with
\begin{equation}\label{corrsl}
	(Z_{\Bs})^{m}
	\langle \sigma_{y_1};\cdots; \sigma_{y_{m}} \rangle_{0,t^*; \bH}\,,
\end{equation}
where $Z_{\Bs}=Z_{\Bs}(\l)$ is the limiting value of $Z_{\Bs,h}$ (see before \eqref{eq:Zs_finite_scale}): in particular, we will prove that the difference 
$\langle \sigma_{y_1}; \cdots; \sigma_{y_{m}} \rangle_{\lambda,t;  {\bH}}-  (Z_{\Bs})^{m}
\langle \sigma_{y_1};\cdots ;\sigma_{y_{m}} \rangle_{0,t^*;\bH}$ 
can be bounded as the remainder of Eq.~\eqref{10b}. To compare the two correlations, we first express \eqref{corrsl} as a sum like the one in the r.h.s.\ of \eqref{corr}: in fact this can be 
computed from the derivatives of $\log \Xi^\qf_{0,t^*;\bH}$, where $\Xi^\qf_{0,t^*;\bH}  \propto \exp \Big(\sum_{h\le -1} \cW_{\bH,\qf}^{(h)}(\bs{\varphi})\Big)$ can be obtained with the same procedure with which the half-plane limit of Eq.~\eqref{eventually} was obtained, but starting from a `quasi-free' generating function $\Xi^\qf_{0, t^*;\L} \propto \int P_c^* (\cD\phi)P_m^* (\cD\xi) e^{Z_{\Bs}\sum_{z \in \partial^l \L }\phi_{-,z}\varphi_{z}}$ instead of the one in Eq.~\eqref{eq:startfrom}. Then we can obtain  $\cW_{\bH,\qf}^{(h)}(\bs{\varphi})=\sum_{\substack{\bs{y} \in \mathcal{Q}(\partial \bH)}}w_{\qf}^{(h)}(\bs{y}) {\varphi}(\bs{y})$ (cf.\ Eq.~\eqref{wh-1}) and express the kernel $w_{\qf}^{(h)}$ via the GN tree expansion analogous to the one in \cref{wtree} 
with $\lambda =0$, which is extremely simple: the only endpoints are of type \tikzvertex{FSSpinEP} with 
$K_{v} (\bs{\Psi}_{v},\bs{y}_{v}) =  Z_{\Bs}\delta_{\bs{\Psi}_{v},(-,0,y)}\delta_{\bs{y}_v,y}$.
Since the \tikzvertex{FSSpinEP} endpoint $v$ has $|P_v|=1$, without other types of endpoints only one tree is possible with a dotted root on any given scale, the one that has two \tikzvertex{FSSpinEP} endpoints, $v_1$ and $v_2$, on scale $h_{v_1}=h_{v_2}=h+2$, and root  $v_0$ on scale $h_{v_0}=h+1$, which automatically has $P_{v_0}=0$. Let $\tau_{2,h}$ be this tree: the expression in \eqref{corrsl} for $m=2$ then becomes   
\begin{equation} \begin{aligned}\label{wfreetree}	
&(Z_{\Bs})^{2}
	\langle \sigma_{y_1} ;\sigma_{y_{2}} \rangle_{0,t^*;\bH} = - \sum_{\pi}(-1)^\pi \sum_{h\le-1}
 w_\qf [\tau_{2,h}](\pi(y_1,y_2))\,,\end{aligned}
\end{equation}
where $w_\qf[\tau_{2,h}]$ can be obtained via the same recursive definition of the second expression in Eq.~\eqref{wtree}, computed at $\tau = \tau_{2,h}$, i.e.\ $w_\qf[\tau_{2,h}] =  
B_\qf[\tau_{2,h}, \ul P_2, T_2, \ul 0]$, where $\ul P_2$ is the unique element of $\cP(\tau_{2,h})$ giving a nonvanishing contribution, and $T_2$ denotes the unique element of 
$\cS(\tau,\ul P_2)$; more concretely, 
$w_\qf[\tau_{2,h}] = -(Z_{\Bs})^{2} g^{(h)}_{\bH,--}(y_1,y_2)$, with $g^{(h)}_{\bH,--}(y_1,y_2)$ as in \eqref{g0h}-\eqref{gceta}. On the other hand, if $m>2$ then
the truncated correlations \eqref{corrsl} vanish, 
which is equivalent to saying that the standard (non-truncated) boundary spin correlations at $\lambda=0$ have a Pfaffian structure, as well known in the theory of the nearest neighbor 
2D Ising model. 

In view of \eqref{corr}, \eqref{wfreetree} and these considerations, we rewrite $$\langle \sigma_{y_1}; \cdots ;\sigma_{y_{m}} \rangle_{\lambda,t;\bH} -(Z_{\Bs})^m
\langle \sigma_{y_1};\cdots; \sigma_{y_{m}} \rangle_{0,t^*;\bH}$$ for $m=2$, as
\begin{equation}\begin{aligned} \label{s1s2}
	&\langle \sigma_{y_1}  ;  \sigma_{y_{2}} \rangle_{\lambda,t;\bH} - (Z_{\Bs})^2
	\langle  \sigma_{y_1}   ; \sigma_{y_{2}} \rangle_{0,t^*;\bH}  =-
	\sum_{\pi}(-1)^\pi\\
	&\times\Big\{
	\sum_{h\le-1} 
	\Big(
		w_\bH [\tau_{2,h}](\pi( y_1,y_2))
		- 
		w_\qf[\tau_{2,h}](\pi( y_1,y_2))
	+  \sum^*_{\substack{\tau \in \cT_\Bs^{(h)}:\ \tau \ne \tau_{2,h},\\ m_{v_0}=2}}
	w_\bH[\tau](\pi( y_1,y_2))	\Big)\Big\}\,,
		\end{aligned}
\end{equation}
 and, for $m > 2$, as
\begin{equation}\begin{aligned} \label{s1sm}
	&\langle \sigma_{y_1}; \ldots  ;\sigma_{y_{m}} \rangle_{\lambda,t;\bH}
	 =(-1)^{m/2}
	\sum_{\pi}(-1)^\pi
		\sum_{h\le-1}\sum^*_{\substack{\tau \in \cT_\Bs^{(h)} :\\ m_{v_0}=m}}
	w_\bH[\tau](\pi(\bs y))\,.
\end{aligned}\end{equation}
Note that both the sum over trees with $m_{v_0} = 2$ other than $\tau_{2,h}$ (last term in the r.h.s.\ of \eqref{s1s2}) and the sum over trees with $m_{v_0}=m > 2$ (in the r.h.s.\ of \eqref{s1sm}) involve only trees that have at least one other endpoint of another type other than \tikzvertex{FSSpinEP}, i.e.\ trees with $|V_e(\tau)| > m$. Then we let
\begin{equation} \label{R1}\begin{split}
	R_1(y_1,y_2):=& \sum_{h\le-1} 
	\Big(w_\bH [\tau_{2,h}](y_1,y_2)
		- 
		w_\qf[\tau_{2,h}](y_1,y_2)
	\Big)\\
	\equiv & \sum_{h\le -1}\big((Z_{\Bs,h})^2-(Z_\Bs)^2\big)g^{(h)}_{\bH,--}(y_1,y_2),\end{split}
\end{equation}
where for the second identity we used the very definitions of $w_\bH [\tau_{2,h}](y_1,y_2)$ and $w_\qf[\tau_{2,h}](y_1,y_2)$, and 
\begin{equation}\begin{split} \label{R2}
	R_2(\bs y):= & \sum_{h\le -1}\begin{cases} \sum^*_{\substack{\tau\in \cT_\Bs^{(h)}:\\ m_{v_0}=2}}
	w_\bH[\tau]( y_1,y_2) \mathds 1 (\tau\ne\tau_{2,h}) & \text{if $m=2$}\\ \sum^*_{\substack{\tau \in \cT_\Bs^{(h)} :\\ m_{v_0}=m}}
	w_\bH[\tau](\bs y) &  \text{if $m>2$}\end{cases}\\
	=&	\sum_{h\le-1}\sum^*_{\substack{\tau \in \cT_\Bs^{(h)} :\ m_{v_0}=m,\\ |V_e(\tau)|>m}}
	w_\bH[\tau](\bs y)\,,
	\end{split}
\end{equation}
so that 
\begin{equation}\label{R2primo}\begin{split} & \langle \sigma_{y_1}; \cdots ;\sigma_{y_{m}} \rangle_{\lambda,t;\bH}-  (Z_{\Bs})^m
	\langle \sigma_{y_1};\cdots; \sigma_{y_{m}} \rangle_{0,t^*;\bH}\\
	&\qquad \qquad =(-1)^{m/2}\sum_\pi(-1)^\pi\Big(R_1(\pi(\bs y))\mathds 1(m=2)+R_2(\pi(\bs y))\Big)\end{split}\end{equation}
The remainder term $R_1(y_1,y_2)$ can be readily bounded from its explicit expression in the second line of \eqref{R1}: in fact, using the first of \eqref{gdecay} as well as the bound \eqref{eq:Zs_finite_scale} we find that, for any $\theta\in(0,1)$ and some $C_\theta,C_\theta'>0$, 
\begin{equation} |R_1(y_1,y_2)|\le C_\theta\sum_{h\le -1}2^{h(1+\theta)}e^{-c_0 2^h|y_1-y_2|}\le C_\theta'\, |y_1-y_2|^{-1-\theta}.\label{R2secondo}\end{equation}
On the other hand, $R_2(\bs{y})$ can be bounded using the same approach as in \cite[Sections~4.1 and~4.2]{AGG_CMP}, with some modifications as follows.

\medskip

\textit{The remainder term $R_2$.} Compared to \cite[Section~4.2]{AGG_CMP}, here the term $R_2$ is different because the parity constraints on $\ul P$ are different; also in some respects it is more similar to the subdominant terms studied in \cite[Section~4.2]{GGM} since, as we will see, in the presence of boundary spin sources there are no non-endpoint vertices with $d_v \ge 0$.
We can bound $R_2$ of \eqref{R2} as
\begin{equation}\begin{aligned}
| R_2( \bs{y})| \le& \sum_{h \le-1} \sum_{\substack{\tau \in \mathcal{T}^{(h)}_{\Bs}:\\ m_{v_0}=m, |V_e(\tau)| > m}}^* e^{-\frac{c_0}{2} 2^{h_{v_0}} \delta (\bs{y})} 
 \sum_{\substack{\ul{P}\in \mathcal{P}(\tau):\\ P_{v_0} = \emptyset}} \sum_{\ul{T} \in \mathcal{S}(\tau,\ul{P})}  \sum_{\ul{D} \in \mathcal{D}(\tau,\ul{P})}  \| B_\bH[\tau, \ul P, \ul T, \ul D](\bs{y})\|_{h_{v_0}}
\end{aligned}
\end{equation}
where in the r.h.s.\ $h_{v_0}=h+1$ and the norm is defined as in \eqref{normHH}: $\| B_\bH[\tau, \ul P, \ul T, \ul D](\bs{y})\|_{h_{v_0}}$ can be bounded by Eq.~\eqref{eq:WL_bigbound} (noting that if $|P_{v_0}| =0$ also $|D_{v_0}|=0$ and that if $m_{v_0}>0$ necessarily $E_{v_0}=1$). 
Bounding the sum over $\ul T$ using  \cite[Eq.~(4.4.27)]{AGG_AHP}, 
\begin{equation}\begin{aligned}\label{R22}
| R_2( \bs{y})| \le& \sum_{h \le-1} \sum_{\substack{\tau \in \mathcal{T}^{(h)}_{\Bs}:\\ m_{v_0}=m, |V_e(\tau)| > m}} e^{-\frac{c_0}{2} 2^{h_{v_0}} \delta (\bs{y})} %
 \sum_{\substack{\ul{P}\in \mathcal{P}(\tau):\\ P_{v_0} = \emptyset}}  {C^{\sum_{v\in V_e(\tau)}|P_v|}} 2^{h_{v_0}(1-\frac{m}{2})}\\
&\times \sum_{\ul{D} \in \mathcal{D}(\tau,\ul{P})}   \Big(
	      \prod_{v \in V(\tau)\setminus \set{v_0}}	      2^{(h_v-h_{v'}) d_v}  \Big) \Big(\prod_{\substack{v\in V(\tau):\\
		    m_v\ge 1}} 2^{[|S^*_v|-1]_+h_v} e^{-{\frac{c_0}{12}}2^{h_v}\delta(\bs y_v)}\Big)
		    \\
		    & \times
		      \Big(\prod_{\substack{v \in V_e(\tau):\\ m_v=1}}  |Z_{\Bs,h_v-1}| \Big) \Big(\prod_{\substack{v \in V_e(\tau):\\ m_v=0}} |\lambda|^{\max\{1,\kappa |P_v|\}}2^{\theta h_v} \Big)\,,
      \end{aligned}
\end{equation}
where we recall that $h_{v_0}=h+1$ and in general $C$ may depend upon the choice of $\theta$ (a similar comment holds for the other constants $C',C'', \ldots$ below, 
as well as for the re-definitions of $C$ introduced in the following).
This has the same form as \cite[Eq.(4.20)]{GGM} (with $N=0$) or \cite[Eq.(4.2.3)]{AGG_CMP} (without end-point vertices with $m_v =1$ and $d_v=0$), with the differences that 
the factor $2^{h_{v_0}(2-m)}$ in these equations is replaced here by $2^{h_{v_0}(1-m/2)}$, and that $[|S_v^*| -1]_+$ appears here at exponent without a factor of $2$. Moreover, bearing in mind our current definition in Eq.~\eqref{eq:scaling_d},  $d_v$ is now bounded by $d'_v:= \min \{ 1- \frac{|P_v|}{2} - \frac{m_v}{2}, -1 \}$ (cf.\ first line of \cite[Eq.(4.2.4)]{AGG_CMP}). Once this is done, given $\tau$ and $\ul P$, the sum over $\ul D$ can be trivially bounded by the number of elements of $\mathcal D(\tau,\ul P)$, which is smaller than 
$({\rm const.})^{|V(\tau)|}$ for an appropriate constant, cf. \cite[Remark 3.14]{AGG_CMP}. Next, proceeding as described below \cite[Eq.(4.2.4)]{AGG_CMP}, we let $\tau^*:= \tau^*(\tau)$ be the `fully pruned' minimal subtree\footnote{Compared to the notation used in \cite{GGM}, here we are directly considering fully pruned minimal subtrees, which have only `non-trivial' vertices (i.e., non-endpoint vertices $v$ distinct from the root with $|S^*_v|\ge 2$) in addition to the root and endpoints.} of $\tau$ as shown in \cref{fig:pruning}: this is obtained by removing from $\tau$ all vertices having $m_v = 0$ and all the vertices $v\neq v_0$ with $m_v = m_{w}$ for some $w \in S_w$ (after having removed such a vertex $v$ we connect $w$ and $v'$ in $\tau^*$, and then iterate the procedure if necessary), so that the vertex set of $\tau^*$ is 
$$V(\tau^*):= \set{v_0} \cup \set{ v \in V_e(\tau): m_v=1} \cup \set{ v \in V(\tau):|S^*_v| \ge 2} \subset V(\tau),$$
and $\mathcal{T}^*_{h,h_0^*,m}$ be the set of these trees that have dotted root $v_0$ on scale $h+1$,  the leftmost branching point $v_0^*\in V(\tau^*)$ on scale $h_0^* \ge h+1$ (if $h_0^* = h+1$, $v_0^*=v_0$) and $m$ boundary spin end-points. 
\begin{figure}[h]
\centering
\begin{tikzpicture}[baseline=-0.4em]
\draw[dashed, gray] (0,0)node[below]{$h+1$} -- (0,8)  ;
\draw[dashed, gray] (1,0) -- (1,8) ;
\draw[dashed, gray] (2,0) -- (2,8) ;
\draw[dashed, gray] (3,0)  -- (3,8) ;
\draw[dashed, gray] (4,0) node[below]{$0$} -- (4,8) ;
\draw[dashed, gray] (5,0) node[below]{$1$}-- (5,8) ;
\draw[dashed, gray] (6,0) node[below]{$2$} -- (6,8) ;
\draw[red] (0,4) -- (1,4.5);
\draw[red] (1,4.5)--(2,5)-- (3,5.5);
\draw[red] (2,5) -- (3,4.5) --(4,4)--(5,3.5);
\node[left,red] at (0,4) {$v_0$};
\node[red] at (0,4) {\pgfuseplotmark{o}};
\node[black] at (1,4.5) {\pgfuseplotmark{o}};
\node[red] at  (2,5) {\pgfuseplotmark{o}};
\node[above,red] at (2,5) {$v_1$};
\node[black]  at  (3,4.5) {\pgfuseplotmark{o}};
\draw[red] (0,4)--(1,3.5)--(2,3);
\draw[red](2,3)--(3,2.5)--(4,2);
\draw[black] (4,2)--(5,1.5)--(6,1);
\draw[red] (4,2)--(5,2.5);
\node[black] at  (5,1.5) {\pgfuseplotmark{*}};
\node[black] at  (6,1) {\pgfuseplotmark{*}};
\draw[color=red]   plot[mark=triangle,mark size=3.4pt] (5,2.5);
\node[above,red] at (5,2.5) {$v_6$};
\node[red] at  (2,3) {\pgfuseplotmark{o}};
\node[above,red] at (2,3) {$v_2$};
\node[black] at  (4,2) {\pgfuseplotmark{o}};
\node[above,red] at (3,5.5) {$v_3$};
\draw[color=red]   plot[mark=triangle,mark size=3pt] (3,5.5);
\draw[red] (2,3)--(3,3.5);
 
\draw[black](3,4.5)--(4,5)--(5,5.5)--(6,6);
\node[black] at  (4,5) {\pgfuseplotmark{o}};
\node[black] at  (4,4) {\pgfuseplotmark{o}};
\node[black] at  (5,5.5) {\pgfuseplotmark{o}};
\node[black] at  (6,6) {\pgfuseplotmark{o}};

\draw[color=black]   plot[mark=diamond,mark size=3pt] (5,4.5);
\draw[black] (4,4) -- (5,4.5);

\node[above,red] at (5,3.5) {$v_5$};
\draw[color=red]  plot[mark=triangle,mark size=3pt] (5,3.5) ;
\node[above,red] at (3,3.5) {$v_4$};
\draw[color=red]  plot[mark=triangle,mark size=3pt] (3,3.5) ;
\end{tikzpicture}
\caption{The red highlighted subtree is $\tau^{*}(\tau)$ and $V ( \tau^*) = \set{v_0,v_1,\ldots, v_6}$. Note in particular that $v_{5}$ and $v_{6}$ as endpoints of $\tau$ have $h_v = h_{v'}+1$, while as endpoints of $\tau^*$ they have  $h_v > h_{v'}+1$.} \label{taustar}
\label{fig:pruning}
\end{figure}
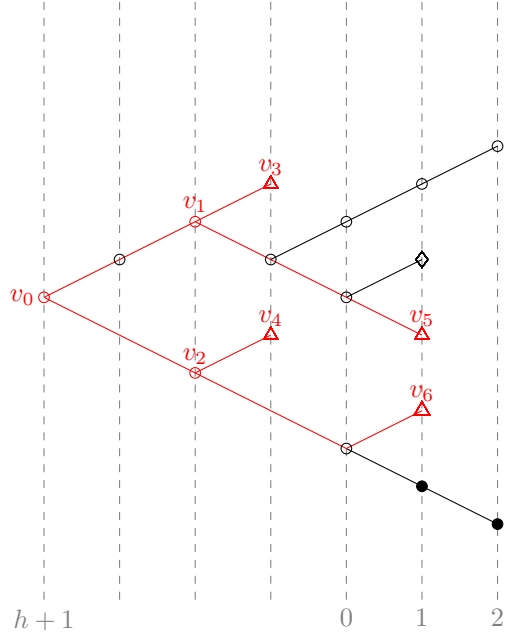

The idea is to rewrite the sum over $\tau$ in \eqref{R22} as a sum over the pruned trees times a sum over $\tau$ compatible with pruning, to first sum over compatible $\tau$ and $\ul P$ at fixed pruning and then over $\tau^*$. Exploiting the fact that some vertices can be considered both vertices of $\tau$ and $\tau^*$, we replace $d'_v$ with 
$d_v' -\tilde{d}_v +\tilde{d}_v$, where  $\tilde{d}_v$ is chosen so as to extract from $d'_v$ a minimal negative part to be associated as $2^{(h_v-h_{v'})\tilde{d}_v}$ with the vertices of $\tau^*$ and, at the same time, make $2^{(h_v-h_{v'})(d'_v-\tilde{d}_v)}$ exponentially small, both in $(h_v-h_{v'})$ and in $|P_v|$, to be associated with the vertices of $\tau$
that are not in $\tau^*$. Concretely, we let $\varepsilon = \frac{(1-\theta)}{3}$, with $\theta$ the same as in \eqref{R22}, and
\begin{equation}\label{dtilde}
	\tilde{d}_v 
	:= \begin{cases}
	\frac12 + \varepsilon - \frac{m_v}{2} & \mbox{ if } m_v > 1\\
	-\theta & \mbox{ if } m_v = 1\\
	 0 & \mbox{ if } m_v =0\,
	\end{cases} 
	 \end{equation}
	 so that $d_v' -\tilde{d}_v \le -\frac{\varepsilon}{2}|P_v|- \theta \delta_{m_v,0}$, which is easily verified from the definitions (cf. \cite[Eq.(4.2.6)]{AGG_CMP}). 
	 Thanks to this bound, the considerations spelled out 
	 after \eqref{R22}, the remark that $|V(\tau)|\le ({\rm const.})\sum_{v\in V_e(\tau)}|P_v|$, and recalling that $ |Z_{\Bs,h_v-1}| \le 1+({\rm const.})|\lambda|$ for an appropriate constant independent of $h$, we find that \eqref{R22} can be further bounded as
\begin{equation}\begin{aligned}\label{R222}
| R_2( \bs{y})| & \le  \sum_{h \le-1}2^{h(1-\frac{m}2)}  \sum_{h < h_0^* \le 0} \sum_{\tau^*\in\mathcal{T}^*_{h,h_0^*,m}}
 \Big(\prod_{v \in V(\tau^*)\setminus \set{v_0}}2^{(h_v-h_{v'})\tilde d_v}  \Big)\\
 &\times \Big(\prod_{v\in V(\tau^*)} 2^{[|S_v|-1]_+h_v} e^{-{\frac{c_0}{12}}2^{h_v}\delta(\bs y_v)}\Big)\sum^*_{\tau} 
 \sum_{\substack{\ul{P}\in \mathcal{P}(\tau):\\ P_{v_0} = \emptyset}} {(C')^{\sum_{v\in V_e(\tau)}|P_v|}}  \\
&\times \Big(\prod_{v \in V(\tau)\setminus \set{v_0}}2^{-\frac{\varepsilon}2|P_v|(h_v-h_{v'})} \Big) \Big(\prod_{\substack{v \in V(\tau):\\ m_v=0}}2^{-\theta(h_v-h_{v'})} \Big)\Big(\prod_{\substack{v \in V_e(\tau):\\ m_v=0}} |\lambda|^{\max\{1,\kappa |P_v|\}}2^{\theta h_v} \Big)\,,
      \end{aligned}
\end{equation}
for some $C,C'>0$, where $\sum^*_\tau$ indicates the sum over the trees in $\mathcal T^{(h)}_\Bs$ with $m_{v_0}=m$, $|V_e(\tau)|>m$ and $\tau^*(\tau)=\tau^*$. Note that, in 
order to obtain this estimate from  \eqref{R22}, we also dropped the 
factor $e^{-c_0 2^{h} \delta (\bs{y})}$ in the first line of  \eqref{R22}, since it is comparable with the factor $e^{-{\frac{c_0}{12}}2^{h_{v_0}}\delta(\bs y_{v_0})}$ (here $\bs y_{v_0}=\bs y$) appearing in the first product
in the second line and, therefore, un-necessary for the subsequent estimates, and for any $\tau\in \mathcal T^*_{h,h_0^*,m}$ we bounded $\prod_{\substack{v\in V(\tau):\\
m_v\ge 1}} 2^{[|S^*_v|-1]_+h_v} e^{-{\frac{c_0}{12}}2^{h_v}\delta(\bs y_v)}$ from above by 
$\prod_{v\in V(\tau^*)} 2^{[|S_v|-1]_+h_v} e^{-{\frac{c_0}{12}}2^{h_v}\delta(\bs y_v)}$. 
Next, letting $h_b^*:=\max\{h_{v'} : v\in V_e(\tau^*)\}$, we use that, for some $C''>0$, 
\begin{equation}\begin{split} & \sum^*_{\tau} 
 \sum_{\substack{\ul{P}\in \mathcal{P}(\tau):\\ P_{v_0} = \emptyset}} {(C')^{\sum_{v\in V_e(\tau)}|P_v|}}  
 \Big(\prod_{v \in V(\tau)\setminus \set{v_0}}2^{-\frac{\varepsilon}2|P_v|(h_v-h_{v'})} \Big) \Big(\prod_{\substack{v \in V(\tau):\\ m_v=0}}2^{-\theta(h_v-h_{v'})} \Big) \\
 &\times \Big(\prod_{\substack{v \in V_e(\tau):\\ m_v=0}} |\lambda|^{\max\{1,\kappa |P_v|\}} 2^{\theta h_v}\Big)\le (C'')^m|\lambda| 2^{\theta h^*_b} \,,
      \end{split}
\end{equation} 
which is the analogue of \cite[Eq.(4.23)]{GGM}. Therefore, up to a redefinition of the constant $C$, 
\begin{equation}\begin{aligned}\label{R2222}
| R_2( \bs{y})| & \le C^m |\lambda|  \sum_{h \le-1}2^{h(1-\frac{m}2)}  \sum_{h < h_0^* \le 0} \sum_{\tau^*\in\mathcal{T}^*_{h,h_0^*,m}} 2^{\theta h^*_b}
 \Big(\prod_{v \in V(\tau^*)\setminus \set{v_0}}2^{(h_v-h_{v'})\tilde d_v}  \Big)\\
 &\times \Big(\prod_{v\in V(\tau^*)} 2^{[|S^*_v|-1]_+h_v} e^{-{\frac{c_0}{12}}2^{h_v}\delta(\bs y_v)}\Big).
      \end{aligned}
\end{equation}
If we now isolate from the two last products the contributions from the (possibly coinciding) vertices $v_0$ and $v_0^*$, use the fact that $\tilde d_{v_0^*}=\frac12+\varepsilon-
\frac{m}2$, that $\delta(\bs y_{v_0^*})=\delta(\bs y_{v_0})=\delta(\bs y)$, and recall that $h_{v_0}=h+1$, we obtain, up to a further redefinition of $C$, 
 \begin{equation}\begin{aligned}\label{R23}
|R_2(\bs{y})|  \le & C^m|\lambda| \sum_{h \le -1}  \sum_{h < h_0^* \le -1} 2^{(h-h_0^*)(\frac12-\varepsilon)} \sum_{\tau^*\in\mathcal{T}^*_{h,h_0^*,m}} 2^{\theta h^*_b}  2^{(|S_{v_0^*}|- \frac{m}{2})h_0^*} e^{-{\frac{c_0}{12}}2^{h_0^*}\delta(\bs y)}\\ 
	& \times \Big( \prod_{\substack{v \in V(\tau^*) :\\ v>v_0^*}}   2^{(h_v-h_{v'}) \tilde{d}_v}  2^{[|S^*_v|-1]_+h_v} e^{-{\frac{c_0}{12}}2^{h_v}\delta(\bs y_v)}\Big)
\end{aligned}\end{equation}
We can easily sum over the classes of trees which differ only in the scale of the root and of the endpoints (given all the scales of other vertices fixed), obtaining, 
up to a further redefinition of $C$, 
\begin{equation}\begin{aligned}\label{R24}
	|R_2(\bs{y})|  \le & C^m |\lambda|  \sum_{h_0^* \le -1}  
		     \sum_{\tau^*\in\mathcal{T}^{**}_{h_0^*,m}} 2^{\theta h^*_b}  2^{(|S_{v_0^*}|- \frac{m}{2})h_0^*} e^{-{\frac{c_0}{12}}2^{h_0^*}\delta(\bs y)}\\ 
	& \times \Big( \prod_{\substack{v \in V(\tau^*)\setminus V_e(\tau^*) : \\ v>v_0^*}}   2^{(h_v-h_{v'}) \tilde{d}_v}  2^{(|S_v|-1)h_v} e^{-{\frac{c_0}{12}}2^{h_v}\delta(\bs y_v)}\Big)
\end{aligned}\end{equation}
(cf.\ \cite[Equation~(4.27)]{GGM})
where $ \mathcal{T}^{**}_{h_0^*,m}$ is the subset of $\mathcal{T}^{*}_{h_0^*,h_0^*,m}$ with $h_w = h_w'+1$ for all endpoints, and implicitly $v_0 = v_0^*$. 
Regrouping the factors depending on the different $h_v$, we can recast \eqref{R24} in the following form: 
 \begin{equation}\begin{aligned}\label{R34}
		 |R_2(\bs{y})|\le C^m  |\lambda|   \sum_{ h_0^* \le -1}  \sum_{\tau^*\in\mathcal{T}^{**}_{h_0^*,m}} \prod_{\substack{v \in V(\tau^*)\setminus V_e(\tau^*)}}   2^{\alpha_v h_v} e^{-{\frac{c_0}{12}}2^{h_v}\delta(\bs y_v)} 
	  \end{aligned}\end{equation}
(cf.\ \cite[Equation~(4.31)]{GGM}),
where, denoting, for any $v\in V(\tau^*)\setminus V_e(\tau^*)$, $S_v^1 := S_v \cap V_e(\tau^{*})$ and $S_v^2 := S_v\setminus S_v^1$,
\begin{align}
	\alpha_{v_0^*}
	&
	\begin{aligned}[t]
		& =
		|S_{v_0^*}| - \frac{m}{2}
		-
		\sum_{w \in S_{v_0^*}^2} \tilde d_w
		=
		|S_{v_0^*}| - \frac{m}{2}
		-
		(\tfrac12 + \varepsilon)|S_{v_0^*}^2|
		+
		\sum_{w \in S_{v_0^*}^2} \tfrac{m_w}{2}
		\\ &
		=
		(\tfrac12 - \varepsilon)|S_{v_0^*}|
		+
		(\tfrac12 + \varepsilon)|S_{v_0^*}^1|
		-
		\tfrac12 |S_{v_0^*}^1|
		=
		(\tfrac12 - \varepsilon)|S_{v_0^*}|
		+
		\varepsilon |S_{v_0^*}^1|
		,
	\end{aligned}
	\label{eq:alpha_v0}
	\\
	\alpha_{v}
	&
	\begin{aligned}[t]
		& =
		\tilde d_v
		+|S_{v}| - 1
		-
		\sum_{w \in S_{v}^2} \tilde d_w
		=
		(\tfrac12 + \varepsilon)(1-|S_{v}^2|)
		+
		|S_{v}| - 1
		-
		\frac12 \Big( m_v - \sum_{w \in S_{v}^2} m_w \Big)
		\\ &
		=
		(\tfrac12 - \varepsilon)(|S_{v}|-1)
		+
		(\tfrac12 + \varepsilon)|S_{v}^1|
		-
		\tfrac12 |S_{v}^1|=
		(\tfrac12 - \varepsilon)(|S_{v}|-1)
		+
		\varepsilon |S_{v}^1|
		,
		\vphantom{\sum_v}
		\quad\quad v \neq v_0^*,v_b
	\end{aligned}
	\label{eq:alpha_v}
	\\
	\alpha_{v_b}
	& =
	(\tfrac12 - \varepsilon)(|S_{v_b}|-1)
	+
	\varepsilon |S_{v_b}^1|
	+ \theta 
	\label{eq:alpha_vb}
\end{align}
where $v_b$ is an arbitrary vertex in $V(\tau)\setminus V_e(\tau^*)$ with $h_{v_b}= h^*_b$.
In order to perform the sums in the right hand side of \eqref{R34}, we proceed similarly to what is described below \cite[Eq.(4.12)]{GGM}: we bound $\delta(\bs y_v)\ge (|S_v|-1)d$, with 
$d$  the minimum pairwise distance among $\bs{y} = (y_1,\dots,y_m)$, and write the sum over $\tau^*\in \mathcal T^{**}_{h_0^*,m}$ as the sum over the ``shape'' (or ``skeleton'') $t$ of $\tau^*$ times the sum over the scale assignments of the non-endpoint vertices of $t$, thus finding 
 \begin{equation}\label{R345}
		 |R_2(\bs{y})|\le C^m  |\lambda|  \sum_{t\in \mathfrak{T}_{m}}\Big[ \prod_{v \in V(t)\setminus V_e(t)} \Big(\sum_{h_v\in\mathbb Z}2^{\alpha_v h_v} e^{-{\frac{c_0}{12}}2^{h_v}(|S_v|-1)d}\Big)\Big],\end{equation}
where $ \mathfrak{T}_{m}$ is the set of possible shapes of the tree (cf. with the definition of $\mathfrak{T}_{0,m}$ given below \cite[Eq.(4.13)]{GGM}). Now,  noting that all the $\alpha_v$ are positive and
\begin{equation*}
	\sum_{v \in V(\tau^*) \setminus V_e(\tau^*)} 
	\left( \left|S_v\right|-1 \right)
	= m-1
	, \quad
	\sum_{v \in V(\tau^*) \setminus V_e(\tau^*)} 
	|S_v^1|
	=m,
\end{equation*}
so that 
\begin{equation}
	\sum_{v \in V(\tau^*) \setminus V_e(\tau^*)} 
	\alpha_v
	=
	\tfrac{m}2 +\theta,
	\label{eq:alpha_sum}
\end{equation}
we perform the sums over the scale labels via  \cite[Eq.(4.13)]{GGM}, with $c_\alpha$ as in \cite[Eq.(4.14)]{GGM}, which is such that $c_\alpha\le ({\rm const.})^\alpha\alpha^\alpha$ for an appropriate constant (see line after 
\cite[Eq.(4.14)]{GGM}), so that, up to a redefinition of $C$, 
 \begin{equation}\label{R3456}
		 |R_2(\bs{y})|\le C^m  |\lambda|  \sum_{t\in \mathfrak{T}_{m}}\Big[ \prod_{v \in V(t)\setminus V_e(t)} \Big(\frac{\alpha_v}{(|S_v|-1)d}\Big)^{\alpha_v}
\Big].\end{equation}
Using the fact that $\alpha_v\le |S_v|-1$ for any non-endpoint vertex of $t$, and that the number of elements of $\mathfrak{T}_m$ is smaller  than $16^m$ (see \cite[Lemma A.1]{GM01}), we find, up to another redefinition of $C$,  
 \begin{equation}\label{R34567}
		 |R_2(\bs{y})|\le C^m  |\lambda|  d^{-m/2 -\theta}.
\end{equation}
By plugging this and \eqref{R2secondo} in \eqref{R2primo} we finally obtain
\begin{equation}\label{R2finalissimo}\Big| \langle \sigma_{y_1}; \cdots ;\sigma_{y_{m}} \rangle_{\lambda,t;\bH}-  (Z_{\Bs})^m
	\langle \sigma_{y_1};\cdots; \sigma_{y_{m}} \rangle_{0,t^*;\bH}\Big|\le C^m |\lambda|m! d^{-m/2-\theta},\end{equation}
for any even, non zero, $m$, and it should be recalled that $\langle \sigma_{y_1};\cdots; \sigma_{y_{m}} \rangle_{0,t^*;\bH}$ is non zero only for $m=2$, in which case 
\begin{equation}\label{eq:edaje1}(Z_\Bs)^2\big|\langle \sigma_{y_1}; \sigma_{y_{2}} \rangle_{0,t^*;\bH}\big|\le C'|y_1-y_2|^{-1}\end{equation}
 and, therefore, in view of \eqref{R2finalissimo}, 
\begin{equation}\label{eq:edaje2}\big|\langle \sigma_{y_1}; \sigma_{y_{2}} \rangle_{\lambda,t;\bH}\big|\le C''|y_1-y_2|^{-1},\end{equation} 
for some $C',C''>0$. Without loss of generality, we can choose these constants $C',C''$ 
and the constant $C$ in \eqref{R2finalissimo} (possibly by increasing some of them) in such a way that $C'=C''=C^2$, and we shall do so. 

\medskip

In order to prove Theorem \ref{prop:main}, we are left with translating the bound \eqref{R2finalissimo} into an analogous bound for the standard, non-truncated, 
correlations. We use the fact that 
\begin{equation} \langle \sigma_{y_1} \cdots \sigma_{y_{m}} \rangle_{\lambda,t;\bH}=\sum_{P\in\mathcal P_m^{\text{even}}}s(P)\prod_{Y\in P}\langle\sigma(\bs y(Y))\rangle^T_{\lambda,t;\bH}\label{eq:daje3}
\end{equation}
where $\mathcal P_m^{\text{even}}$ is the set of partitions of $(1,\ldots,n)$ whose elements all have even cardinality (we will think any $P\in\mathcal P_m^{\text{even}}$ 
as a collection $P=\{Y_1,\ldots, Y_k\}$, with $1\le k\le m/2$, whose elements $Y_i$ are tuples of the form $Y_i=(j_1^i,\ldots,j_n^i)$ with $j_1^i<\cdots<j_n^i$). Moreover, 
if $Y=(j_1,\ldots,j_n)\subseteq(1,\ldots,n)$ with $j_1<\cdots<j_n$, we let $\langle\sigma(\bs y(Y))\rangle^T_{\lambda,t;\bH}:=\langle\sigma_{y_{j_1}};\cdots;
\sigma_{y_{j_n}}\rangle_{\lambda,t;\bH}$, and $s(P)$ with $P=\{Y_1,\ldots, Y_k\}$ is the sign of the permutation from $(1,\ldots,n)$ to $Y_1\oplus\cdots\oplus Y_k$. Thanks 
to \eqref{eq:daje3} we can write: 
\begin{equation} \begin{split} & \langle \sigma_{y_1} \cdots \sigma_{y_{m}} \rangle_{\lambda,t;\bH}-(Z_\Bs)^{m}\langle \sigma_{y_1} \cdots \sigma_{y_{m}} \rangle_{0,t^*;\bH}
=\sum_{P\in\mathcal P_m^{\text{even}}}^*s(P)\sum_{\substack{P_1,P_2\neq\emptyset : \\ P_1\cup P_2=P}}Z_{\Bs}^{2|P_1|}
\Big(\prod_{Y\in P_1}\langle\sigma(\bs y(Y))\rangle^T_{0,t^*;\bH}\Big)\\
&\qquad \times \Big(\prod_{Y'\in P_2}
\big(\langle\sigma(\bs y(Y'))\rangle^T_{\lambda,t;\bH}-Z_{\Bs}^2\langle\sigma(\bs y(Y'))\rangle^T_{0,t^*;\bH}\big)\Big)+
\sum_{P\in\mathcal P_m^{\text{even}}}^{**}s(P)\prod_{Y\in P}\langle\sigma(\bs y(Y))\rangle^T_{\lambda,t;\bH},\end{split}\label{eq:daje4}
\end{equation}
where the $*$ on the first sum in the right hand side indicates the constraint that $P$ has $m/2$ elements (i.e., all the elements of $P$ are pairs), while 
the $**$ on the second sum indicates that the sum is over the remaining elements of $\mathcal P_m^{\text{even}}$(i.e.\ those with at least one elements of $P$ whose cardinality is larger than 2). 
Now, using \eqref{R2finalissimo}, \eqref{eq:edaje1} and \eqref{eq:edaje2} with $C'=C''=C^2$ in \eqref{eq:daje4}, we find: 
\begin{equation}\label{R2d}
	\begin{split} \Big| \langle \sigma_{y_1} \cdots\sigma_{y_{m}} \rangle_{\lambda,t;\bH}-  (Z_{\Bs})^m
\langle \sigma_{y_1}\cdots \sigma_{y_{m}} \rangle_{0,t^*;\bH}\Big|&\le
\frac{C^m}{d^{m/2}}\sum_{P\in\mathcal P_m^{\text{even}}}^* \sum_{k=1}^{m/2}{m/2\choose k}(|\lambda|d^{-\theta})^k\\
&\quad +\frac{C^m |\lambda|}{d^{m/2+\theta}}\sum_{P\in\mathcal P_m^{\text{even}}}^{**}\prod_{Y\in P}|Y|!\, ,\end{split}
\end{equation}
in particular the factor of $|\lambda|$ in the second line is obtained using \cref{R2finalissimo}, noting that $\langle \sigma_{y_1};\cdots; \sigma_{y_{m}} \rangle_{0,t^*;\bH} =0$ when $m > 2$.

Noting that $\sum_{k=1}^{m/2}{m/2\choose k}(|\lambda|d^{-\theta})^k=(1+|\lambda|d^{-\theta})^{m/2}-1$, which is smaller than $2^{m/2}|\lambda|d^{-\theta}$ for $\lambda$ small enough, \eqref{R2d} implies
\begin{equation}\label{R3d} \Big| \langle \sigma_{y_1} \cdots\sigma_{y_{m}} \rangle_{\lambda,t;\bH}-  (Z_{\Bs})^m
\langle \sigma_{y_1}\cdots \sigma_{y_{m}} \rangle_{0,t^*;\bH}\Big|\le
\frac{C^m|\lambda|}{d^{m/2+\theta}}\sum_{P\in\mathcal P_m^{\text{even}}}\prod_{Y\in P}|Y|!\end{equation}
For any positive even $m$, the right hand side can be further bounded from above by $A_m:=\sum_{P\in\mathcal P_m}\prod_{Y\in P}|Y|!$, 
where $\mathcal P_m$ is the set of all partitions of $(1,\ldots,n)$. Note that $A_m$ is well defined for any positive integer, and its first few values are $A_1=1$, $A_2=3$, $A_3=13$, 
$A_4=73$, $\ldots$ In general, for any $m\ge 1$, $A_m$ satisfies the recursion 
\begin{equation}A_{m+1}=\sum_{k=0}^m{m\choose k}(k+1)! A_{m-k},\label{induct}\end{equation}
with the understanding that $A_0:=1$. From this, it is easy to prove inductively that $A_m\le 2^m m!$: first of all, the inequality is valid for $m=0,1$; second, 
assuming inductively its validity for all 
integers $m\le n$, \eqref{induct} implies
\begin{equation} A_{n+1}\le 2^n n! \sum_{k=0}^n (k+1) 2^{-k}, \end{equation}
from which the desired bound follows, simply because $\sum_{k=0}^n(k+1)2^{-k}\le 2(n+1)$ for all non-negative integers $n$, as the reader can easily check. In conclusion, 
$A_m\le 2^m m!$ and, therefore, from \eqref{R3d}, 
\begin{equation}\label{R4d} \Big| \langle \sigma_{y_1} \cdots\sigma_{y_{m}} \rangle_{\lambda,t;\bH}-  (Z_{\Bs})^m
\langle \sigma_{y_1}\cdots \sigma_{y_{m}} \rangle_{0,t^*;\bH}\Big|\le \frac{(2C)^m|\lambda|m!}{d^{m/2+\theta}},\end{equation}
which concludes the proof of \eqref{10b} and, therefore, of \cref{prop:main}.

\appendix
\section{Grassmann representation: spins on both boundaries of the cylinder.} 
\label{ARC}

In \cref{lambda0} we derived the representation of correlations of spins placed on a single boundary of the cylinder for the Ising model with periodic horizontal boundary conditions. In 
this section we illustrate how to adapt the procedure to derive the representation of correlations of spins placed on both boundaries of the cylinder. This may be used to derive the 
scaling limit of multipoint spin correlations on a cylinder with spin observables located on both boundaries for the same class of non-integrable Ising models considered n this paper. 
The construction would require further generalization of the method of this and preceding papers. We expect this to be straightforward, but we prefer not to belabor the details here, in 
order not to overwhelm an already lengthy paper. 

\medskip

We start by treating the case $\lambda=0$, i.e., the integrable, nearest-neighbor, case.
Let $\partial \L =\partial^l\L \cup \partial^u\L$  be the boundary of the cylinder, where $\partial^l\L=\{z=((z)_1,(z)_2) \in\L: (z)_2=1\}$ is the lower boundary and $\partial^u\L:=\{z=((z)_1,(z)_2) \in\L : (z)_2=M\}$ is the upper boundary. Let $n_l$ (resp.\ $n_u$) be the number of spins on the $\partial^{l} \L$ (resp.\ on the $\partial^{u} \L$), $n_l,n_u \in  \mathbb{N}_0$ such that $n_l+n_u = m\in 2 \mathbb{N}$; let $\bs{y} = (y_1,\ldots,y_m) \in (\partial \L)^m$ be a specific `cyclically ordered'  sequence of distinct boundary vertices, that is, first on the lower boundary, ordering the $n_l$ vertices from right to left, then on the upper boundary, ordering the $n_u$ vertices from left to right.

For reasons that will become clear later, we consider the case of an odd number of spins on each boundary for the Ising model  with both periodic and anti-periodic horizontal boundary 
conditions. In this case, in fact, we will obtain the Grassmann representation of  the most general generating function, from which also  the correlations of an even-even number of spins 
on the upper-lower boundaries may be obtained.

For $n_l,n_u \in 2\mathbb{N}+1$ such that  $n_l + n_u = m \in 2 \mathbb{N}$, we consider the $\bs y=( y_1,\ldots, y_{m})\in (\partial\L)^m$ cyclically ordered sequence of distinct boundary vertices. We let $\tilde{G}'_\L = (\L, \mathfrak{B}_\L \cup \tilde{\mathfrak{B}}')$ be the weighted graph, shown in Fig.~\ref{LambdanOO}, obtained from the graph $G_\L = (\L, \mathfrak{B}_\L)$ by adding the edge set $\tilde{\mathfrak{B}}'$ consisting of $m/2$ auxiliary edges connecting the sites labelled by $\bs y$ according to the cyclic order. Let $\tau \in \{+,-\}$ and let
\begin{equation}\label{H_addbOO}
	\tilde{H}_{\L}^\tau(\s)
	=
	- \sum_{i=1}^2 J_i \sum_{\{z, z+\hat{e}_i\} \in \mathfrak{B}_\L} \sigma_{z}\sigma_{z+\hat{e}_i}  
	- \sum_{k=1}^{m/2} \tilde{J}_{k}  \s_{y_{2k-1}}\s_{y_{2k}} \,,
\end{equation}
be an auxiliary spin Hamiltonian with periodic or anti-periodic horizontal boundary conditions, depending on whether $\tau=+$ or $\tau=-$\footnote{That is, if $z = (L,(z)_2)$, $(z)_2 \in \{1,\ldots,M\}$, we interpret 
$\sigma_{z+\hat{e}_1}$ as being equal to $\tau \sigma_{1,(z)_2}$. As per the vertical boundary conditions, we assume them to be open, as in the case \eqref{eq:HM}, i.e., 
if $z=((z)_1,M)$ for some $(z)_1\in\{1,\ldots,L\}$, then we interpret $\sigma_{z+\hat e_2}$ as being equal to zero.}. Note that $\tilde H_\Lambda^\tau(\sigma)$ depends upon the auxiliary couplings $\tilde{\bs{J}} = (\tilde{J}_1,\ldots, \tilde{J}_{m/2})$ and describes a nearest-neighbor model defined on the modified graph $\tilde{G}'_\L$. Note that,
if $\tau=+$ the first term in the r.h.s.\ is nothing but the model in \eqref{eq:HM} with $\lambda=0$.
\begin{figure}[h]
\centering
\includegraphics[scale=0.5]{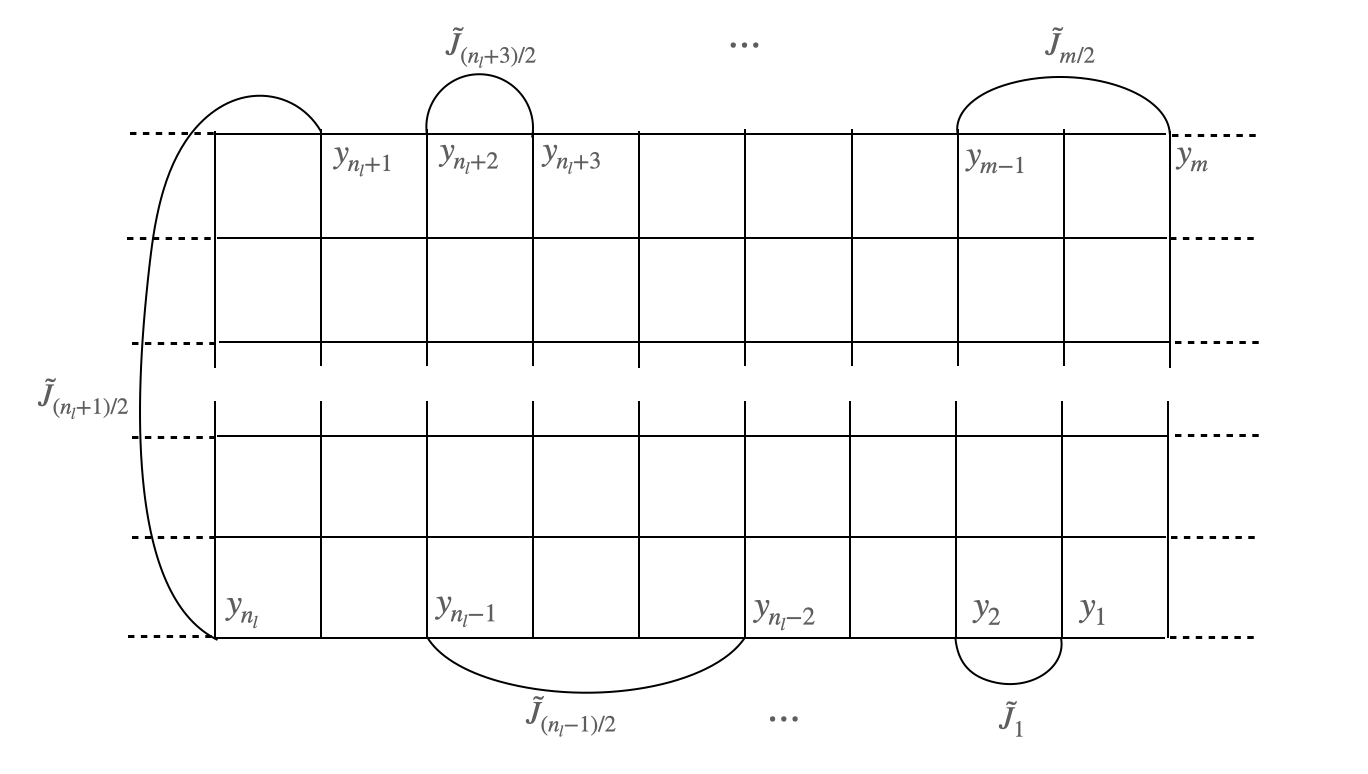}
\caption{The weighted graph  $\tilde{G}_\L'$ associated with the auxiliary Hamiltonian \eqref{H_addbOO}. The $n_l$ vertices on the lower boundary are ordered from right to left while the $n_u$ vertices on the upper boundary are ordered from left to right and $m =n_l+n_u$. Note that the auxiliary edges do not intersect each other; the auxiliary edges connecting the vertices on the same boundary 
 do not intersect any other horizontal
or vertical edge of the graph, including those connecting the first and last columns of the graph; on the other hand, the auxiliary edge connecting the two opposite boundaries intersects
all and only the horizontal dashed edges connecting the first and last columns of the graph. 
}\label{LambdanOO}
\end{figure}
Let $Z^\tau_{0,t;\L} (\tilde{\bs{J}})=\sum_{\sigma \in \Omega_\L}e^{- \beta \tilde{H}^\tau_{\L} (\sigma)}$ be the partition function corresponding to \eqref{H_addbOO}: then
\begin{equation}\label{evspinOO}\begin{aligned}
&\media{\sigma_{y_1}\cdots\sigma_{y_{m}}}^\tau_{0,t;\L}=\left. \frac{\beta^{-m/2}}{Z^\tau_{0,t;\L}( \tilde{\bs{J}})} \frac{\partial^{m/2}}{ \partial \tilde{J}_1 \cdots \partial  \tilde{J}_{m/2}}\, Z^\tau_{0,t;\L} ( \tilde{\bs{J}})  \right|_{\tilde{\bs{J}}=\0}\,.\end{aligned}
\end{equation}
We repeat the procedure described below \eqref{evspin} to obtain a decorated weighted graph  $\tilde{G}'_*$: in addition to the replacements of \cref{replacement123}, those in Fig.~\ref{addVb} must now also be considered. 
\begin{figure}[h!]
\centering
\begin{subfigure}[h]{0.4\textwidth}
      \includegraphics[scale=0.65]{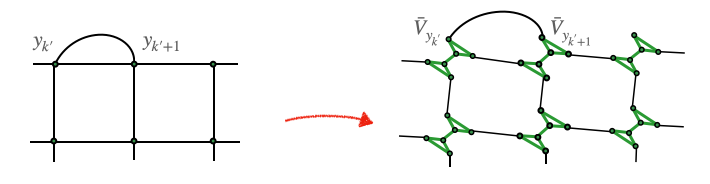}    
\caption{} \label{fig:subfigAA}
\end{subfigure}\qquad \qquad\qquad
\begin{subfigure}[h]{0.4\textwidth}
      \includegraphics[scale=0.5]{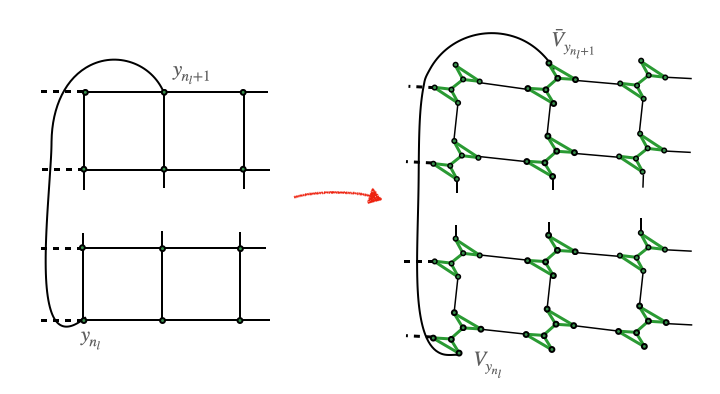}    
\caption{} \label{fig:subfigBA}
\end{subfigure}
\caption{Fig.~\ref{fig:subfigAA} shows an additional edge of $\tilde G'_\Lambda$ connecting two sites on the upper boundary (left), which in the decorated graph $\tilde G'_*$ 
corresponds to an edge connecting two sites of type $\bar{V}$ (right); 
Fig.~\ref{fig:subfigBA} shows the additional edge of $\tilde G'_\Lambda$ connecting the two sites $y_{n_l}$ and $y_{n_l+1}$ on the opposite boundaries (left), which 
in the decorated graph $\tilde G'_*$ corresponds to an edge connecting a site of type $V$ in the lower boundary with a site of type $\bar{V}$ in the upper boundary (right).
\label{addVb}}\end{figure} 
The main difference compared to \cref{lambda0} is that now the additional edge connecting the two opposite boundaries of the cylinder intersects the dashed horizontal edges, so that $\tilde{G}'_*$ is non planar (it cannot be embedded on a cylinder without crossings); rather, it can be embedded on a two-dimensional torus. 
In this case, Kasteleyn's solution tells us that the partition function of the dimer model on $\tilde{G}'_*$ is a linear combination of the Pfaffians of four different antisymmetric matrices, denoted by $A^{(\alpha \tau,\alpha')}_{K}$, $\alpha,\alpha' \in \{+,-\}$, and defined as follows, see e.g.\ \cite{Ga99}.
\begin{figure}\centering
      \includegraphics[scale=0.55]{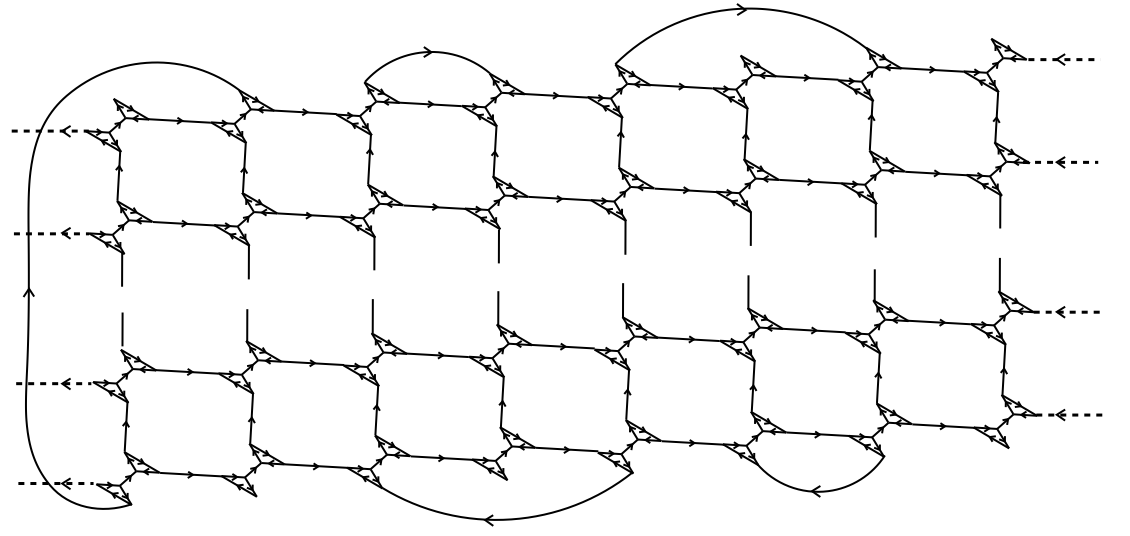}    
\caption{The decorated graph $\tilde{G}'_*$ in which the arrows indicate the principal orientation of the graph. The graph $\tilde{G}'_*$ is not planar due to the intersections between the edge connecting the two boundaries and the dashed horizontal edges connecting the first and last columns: therefore we define the principal orientation starting from the clockwise odd orientations that would occur respectively in the absence of the edge that connects the two boundaries (orienting the dashed edges from the first to the last column, as in \cref{figOK}) and in the absence of the dashed edges (orienting the edge that connects the two boundaries from bottom to top, regardless of the value of $M$). 
Note that the auxiliary edge connecting the two boundaries is directed from $V$ to $\bar{V}$, contrary to every other vertical edge, which is directed from $\bar{V}$ to $V$ instead.
\label{figB}}
\end{figure} 
Similarly to $A_{K}$, see \eqref{eq:Z'Pf} and following lines, the entries of each $A^{(\alpha \tau,\alpha')}_{K}$, $\alpha,\alpha'\in \{+,-\}$, are labelled by pairs of vertices connected to each other by an edge  in $\tilde{G}'_*$ and correspond to the weight of the edge as introduced below \eqref{eq:Z'Pf} multiplied by a sign which depends on the orientation of the edge as illustrated in \cref{figB}.
Each of the four matrices differs in the orientation of the edges which, by intersecting, break the planarity: $A^{(-\tau,-)}_{K}$ denotes the matrix associated with the {\it principal orientation} shown in Fig.~\ref{figB}, $A_{K}^{(\tau,-)}$ (resp.\ $A_{K}^{(-\tau,+)}$) denotes the matrix associated with the orientation obtained from the principal one by reversing the orientation of the edges between the first and last columns (resp.\  of the edge connecting the two boundaries) and $A_{K}^{(\tau,+)}$ denotes the matrix associated with the orientation obtained from the principal one by reversing both  the orientation of the edges between the first and last columns and that of the edge connecting the two boundaries. 
Kasteleyn's solution is then given explicitly by
\begin{equation}\label{Z4Pf}
Z^\tau_{0,t;\L}( \tilde{\bs{J}}) =  {C}_\beta(J_1,J_2,\tilde{\bs{J}})   \sum_{\alpha,\alpha'\in\set{\pm}} \frac{c_{\alpha,\alpha'}}{2} \Pf{A_{K}^{(\alpha\tau,\alpha')}} \,,
\end{equation}
with ${C}_\beta (J_1,J_2,\tilde{\bs{J}})$ defined after \eqref{eq:Z'Pf} and $c_{+,-}=c_{-,+}=c_{-,-}= -c_{+,+}=+1$. Next, by rewriting each Pfaffian in \eqref{Z4Pf} in the form of a Gaussian Grassmann integral, we obtain
\begin{equation}
Z^\tau_{0,t;\L}(\tilde{\bs{J}}) = (-1)^{LM}{C}_\beta(J_1,J_2,\tilde{\bs{J}}) \sum_{\alpha,\alpha'=\pm} \frac{c_{\alpha,\alpha'}}{2} 
\int  \mathcal D \Phi \, e^{\cS^{\alpha\tau}_{t}(\Phi)  +  \cS^{\alpha'}_{\tilde{t}}(\Phi)},
\label{Z_Grass_OO}
\end{equation} 
where $\mathcal D \Phi$ is defined as in \eqref{diffPhi}, $\cS^{\alpha\tau}_{t}(\Phi)$ as in \eqref{cS_def2} with the only difference that $H_{(L+\hat{e}_1,(z)_2)}$,  $(z)_2 \in \set{1,\ldots,M}$, should be interpreted as representing $\alpha\tau H_{(1,(z)_2)}$ (so \cref{cS_def2} corresponds to the case with $\alpha=-$ and $\tau=+$), and
\begin{equation} 
\cS_{\tilde{t}}^{\alpha'}(\Phi) :=  \alpha' \tilde{t}_{(n_l+1)/2}(\beta) \bar{V}_{y_{n_l+1}}  V_{y_{n_l}} +\sum_{j=1}^{(n_l-1)/2} \tilde{t}_j(\beta) V_{y_{2j} }V_{y_{2j-1} }   + \sum_{j'=(n_l+3)/2}^{(n_l+n_u)/2} \tilde{t}_{j'}(\beta) \bar{V}_{{y}_{2j'-1} }\bar{V}_{{y}_{2j'} } \,,
\label{eq:cS_def_addbOO}
\end{equation}
with $\tilde{t} = \{\tilde{t}_k(\beta)\}_{k=1}^{m/2}= \{\tanh\beta \tilde{J}_k\}_{k=1}^{m/2}$ and $n_l+n_u=m$. 
Note that the Grassmann integral in the r.h.s.\ of \eqref{Z_Grass_OO} depends on $\tilde{\bs{J}}$ via the dependence of $\cS_{\tilde{t}}^{\alpha'}$ on $\tilde{t}$;  so when $\tilde{\bs{J}}=\0$ the terms in \cref{Z_Grass_OO} with $\alpha=+$ cancel, giving
\begin{equation}\label{Z0tau}
Z^\tau_{0,t;\L}\equiv Z^\tau_{0,t;\L}({\0}) = (-1)^{LM}{C}_\beta(J_1,J_2,{\0}) \int  \mathcal D \Phi \, e^{\cS^{- \tau}_{t}(\Phi)  }\,,
\end{equation}
consistently with \cref{eq:Z_Grass_addb2,cS_def2}.
Let us consider for a moment the simplest case where there are only two spins with $y_1 \in \partial^l\L$ and  $y_2 \in \partial^u\L$. By plugging Eq.~\eqref{Z_Grass_OO} into Eq.~\eqref{evspinOO} with $n_l=n_u=1$, $m=2$, we express the $2$-point correlation as
\begin{equation}\label{spin_corr_VOO2}\begin{aligned}
\braket{ \s_{y_1} \s_{y_{2}}}^\tau_{0,t;\L} =& \frac{1}{Z^\tau_{0,t;\L}(\bs 0) }  (-1)^{LM}{C}_\beta(J_1,J_2,\bs 0)   
\sum_{\alpha,\alpha'=\pm} \alpha' \frac{c_{\alpha,\alpha'}}{2} \int \mathcal D \Phi \, e^{\cS^{\alpha \tau}_{t}(\Phi)   }  \bar{V}_{y_{2}} V_{y_1}    \\
 =&   \frac{\int \mathcal D \Phi \, e^{\cS^\tau_{t}(\Phi) } \left(- \bar{V}_{y_{2}}V_{y_1} \right) }{
\int \mathcal D \Phi \, e^{\cS^{-\tau}_{t}(\Phi) }}=  
\frac{Z^{-\tau}_{0,t;\L}}{Z^\tau_{0,t;\L}}  
\frac{\int \mathcal D \Phi \, e^{\cS^\tau_{t}(\Phi) } {V_{y_1}\bar{V}_{y_{2}}   }}{\int \mathcal D \Phi \, e^{\cS^{\tau}_{t}(\Phi) }}\,,
\end{aligned}
\end{equation}
where we used that here the terms with $\alpha=-$ are those who cancel after setting $\tilde{\bs J}=\bs 0$. 

\begin{remark}
The Grassmann representation of the spin-spin correlation with 
spin located on two opposite boundaries in \eqref{spin_corr_VOO2} has different features, as compared to the one of two spins located on the same boundary. In fact, in the case that spins are both located on (say) the bottom boundary, then the spin-spin correlation of $\sigma_y$ with $\sigma_{y'}$ (with $y$ to the right of $y'$) with $\tau$ boundary conditions in the horizontal direction equals the Grassmann expectation of $V_yV_{y'}$ with respect to the Gaussian Grassmann measure $\propto \mathcal D\Phi e^{\mathcal S_t^{-\tau}(\Phi)}$: note the switch in the sign of the boundary conditions from the spin to the Grassmann case (in Section \ref{lambda0} we worked this out only for the case $\tau=+$, but it is straightforward to check that the same discussion applies to $\tau=-$, too). On the contrary, \eqref{spin_corr_VOO2} shows that the spin-spin correlation of a spin located on the bottom boundary 
(at site $y_1$) with one located on the top boundary (at site $y_2$) is proportional to the Grassmann expectation of $V_{y_1}{\bar V}_{y_2}$ with respect to the Gaussian Grassmann 
measure $\propto {\mathcal D}\Phi e^{{\mathcal S}_t^{\tau}(\Phi)}$ (note here that the Grassmann measure has the {\bf same} horizontal boundary conditions  as the original spin
Hamiltonian) times a correcting factor given by the ratio of the partition functions with distinct horizontal boundary conditions. 
\end{remark}
Similarly to \eqref{spin_corr_VOO2}, starting from \eqref{Z_Grass_OO}, 
for general $n_l,n_u \in 2 \mathbb{N}+1$, $n_l+n_u = m$, we obtain the following Grassmann representation for the multipoint boundary spin 
correlations:
\begin{equation}\label{spin_corr_VOO}\begin{aligned}
\braket{ \s_{y_1}\cdots \s_{y_{m}}}^\tau_{0,t;\L} &=  \frac{Z^{-\tau}_{0,t;\L}}{Z^\tau_{0,t;\L}}  \,   \frac{\int \mathcal D \Phi \, e^{\cS^\tau_{t}(\Phi) } {V_{y_1}\cdots V_{y_{n_l}}\bar{V}_{y_{n_l+1}} \cdots \bar{V}_{y_{m}}  }}{\int \mathcal D \Phi \, e^{\cS^\tau_{t}(\Phi) }},
\end{aligned}
\end{equation}
while for $n_l,n_u \in 2 \mathbb{N}+1$, $n_l+n_u = m$, we get: 
\begin{equation}\label{spin_corr_VOOeveneven}\begin{aligned}
\braket{ \s_{y_1}\cdots \s_{y_{m}}}^\tau_{0,t;\L} &=  \frac{\int \mathcal D \Phi \, e^{\cS^{-\tau}_{t}(\Phi) } {V_{y_1}\cdots V_{y_{n_l}}\bar{V}_{y_{n_l+1}} \cdots \bar{V}_{y_{m}}  }}{\int \mathcal D \Phi \, e^{\cS^{-\tau}_{t}(\Phi) }}.
\end{aligned}
\end{equation}
Both \eqref{spin_corr_VOO} and \eqref{spin_corr_VOOeveneven} can be expressed in the form of Pfaffians, as follows: let  $A_{\L,\tau}$ be the $4LM\times 4LM$ antisymmetric matrix defined by $S_t^{\tau}(\Phi) = \frac12 (\Phi, A_{\L,\tau}\Phi)$, in terms of which the expectation of a Grassmann binomial (like the one appearing in the r.h.s.\ of \eqref{spin_corr_VOO2} or in \eqref{spin_corr_VOOeveneven} with $m=2$), can be written as
\begin{equation}
\frac{\int \mathcal D \Phi \, e^{\cS^\tau_{t}(\Phi) } {F_{y_1}F'_{ y_2}}}{\int \mathcal D \Phi \, e^{\cS^{\tau}_{t}(\Phi) }} = -[A_{\L,\tau}^{-1}]_{(F,{y_1} ),(F',{y_2})}
\end{equation}
with $F_{y_i} = V_{y_i}$ (resp.\ $F_{y_i} =\bar{V}_{y_i}$) if $y_i \in \partial^l\L$ (resp.\ $y_i \in \partial^u\L$), $i=1,2$, and similarly for $F'_{y_i}$, so that  $F,F' \in \{V,\bar{V}\}$. 
We thus have 
\begin{equation}
\media{\sigma_{y_1}\sigma_{y_2}}^\tau_{0,t;\L}= \begin{cases} 
  -[A_{\L,- \tau}^{-1}]_{(V,{y_1} ),({V},{y_2})} & \mbox{ if } y_1,y_2\in \partial^l\L\,,\\
- \frac{Z^{-\tau}_{0,t;\L}}{Z^\tau_{0,t;\L}}   [A_{\L,\tau}^{-1}]_{(V,{y_1} ),(\bar{V},{y_2})} & \mbox{ if } y_1\in \partial^l\L, y_2 \in \partial^u\L\,,\\
-  [A_{\L,- \tau}^{-1}]_{(\bar{V},{y_1} ),(\bar{V},{y_2})} & \mbox{ if } y_1, y_2 \in \partial^u\L\,.\end{cases}
\end{equation}
and 
\begin{equation}
\frac{\int \mathcal D \Phi \, e^{\cS^\tau_{t}(\Phi) } F_{y_1}\cdots F_{y_{m}} }{\int \mathcal D \Phi \, e^{\cS^{\tau}_{t}(\Phi) }} = \Pf M_{\L,\tau}(\bs y)
\end{equation}
where $M_{\L,\tau}(\bs y)$  is the $m\times m$ antisymmetric matrix with elements \begin{equation}
[M_{\L,\tau}(\bs y)]_{i,j} = -[A_{\L,\tau}^{-1}]_{(F,{y_i} ),(F',{y_j})} = \begin{cases} \frac{Z^{\tau}_{0,t;\L}}{Z^{-\tau}_{0,t;\L}}  \media{\sigma_{y_i}\sigma_{y_j}}^\tau_{0,t;\L}  & \mbox{ if } y_i, y_j \mbox{ are on opposite boundaries }\\
 \media{\sigma_{y_i}\sigma_{y_j}}^{-\tau}_{0,t;\L}  & \mbox{ if } y_i, y_j \mbox{ are on the same boundary\,.}\end{cases}\end{equation} 
 Putting things together we finally find that, for any pair of non-negative integers $n_l,n_u$ such that 
 $n_l+n_u = m \in 2 \mathbb{N}$, 
\begin{equation}
\media{\sigma_{y_1}\cdots\sigma_{y_{n_l}}\sigma_{y_{n_l+1}}\cdots \sigma_{y_m}}^\tau_{0,t;\L}= \begin{cases}
\Pf  M_{\L,-\tau}(\bs y) & \mbox{ if } n_l,n_u \in 2 \mathbb{N}_0\,,\\
\frac{Z^{-\tau}_{0,t;\L}}{Z^{\tau}_{0,t;\L}} \Pf  M_{\L,\tau}(\bs y) & \mbox{ if } n_l,n_u \in 2 \mathbb{N}_0+1\,.
\end{cases}
\label{eq:spin_Pfaff_final}
\end{equation}
which is the desired generalization of \eqref{PfForm}. Starting from this, one can compute the scaling limit of all the multipoint boundary spin correlations: for this purpose, one needs in particular to compute the scaling limit of the ratios $\frac{Z^{-\tau}_{0,t;\L}}{Z^{\tau}_{0,t;\L}}$, which can be worked out by proceeding as in \cite{Greenblatt24}, see also 
\cite{GrJMP23}, which systematically computes subleading contributions to the partition function, starting from which one can compute the scaling limit of the ratios of interest. Note that these ratios can be understood as pair correlations of disorder insertions.  Similarly, the form of \cref{eq:spin_Pfaff_final} can be understood in terms of Kadanoff-Ceva fermions, and the associated disorder insertions cancel only for correlation functions with an even number of spins on each boundary component.

\medskip

Let us now consider the general, non-integrable, case, with $\lambda\neq0$. In this case we proceed as described after Eq.\eqref{spin_corr_V} and in \cref{lambdane0}, and 
obtain a representation that has similar features to that described in Proposition \ref{evspinQ}. Importantly, as in the case described in Section \ref{lambdane0}, 
the representation of the perturbation does not involve the auxiliary edges with couplings $\tilde{\bs J}$ and this implies that the Grassmann representation of the source 
term, denoted by $\mathcal B^{\alpha'}_\Lambda$ below, is independent of $\lambda$. More precisely, we find that the 
correlation of any even number of boundary spins located at boundary sites $y_1,\ldots,y_m$ ordered in cyclic order (as explained above in the $\lambda=0$ case) can be represented as 
\begin{equation}\label{corrgen}
\braket{ \s_{y_1}\cdots \s_{y_{m}}}^\tau_{\l,t;\L}= \frac1{ \Xi^\tau_{\l,t;\L} (\bs{0}) } \frac{\partial^{m}}{\partial  \varphi_{y_1}\partial \varphi_{y_2} \cdots \partial  \varphi_{y_{m-1}}\partial \varphi_{y_{m}}} \Xi^\tau_{\l,t;\L} (\bs{\varphi}) \Big|_{\bs{\varphi}=\bs 0}\,,
\end{equation}
where, letting $\bs{\varphi} := {\set{\varphi_z}_{z \in \partial \L}}$, $\partial \L = \partial^l \L \cup \partial^u \L$, a collection of Grassmann variables, playing as above the role of 
source fields, 
\begin{equation}
	 \Xi^\tau_{\l,t;\L} (\bs{\varphi}) := \sum_{\alpha,\alpha' \in \set{\pm}} \frac{c_{\alpha,\alpha'}}{2} \int \cD \Phi \; e^{\cS^{\alpha\tau}_{t}(\Phi) 
		+ \cB^{\alpha'}_\L(\Phi,\bs{\varphi})+ \cV_\L^{\rm int}(\Phi)} \,,
	\label{XiLgen}
			\end{equation}
with
\begin{equation}\label{eq:Bs_bis}
 \cB^{\alpha'}_\L(\Phi,\bs{\varphi}) := \sum_{z \in \partial^l\L} V_{z} \varphi_{z} -\alpha'  \sum_{z \in \partial^u\L} \bar{V}_{z} \varphi_{z} \,,\end{equation}
and $\cV_\L^{\rm int}(\Phi)$ as in \eqref{eq:B_expansion_bis}.

Eqs.\eqref{corrgen}-\eqref{XiLgen} generalize \eqref{evspinQ}. In particular, it is easy to check that if all the spins are located on the bottom boundary, i.e., if 
$m = n_l \in 2 \mathbb{N}$ we recover \cref{spin_corr_V}. Based on this representation, one can apply the methods developed in this and preceding works to study the scaling limit of 
the multipoint boundary spin correlations in the general case that $\Lambda$ scales to a finite cylinder and there are spin insertions both on the bottom and upper boundaries. 
As already discussed in the $\lambda=0$ case, the control of the general case requires, in particular, to understand the scaling limit of ratios of interacting partition functions of the form 
$\frac{Z^{-}_{\lambda,t;\L}}{Z^{+}_{\lambda,t;\L}}$ in the critical case. The problem of controlling the scaling limit of such ratios is analogous 
to the study of the subleading corrections to the critical free energy \cite{GM13}, and will be postponed to future work. 
 
 \section{First order computation of $Z_{\Bs}$} 
\label{ordineL}

In this appendix, after making an explicit choice of the interaction $V$ in \eqref{eq:HM}, we compute $Z_\Bs(\lambda)$ at first order in $\lambda$ 
and show that the first order correction to $Z_\Bs(0)=1$ is non zero, thus proving that, in the considered case, $Z_\Bs(\lambda)$ is a non-trivial analytic function of $\lambda$. 
From the computation below, it is clear that the choice of the interaction is not particularly important, if not for simplifying here and there a few calculations: in this sense, the
computation indicates that generically $Z_\Bs$ is a non-trivial, i.e., non-constant, function of $\lambda$. 

\medskip

We let $V(X)$ in \eqref{eq:HM} be equal to 1 if $X=\{z-\hat e_2,z+\hat e_2\}$ for some $z\in\Lambda$, and zero otherwise, 
and $J_1=J_2=1$, in which case, letting $z_2$ be the second component of $z=(z_1,z_2)\in \bH$ (to lighten notation, in this appendix we write $(z_1,z_2)$ instead of $((z)_1,(z)_2)$), 
\begin{equation}\label{app.formVint}
\cV_\L^{\rm int} (\Phi) = \alpha_0 \lambda\sum_{\substack{z\in\Lambda : \\ 1<z_2<M}}
\big(\bar{V}_{z-\hat{e}_2} V_{z} + \bar{V}_{z} V_{z+\hat{e}_2} + 2 \bar{V}_{z-\hat{e}_2} V_{z}\bar{V}_{z} V_{z+\hat{e}_2}\big) + O (\lambda^2),
\end{equation}
where $\alpha_0:=2(\sqrt2-1)^2\tanh^{-1}(\sqrt2-1)$. In order to compute $Z_\Bs$ at first order in $\lambda$, we need to compute at the same precision $\mathcal V^{(1)}_\Lambda(\xi,\phi)$ and, therefore, in view of \eqref{qualequesta}, $Z=Z(\lambda)$, $\beta_c=\beta_c(\lambda)$ and $t^*_1=t_1^*(\lambda)$. Moreover, we recall from \cite[Prop.4.11]{AGG_AHP} that $Z=Z(\lambda)$, $\beta_c=\beta_c(\lambda)$ and $t^*_1=t_1^*(\lambda)$ can be computed from $\nu_1,\zeta_1,\eta_1$ via
\cite[Eq.(4.5.4) or equivalently Eqs.(4.5.13)--(4.5.15)]{AGG_AHP}, where 
\begin{equation}\begin{split} & 2\nu_1=2\alpha_0\lambda\big[-1+\partial_{1,2}g_{\infty,-+}(z,z)-\partial_{1,2}g_{\infty,-+}(z,z-\hat e_2)\big]+O(\lambda^2)\\
&\zeta_1=O(\lambda^2)\\
&\eta_1=\alpha_0\lambda\big[-1+2\partial_{1,2}g_{\infty,-+}(z,z)\big]+O(\lambda^2),\end{split}\label{B.2}\end{equation}
and $\partial_{1,2}g_{\infty,-+}(z,z'):=g_{\infty,-+}(z+\hat e_2,z')-g_{\infty,-+}(z,z')$, with 
$g_{\infty,-+}(z,z')=g_{\infty,-+}(z-z',0)=\tfrac{\sqrt2+1}{2}\int_{[-\pi,\pi]^2} \frac{d^2k}{(2\pi)^2}  e^{-ik\cdot (z-z')}\frac{1-B(k_1)e^{ik_2}}{2-\cos k_1-\cos k_2}$, $k=(k_1,k_2)$ and 
$B(k_1)=(\sqrt2-1)(\sqrt2+\cos k_1)$. The explicit expressions of $\nu_1,\zeta_1,\eta_1$ in \eqref{B.2} follow from the first order computation of the corresponding beta functions: in fact, 
they can be computed via \cite[Eq.(4.5.10)]{AGG_AHP} as follows (for the involved notations and corresponding definitions, we refer the reader to \cite{AGG_AHP})
\begin{equation} 2\nu_1=-\sum_{h\le 1}2^{h-1}B^\nu_h[\underline{v}], \qquad \zeta_1=-\sum_{h\le 1}B^\zeta_h[\underline{v}], \qquad \eta_1=-\sum_{h\le 1}B^\eta_h[\underline{v}].\label{Bapp.extra}\end{equation}
In general, these equations should be read as fixed point equations, whose solution is unique in a sufficiently small neighborhood of the origin. However, if we limit 
ourselves to first order in $\lambda$, it is easy to check that the right hand sides of the three equations in \eqref{Bapp.extra} are independent of $\nu_1,\zeta_1,\eta_1$ and read:
\begin{equation}\label{expl.app} 2\nu_1=-\lambda \hat W_{\infty,+-}^{(1)}(0)+O(\lambda^2),\quad \zeta_1=-i\lambda \partial_{1}\hat W_{\infty,--}^{(1)}(0)+O(\lambda^2),\quad 
\eta_1=-\tfrac{i}2\lambda \partial_{2}\hat W_{\infty,+-}^{(1)}(0)+O(\lambda^2),\end{equation}
where $\hat W_{\infty,\omega\omega'}^{(1)}(k)$, with $\omega,\omega'\in\{+,-\}$, are the coefficients of the first order contributions in $\lambda$ to 
$\hat W_{\infty,\omega\omega'}(k)=\sum_{z\in\mathbb Z^2}e^{-ik\cdot z}$ $W_{\infty,\omega\omega'}(z,0)$ where $W_{\infty,\omega\omega'}(z,z')=W_{\infty,\omega\omega'}(z-z',0)$ 
is the infinite-plane limit of the kernel of the quadratic part of the effective potential 
\begin{equation} V_{\Lambda,\text{eff}}(\psi)=\log\int \int P_c^*(\mathcal D\phi)P_m^*(\mathcal D\xi)e^{\widetilde{\mathcal V}^{{\rm int}}_\Lambda(\xi,\phi+\psi)}.
\label{app.effective}\end{equation}
We recall that $\widetilde {\mathcal V}^{{\rm int}}_\Lambda(\xi,\phi)$ was defined right after \eqref{qualequesta}, and from \eqref{app.formVint} we find that in the considered case
$$\widetilde {\mathcal V}^{{\rm int}}_\Lambda(\xi,\phi)= \alpha_0 \lambda\sum_{\substack{z\in\Lambda : \\ 1<z_2<M}}
\big(\phi_{+,z-\hat{e}_2} \phi_{-,z} + \phi_{+,z} \phi_{-,z+\hat{e}_2} + 2 \phi_{+,z-\hat{e}_2} \phi_{-,z}\phi_{+,z} \phi_{-,z+\hat{e}_2}\big) + O (\lambda^2).$$
The straightforward computation leading from \eqref{expl.app}-\eqref{app.effective} to \eqref{B.2} is left to the reader. Given \eqref{B.2} and the explicit form of the propogator
defined right after \eqref{B.2} we can compute \eqref{B.2} more explicitly. In fact, by definition,
\begin{equation}\partial_{1,2}g_{\infty,-+}(z,z')=\tfrac{\sqrt2+1}{2}\int_{[-\pi,\pi]^2} \frac{d^2k}{(2\pi)^2}  e^{-ik\cdot (z-z')}\frac{1-B(k_1)e^{ik_2}}{2-\cos k_1-\cos k_2}(e^{-ik_2}-1)
\end{equation}
 We now use the fact that 
\begin{equation} \begin{split} & \int_{-\pi}^\pi \frac{dk_1}{2\pi}\frac{1-B(k_1)e^{ik_2}}{2-\cos k_1-\cos k_2} =
e^{ik_2}\Big(\sqrt2-1-\frac{\sqrt2(1-\cos k_2)+i\sin k_2}{\sqrt{(1-\cos k_2)(3-\cos k_2)}}\Big),\end{split}\label{app.vasco3}\end{equation}
to rewrite
\begin{equation}\partial_{1,2}g_{\infty,-+}(z,z')=\tfrac{\sqrt2+1}{2}\int_{-\pi}^\pi\frac{dk_2}{2\pi}  e^{-ik\cdot (z-z')}(1-e^{ik_2})
\Big(\sqrt2-1-\frac{\sqrt2(1-\cos k_2)+i\sin k_2}{\sqrt{(1-\cos k_2)(3-\cos k_2)}}\Big),
\end{equation}
so that, in particular, 
\begin{equation}\begin{split}\partial_{1,2}g_{\infty,-+}(z,z)&=\tfrac12-\tfrac{\sqrt2+1}{2}\int_{-\pi}^\pi\frac{dk_2}{2\pi}  \frac{\sqrt2(1-\cos k_2)^2+\sin^2 k_2}{\sqrt{(1-\cos k_2)(3-\cos k_2)}}\\
&=-\frac{\sqrt2}{2}-\frac1\pi,\\
\partial_{1,2}g_{\infty,-+}(z,z-\hat e_2)&=-\tfrac12+\tfrac{\sqrt2+1}{2}\int_{-\pi}^\pi\frac{dk_2}{2\pi}  \frac{\sqrt2(1-\cos k_2)^2-\sin^2 k_2}{\sqrt{(1-\cos k_2)(3-\cos k_2)}}\\
&=-1-\frac{\sqrt2}{2}+\frac{(\sqrt2+1)^2}\pi
\end{split}
\end{equation}
so that, from \eqref{B.2}, 
\begin{equation}\label{rockyr}\begin{split} & 2\nu_1=-\frac8\pi(2-\sqrt2)\beta_0\lambda+O(\lambda^2)\equiv2\nu_1^{(1)}\lambda+O(\lambda^2)\\
&\zeta_1=O(\lambda^2)\\
&\eta_1=-2(\sqrt2-1+\tfrac{2(3-2\sqrt{2})}\pi)\beta_0\lambda+O(\lambda^2)\equiv \eta_1^{(1)}\lambda+O(\lambda^2)\end{split}\end{equation}
where $\beta_0:=\tanh^{-1}(\sqrt2-1)$. Substituting these expressions into \cite[Eq.(4.5.4) or equivalently Eqs.(4.5.13)--(4.5.15)]{AGG_AHP} we can solve 
for $Z,\beta_c,t_1^*$ thus finding
\begin{equation}\label{firstorders.app}\begin{split} & Z=1+(2-\sqrt{2})(2-\tfrac4\pi)\beta_0\lambda+O(\lambda^2)\equiv 1+Z^{(1)}\lambda+O(\lambda^2)\\
& \beta_c=\beta_0\Big(1-\tfrac{2\sqrt{2}}\pi\lambda\Big)+O(\lambda^2)\equiv \beta_0+\beta_1\lambda+O(\lambda^2)\\
& t^*_1=(\sqrt2-1)\Big(1-2\sqrt{2}\beta_0\lambda\Big)+O(\lambda^2)\equiv \sqrt2-1+\tau_1\lambda+O(\lambda^2).\end{split}\end{equation}
Using \cite[Eqs.(2.1.12)-(2.1.13)]{AGG_AHP}, taking the half-plane limit,  letting $t_1(\lambda)=t_2(\lambda)=\tanh\beta_c(\lambda)\equiv t$ and recalling that 
$t_2^*=\tfrac{1-t_1^*}{1+t_1^*}$, we find
\begin{equation}\label{app.V1bH}\begin{split}
\cV_\bH^{(1)} (\xi,\phi) &=(t-t_1^*)\sum_{z\in\bH}\xi_{+,z}\xi_{-,z+\hat e_1}+
 \sum_{z_2\ge 1}\int_{-\pi}^\pi \frac{dk_1}{2\pi} \Big[\big(-\tfrac{b_t(k_1)}{Z}+b_{t^*_1}(k_1)\big)\phi_{+,z_2}(-k_1)\phi_{-,z_2}(k_1)\\
 &+\big(\tfrac{t}{Z}-t_2^*\big)\phi_{+,z_2}(-k_1)\phi_{-,z_2+1}(k_1)
 -\tfrac{i}{2}\big(\tfrac{\Delta_t(k_1)}{Z}-\Delta_{t_1^*}(k_1)\big)\sum_{\omega=\pm}\omega\phi_{\omega,z_2}(-k_1)\phi_{\omega,z_2}(k_1)\Big]\\
 &+
\alpha_0 \lambda\sum_{\substack{z\in\Lambda : \\ z_2\ge 2}}
\big(\phi_{+,z-\hat{e}_2} \phi_{-,z} + \phi_{+,z} \phi_{-,z+\hat{e}_2} + 2 \phi_{+,z-\hat{e}_2} \phi_{-,z}\phi_{+,z} \phi_{-,z+\hat{e}_2}\big) + O (\lambda^2),\end{split}
\end{equation}
where $b_t(k)=\tfrac{1-t^2}{|1+te^{ik}|^2}$ and $\Delta_t(k)=\tfrac{2t\sin k}{|1+te^{ik}|^2}$. Note in particular that, thanks to \cite[Eq.(4.5.4)]{AGG_AHP}, $-\tfrac{b_t(0)}{Z}
+b_{t^*_1}(0)+\tfrac{t}{Z}-t_2^*\equiv2\nu_1$ and $\tfrac{t}{Z}-t_2^*\equiv2\eta_1$. For later convenience, we denote the quadratic-in-$\phi$ first order contribution to the right hand side of \eqref{app.V1bH} by $\lambda V_2^{(1;1)}(\phi)$, and the quartic term by $\lambda V_4^{(1;1)}(\phi)$: 
\begin{equation}\label{app.V2V4}
\begin{split}
V_2^{(1;1)}(\phi) & = \sum_{z_2\ge 1}\int_{-\pi}^\pi \frac{dk_1}{2\pi} \Big[\big(-\tfrac{b_t(k_1)}{Z}+b_{t^*_1}(k_1)\big)^{(1)}\phi_{+,z_2}(-k_1)\phi_{-,z_2}(k_1)\\
 &+\eta_1^{(1)}\phi_{+,z_2}(-k_1)\phi_{-,z_2+1}(k_1)
 -\tfrac{i}{2}\big(\tfrac{\Delta_t(k_1)}{Z}-\Delta_{t_1^*}(k_1)\big)^{(1)}\sum_{\omega=\pm}\omega\phi_{\omega,z_2}(-k_1)\phi_{\omega,z_2}(k_1)\Big]\\
&+\alpha_0 \sum_{\substack{z\in\Lambda : \\ z_2\ge 2}}
\big(\phi_{+,z-\hat{e}_2} \phi_{-,z} + \phi_{+,z} \phi_{-,z+\hat{e}_2}\big),\\
V_4^{(1;1)}(\phi) & = 2\alpha_0 \sum_{\substack{z\in\Lambda : \\ z_2\ge 2}} \phi_{+,z-\hat{e}_2} \phi_{-,z}\phi_{+,z} \phi_{-,z+\hat{e}_2} ,\end{split}
\end{equation}
where $\eta_1^{(1)}= -2(\sqrt2-1+\tfrac{2(3-2\sqrt{2})}\pi)\beta_0$ is the coefficient of the first order contribution to $\eta_1$, cf. \eqref{rockyr}, and similarly for 
$\big(-\tfrac{b_t(k_1)}{Z}+b_{t^*_1}(k_1)\big)^{(1)}$ and $\big(\tfrac{\Delta_t(k_1)}{Z}-\Delta_{t_1^*}(k_1)\big)^{(1)}$, i.e., letting $b(k):=b_t(k)\big|_{t=\sqrt2-1}$, 
$b'(k):=\partial_tb_t(k)\big|_{t=\sqrt2-1}$, $\Delta(k):=\Delta_t(k)\big|_{t=\sqrt2-1}$, and $\Delta'(k):=\partial_t\Delta_t(k)\big|_{t=\sqrt2-1}$: 
\begin{equation}\begin{split} 
& \big(-\tfrac{b_t(k)}{Z}+b_{t^*_1}(k)\big)^{(1)}= Z^{(1)}\big(b(k)-b'(k)\big),\\
& \big(\tfrac{\Delta_t(k)}{Z}-\Delta_{t_1^*}(k)\big)^{(1)}= -Z^{(1)}\big(\Delta(k)-\Delta'(k)\big)\end{split}\end{equation}
where $Z^{(1)}$ was defined in \eqref{firstorders.app}, together with $\beta_1$ and $\tau_1$, and we used the fact that $2(\sqrt2-1)\beta_1-\tau_1\equiv Z^{(1)}$. 

\bigskip

In order to compute $Z_\Bs$ at first non trivial order, we start from the iterative definition \cref{Bflow}. Recalling that $Z_{\Bs,0}=Z^{-1/2}$, that $Z_{\Bs,h}=1+O(\lambda)$
uniformly in $h\le 0$, 
and that the beta function is an analytic function of $\lambda$ of order $\lambda$, we can write 
\begin{equation} 
Z_{\Bs} = Z^{-1/2}+\sum_{h\le 0}B_{h+1}[\underline{1}]+O(\lambda^2),\label{vasco.app}
\end{equation}
where $\underline 1$ is the sequence, indexed by the scale labels $h\le 0$, with elements all equal to 1. Note that $\sum_{h\le 0}B_{h+1}[\underline{1}]$ is also an analytic function 
of $\lambda$, of order $\lambda$. At first non-trivial order, it can be computed as the sum over all the trees with one endpoint of type \tikzvertex{FSSpinEP} and one additional `interaction' endpoint, i.e., of type \tikzvertex{ctVertex}, \tikzvertex{ctVertex,E}, \tikzvertex{vertex} or \tikzvertex{vertex,E}, of the local part of the corresponding tree values (with the kernel associated with the \tikzvertex{FSSpinEP} endpoint computed by replacing $Z$ by $1$ and the one associated with the interaction endpoints replaced by the first order 
contribution in $\lambda$ of the first four lines of \cref{KPsi}). By construction, the sum over all these trees of the corresponding tree values equals the first order contribution 
in `naive' perturbation theory to $\sum_{z\in\bH} W_{1,1}(z,y)$, where $y$ is an arbitrary site in $\partial\bH$, and  
$W_{1,1}(z,y)=\frac{\partial^2}{\partial\varphi_y\partial\psi_{-,z}}W_\bH(\bs \varphi,\psi)\big|_{\bs \varphi=\psi=0}$, with 
$W_\bH(\bs\varphi,\psi)$ the half-plane limit of the effective potential
\begin{equation} W_\Lambda(\bs\varphi,\psi)=\log\int P_c^*(\mathcal D\phi)P_m^*(\mathcal D\xi)e^{\mathcal V^{(1)}_\Lambda(\xi,\phi+\psi)+\mathcal B_\Lambda(\phi,\bs\varphi)},
\end{equation}
with $\psi$ a Grassmann field. 
Denoting by $\lambda B_\Bs^{(1)}$ the first order contribution to $\sum_{h\le 0}B_{h+1}[\underline{1}]$, and using the expression of $\mathcal V^{(1)}_\bH$ in \eqref{app.V1bH}
as well as the definitions in \eqref{app.V2V4}, we find: 
\begin{equation}\label{appB.10} B_\Bs^{(1)}=\sum_{z\in\bH}\frac{\partial}{\partial\psi_{-,z}}\langle \big(V_2^{(1;1)}(\phi+\psi)+V_4^{(1;1)}(\phi+\psi)\big);\phi_{-,y}\rangle_0\Big|_{\psi=0}.
\end{equation}
where $\langle\cdot;\cdot\rangle_0$ indicates the truncated expectation with respect to the Grassmann field $\phi$, with propagator 
\begin{equation} \begin{aligned} g_\bH(z,z')&=\begin{pmatrix}g_{\bH,++}(z,z') & g_{\bH,+-}(z,z') \\ g_{\bH,-+}(z,z') &g_{\bH,-+}(z,z') \end{pmatrix}=\tfrac{\sqrt2+1}{2}\int\limits_{[-\pi,\pi]^2}\frac{dk_1\, dk_2}{(2\pi)^2}\frac{e^{-ik_1(z_1-z'_1)}}{2-\cos k_1-\cos k_2} \\
 &\hskip3.8truecm \times \left[e^{-ik_2(z_2-z'_2)}  \begin{pmatrix}  -i\sin{k_1} & -1+B(k_1)e^{-ik_2} \\ 1-B(k_1)e^{ik_2} &i\sin{k_1} \end{pmatrix} \right.\\
&\hskip4.15truecm \left. -  \,e^{-ik_2(z_2+z'_2)}  \begin{pmatrix}-i\sin{k_1}&  -1+B(k_1)e^{ik_2}) \\  1-B(k_1)e^{ik_2}& \frac{1-B(k_1)e^{-ik_2}}{1-B(k_1)e^{ik_2}}i\sin{k_1}\end{pmatrix}  \right]\,,
 \end{aligned}
\end{equation}
where $B(k)=(\sqrt2-1)(\sqrt2+\cos k)$. From \eqref{appB.10}, a straightforward computation shows that, fixing $y=y_0:=(0,1)$, 
\begin{equation}\label{appB.11}\begin{split} B_\Bs^{(1)}&=
\sum_{z\in\bH}\Big[(2\nu_1^{(1)}+2\alpha_0)g_{\bH,-+}(y_0,z)-\alpha_0\delta_{z_2,1} g_{\bH,-+}(y_0,z)\\
&+2\alpha_0\Big(-g_{\bH,-+}(y_0,z-\hat e_2)\partial_{1,2}g_{\bH,-+}(z,z)+g_{\bH,-+}(y_0,z)\partial_{1,2}g_{\bH,-+}(z,z-\hat e_2)\Big)\Big]\end{split}
\end{equation}
where $2\nu_1^{(1)}=-\frac8\pi(2-\sqrt2)\beta_0$ is the first order contribution to $2\nu_1$, cf. \eqref{rockyr}, and, for $z_2=0$, $g_{\bH,-+}(y_0,z)$ should be interpreted 
as being zero. Moreover, $\partial_{1,2}g_{\bH,-+}(z,z'):=g_{\bH,-+}(z+\hat e_2,z')-g_{\bH,-+}(z,z')$ and similarly, for later reference, 
$\partial_{2,2}g_{\bH,-+}(z,z'):=g_{\bH,-+}(z,z'+\hat e_2)-g_{\bH,-+}(z,z')$. 

Recalling the fact that $2\nu_1^{(1)}=2\alpha_0\big[-1+\partial_{1,2}g_{\infty,-+}(z,z)-\partial_{1,2}g_{\infty,-+}(z,z-\hat e_2)\big]$, cf. \eqref{B.2}, and letting 
$g_{\E}(z,z')=g_\bH(z,z')-g_\infty(z,z')$, we can recast \eqref{appB.11} in the following form: 
\begin{equation}\label{appB.14}\begin{split} B_\Bs^{(1)}&=
-\alpha_0\sum_{z_1=-\infty}^\infty g_{\bH,-+}(y_0,(z_1,1))\\
&+2\alpha_0\sum_{z\in\bH}\big(-g_{\bH,-+}(y_0,z-\hat e_2)\partial_{1,2}g_{\E,-+}(z,z)+g_{\bH,-+}(y_0,z)\partial_{1,2}g_{\E,-+}(z,z-\hat e_2)\big).\end{split}
\end{equation}
Note that, from its definition, 
\begin{equation} \sum_{z_1=-\infty}^\infty g_{\bH,-+}(y_0,(z_1,z_2))=
(\sqrt2+1)\int_{-\pi}^\pi \frac{dk_2}{2\pi}\frac{\sin k_2\, \sin (k_2z_2)}{1-\cos k_2}.\end{equation}
The integral in the right hand side can be computed explicitly and is identically equal to 1 for any $z_2\ge 1$, irrespective of the choice of $z_2\ge 1$, so that 
\begin{equation} \sum_{z_1=-\infty}^\infty g_{\bH,-+}(y_0,(z_1,z_2))\equiv\sqrt2+1,\qquad  \forall z_2\ge 1.\end{equation}
Using this identity in \eqref{appB.14}, together with the remark that both $\partial_{1,2}g_{\E,-+}(z,z)$ and 
$\partial_{1,2}g_{\E,-+}(z,z-\hat e_2)$ are independent of $z_1$, gives
\begin{equation}\label{appB.15}\begin{split} B_\Bs^{(1)}&=2\alpha_0(\sqrt2+1)\Big[-\tfrac12+
\partial_{1,2}g_{\E,-+}((0,1),(0,1))\\
&+\sum_{z_2\ge 1}\big(-\partial_{1,2}g_{\E,-+}((0,z_2),(0,z_2))+\partial_{1,2}g_{\E,-+}((0,z_2),(0,z_2-1))\big)\Big].\end{split}
\end{equation}
Let us now compute explicitly the terms in brackets, by using the definition of $g_{\E,-+}(z,z')$. The second term is 
\begin{equation}\label{app.vasco2}\partial_{1,2}g_{\E,-+}((0,1),(0,1))=-
\tfrac{\sqrt2+1}{2}\int\limits_{[-\pi,\pi]^2}\frac{dk_1\, dk_2}{(2\pi)^2}\frac{e^{-2ik_2}}{2-\cos k_1-\cos k_2}  (1-B(k_1)e^{ik_2}). \end{equation}
Using \eqref{app.vasco3}, we find
\begin{equation}\label{appappvascovasco} \eqref{app.vasco2}=
\tfrac{\sqrt2+1}{2}\int_{-\pi}^\pi\frac{dk_2}{2\pi}\frac{\sqrt2(1-\cos k_2)\cos k_2+\sin^2 k_2}{\sqrt{(1-\cos k_2)(3-\cos k_2)}}
=\tfrac{(\sqrt2+1)^2}4(\sqrt2-\tfrac4\pi) \end{equation}
The summation in the second line of \eqref{appB.15} can be rewritten as
\begin{equation}\begin{split}& \sum_{z_2\ge 1}\big(-\partial_{1,2}g_{\E,-+}((0,z_2),(0,z_2))+\partial_{1,2}g_{\E,-+}((0,z_2),(0,z_2-1))\big)\\
&=-\tfrac{\sqrt2+1}{2}\sum_{z_2\ge 1}\int\limits_{[-\pi,\pi]^2}\frac{dk_1\, dk_2}{(2\pi)^2}\frac{1}{2-\cos k_1-\cos k_2} 
(1-B(k_1)e^{ik_2})e^{-2ik_2z_2}|e^{-ik_2}-1|^2\end{split}\label{questa.app}\end{equation}
In the right hand side we can bring the summation over $z_2$ under the integral sign and replace $\sum_{z_2\ge 1}e^{-2ik_2z_2}$ by $\tfrac{e^{-2ik_2}}{1-e^{-2ik_2}}$
in weak sense, so that, using again \eqref{app.vasco3}, 
\begin{equation} \eqref{questa.app}=\tfrac{\sqrt2+1}2 \int_{-\pi}^\pi \frac{dk_2}{2\pi}\frac{\sin^2 k_2}{(1+\cos k_2)\sqrt{(1-\cos k_2)(3-\cos k_2)}}=\tfrac{\sqrt2+1}4. \end{equation}
Plugging this and \eqref{appappvascovasco} back into \eqref{appB.15} gives (recalling that $\alpha_0=2(\sqrt2-1)^2\beta_0$)
\begin{equation} B_\Bs^{(1)}=4(\sqrt2-1)\beta_0\big(-\tfrac12+\tfrac{(\sqrt2+1)^2}4(\sqrt2-\tfrac4\pi)+\tfrac{\sqrt2+1}4\big)=\beta_0(5-\sqrt2-\tfrac{4(\sqrt2+1)}{\pi}).\end{equation}
Finally, substituting into \eqref{vasco.app} where $Z^{-1/2}=1-\tfrac{Z^{(1)}}2\lambda+O(\lambda^2)$, with $Z^{(1)}=2\sqrt2(\sqrt2-1)(1-\tfrac2\pi)\beta_0$, gives
\begin{equation} Z_{\Bs}=1+3\beta_0(1-\tfrac{2\sqrt2}{\pi})\lambda+O(\lambda^2),\end{equation}
which shows that $Z_\Bs$ is a non-trivial analytic function of $\lambda$.

\subsection*{Acknowledgements}
This work has been supported by the European Research Council (ERC) under the European Union's Horizon 2020 research and innovation programme (ERC CoG UniCoSM, grant agreement No.\ 724939 for all three authors and also ERC StG MaMBoQ, grant agreement No.\ 802901 for R.L.G.). 
This work was also partial supported by the MUR, PRIN 2022 project MaIQuFi cod. 20223J85K3 (for A.G.), by  the MIUR Excellence Department Project MatMod@TOV awarded to the Department of Mathematics, University of Rome Tor Vergata, CUP E83C23000330006 (for R.L.G.) and by the GNFM Gruppo Nazionale per la Fisica Matematica - INDAM. 

\section*{Data availability}
Data sharing is not applicable to this article as no new data were created or
analyzed in this study.

\section*{Conflict of interest}
The authors are not aware of any conflicts of interest relating to this article.

\begin{refcontext}[sorting=anyt]
	\printbibliography[heading=bibintoc]
\end{refcontext}

\end{document}